\documentclass{report}
\usepackage{geometry}%
\usepackage[utf8]{inputenc}

\usepackage{clrscode3e}
\usepackage{amsmath}
\usepackage{amssymb}

\usepackage{tikz}
\usepackage{subfigure}
\usepackage{float}
\usepackage{tikzpagenodes}
\usepackage{graphicx}

\usepackage{tabularx}

\usepackage{amsthm}
\usepackage{hyperref}
\usepackage{cleveref}

\usepackage{nomencl}
\makenomenclature

\newtheorem{theorem}{Theorem}[chapter]
\newtheorem{lemma}{Lemma}[chapter]
\newtheorem{proposition}{Proposition}[chapter]
\newtheorem{corollary}{Corollary}[chapter]
\newtheorem{algorithm}{Algorithm}[chapter]
\newtheorem{procedure}{Procedure}[chapter]
\newtheorem{definition}{Definition}[chapter]
\newtheorem{observation}{Observation}[chapter]
\newtheorem{conjecture}{Conjecture}[chapter]

\newcounter{ch}
\newcounter{n}
\newcounter{c}
\newcounter{d}

\begin{document}

\pagenumbering{roman}
\chapter*{\centering Burning geometric graphs\\}

\large{\centering \textsc{Arya Tanmay Gupta}, 201861003\\~\\}
\large{\centering Dissertation\\For partial fulfilment of degree\\~\\}
\large{\centering Master of Technology in Computer Science and Engineering\\~\\~\\}
\large{\centering July 2020\\~\\~\\}
\large{\centering Supervisor:\\}
\large{\centering Dr. Swapnil A. Lokhande, Assistant Professor\\~\\}
\large{\centering Department of Computer Science and Engineering\\~\\}
\begin{figure}[h]
    \centering
    \includegraphics[width=100pt]{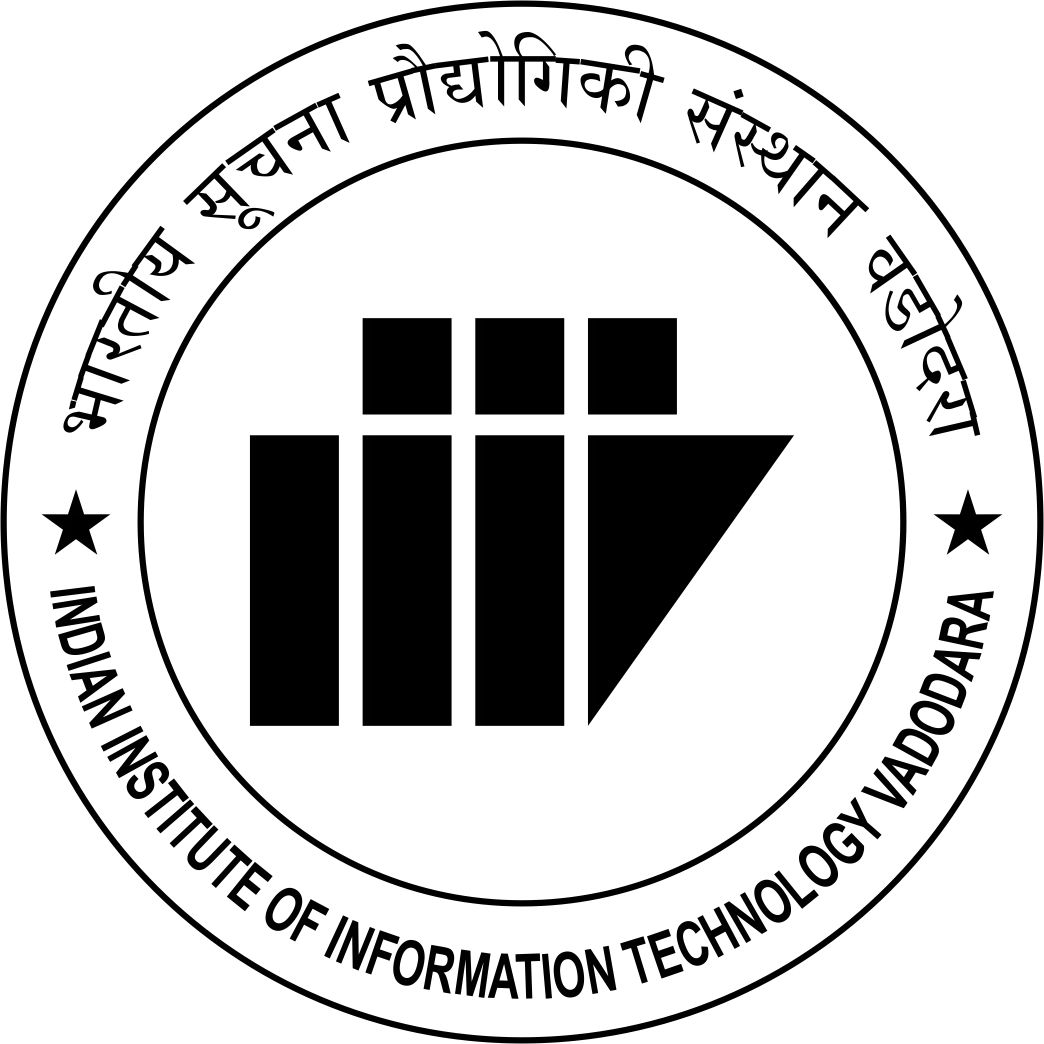}\\~\\~\\
\end{figure}
{\Large{\centering Indian Institute of Information Technology Vadodara, India\\}}
\thispagestyle{empty}
\chapter*{}

\begin{center}
    {\Large\textbf{Indian Institute of Information Technology Vadodara}}\\~\\~\\
    \textbf{\textit{\underline{CERTIFICATE}}}\\~\\
\end{center}

This is to certify that \textsc{Mr. Arya Tanmay Gupta}, Institute ID. \textit{201861003}, was a student of Master of Technology in Computer Science and Engineering at the Institute during 2018-2020. He has satisfactorily completed a dissertation on ``Burning geometric graphs'' during the year 2019-2020.

\begin{table}[h]
    \centering
    \large
    \begin{tabular}{c c}
         & \\ & \\ & \\ & \\
         & Supervisor \\
         & (Dr. Swapnil A. Lokhande)\\
         & \\
         & \\
         & \\
         & \\
        PIC/Dean Academics & Head of the Department \\
        (Dr. Pratikkumar Shah) & (Dr. Jignesh S. Bhatt)\\
         & \\
         & \\
         & \\
         & \\
         & \\
         & INSTITUTE SEAL
    \end{tabular}
\end{table}
\chapter*{Acknowledgements}
\chaptermark{Acknowledgements}

I thank all the persons who share credit in the completion of the research work, presented in this dissertation, at the Indian Institute of Information technology Vadodara (IIITV).

I thank \textbf{\textit{Dr. Swapnil A. Lokhande}} (Assistant Professor, IIIT Vadodara, India) to supervise my work at IIITV and to give useful comments in order to improve my work and presentation.

I thank \textbf{\textit{Dr. Kaushik Mondal}} (Assistant Professor, Indian Institute of Technology Ropar, India), who unofficially supervised my study. I thank him for his continuous support, help, and guidance, which was given to me not only to help me in my study, but also to help me during general life scenarios in various aspects, also from the perspective of my research career. His sincere efforts reduced the labour that I would have to input otherwise. I feel honoured to work under the supervision of a personality such as himself.

I must thank \textbf{\textit{Prof. S. K. Patra}} (Director, IIITV, India) and other members of administration to support my research at the Indian Institute of Information Technology Vadodara.

I thank \textbf{\textit{my students}} at IIITV and \textbf{\textit{my colleagues}} who kept me motivated all along. I thank \textbf{\textit{my parents}}, and almighty god.\\~\\

\begin{flushright}
\textsc{Arya Tanmay Gupta}
\end{flushright}

\tableofcontents
\addcontentsline{toc}{chapter}{List of Symbols}

\setlength{\nomlabelwidth}{2cm}

\nomenclature{$\setminus$}{setminus}
\nomenclature{$\cup_{\setminus s}$}{Left sequential union}
\nomenclature{$\cup_{s/}$}{Right sequential union}
\nomenclature{$G.Adj[X]$}{Adjacency at a distance $1$}
\nomenclature{$G.Adj_i[X]$}{Adjacency at a distance $i$}
\nomenclature{$G.N[X]$}{Neighbourhood at a distance $1$}
\nomenclature{$G.N_i[X]$}{Neighbourhood at a distance $i$}
\nomenclature{$dist(x, y)$}{Shortest distance between vertices}
\nomenclature{$G.V$}{Vertex set in graph $G$}
\nomenclature{$G.E$}{Edge set in graph $G$}
\nomenclature{$path(u, v)$}{Shortest path function}
\nomenclature{$\lceil\ \rceil$}{Ceiling}
\nomenclature{$\lfloor\ \rfloor$}{Floor}
\nomenclature{$\inf$}{Infimum}

\renewcommand{\nomname}{List of Symbols}
\printnomenclature
\addcontentsline{toc}{chapter}{\listfigurename}
\listoffigures
\listoftables
\addcontentsline{toc}{chapter}{\listtablename}
\newpage\pagenumbering{arabic}
\chapter*{\centering Abstract}
\addcontentsline{toc}{chapter}{Abstract}

A procedure called \textit{graph burning} was introduced to facilitate the modelling of spread of an alarm, a social contagion, or a social influence or emotion on graphs and networks.

Graph burning runs on discrete time-steps (or rounds). At each step $t$, first (a) an unburned vertex is burned (as a \textit{fire source}) from ``outside'', and then (b) the fire spreads to vertices adjacent to the vertices which are burned till step $t-1$. This process stops after all the vertices of $G$ have been burned. The aim is to burn all the vertices in a given graph in minimum time-steps. The least number of time-steps required to burn a graph is called its \textit{burning number}. The less the burning number is, the faster a graph can be burned.

Burning a general graph optimally is an NP-Complete problem. It has been proved that optimal burning of path forests, spider graphs, and trees with maximum degree three is NP-Complete. We study the \textit{graph burning problem} on several sub-classes of \textit{geometric graphs}.

We show that burning interval graphs (\Cref{section:burn-interval-graphs}, \Cref{theorem:BIGNPCIG}), permutation graphs (\Cref{section:burn-permutation-graphs}, \Cref{theorem:BPGNPC}) and disk graphs (\Cref{section:burn-disk-graphs}, \Cref{theorem:BDGNPC}) optimally is NP-Complete. In addition, we opine that optimal burning of general graphs (\Cref{section:no-better-than-3-approx}, \Cref{conjecture:no-better-than-3-approx}) cannot be approximated better than 3-approximation factor.
\chapter{Introduction}\label{chapter:introduction}

In this chapter, we discuss some fundamentals related to our subject problem. We discuss what are graphs, what are decision problems and languages in computing theory. We discuss about some interesting facts about the word ``algorithm'', and further when we have discussed enough background details, we describe \textit{algorithms} formally as per the current perspective. We discuss how we differ between \textit{easy} and \textit{hard} problems, along with a brief description of what NP, NP-Complete and NP-Hard problems are. We also briefly discuss the \textit{reducibility} of certain problems into one-another, and the \textit{approximability} of hard problems.

\section{Origin, etymology, history}

\textsc{Lionardo Pisano} \cite{Sigler2002}, more popularly known as \textsc{Fibonacci}, introduced the traditional Indian mathematical methods to Europe in the $13^{th}$ century. Until then, abacus was used to perform all calculations. Pisano introduced a mathematics which was more efficient: computations could be performed on numbers without bounds on their digit-length. A person who could perform computations without the use of abacus was called \textit{Ma\~estro-de-abaci}. And the Europeans started to call this new form of mathematics, which could be performed on ``paper'' without abacus, \textit{algorithms}\index{algorithms: etymology}.

Since then, numerous efforts have been made to translate human intelligence and computing ability into artificial machinery. \textsc{Blaise Pascal} \cite{Dasgupta2014} built a machine in the $17^{th}$ century which could perform addition and subtraction. \textsc{Gottfried Wilhelm Leibniz} built a machine, during the same time, which could perform multiplication and division as well. \textsc{Charles Babbage} built the famous \textit{Difference Engine}\index{difference engine} which could do similar computations ``automatically'', that is, once the input numbers are supplied to it, it was able to do the computation without any human intervention. This machine was able to prepare tables: it was able to compute polynomials of degree $2$ for consecutive integers; this was called the \textit{method of differences}. Babbage built the first prototype of this machine in 1822.

\textsc{Luigi Frederico Menabrea} explained with reference to the Difference Engine that it was limited only to one type of computations, it could not be applied to solve numerous other problems in which mathematicians might be interested. This led Charles Babbage to design the \textit{Analytical Engine}\index{analytical engine}, which could solve the full range of algebraic problems. The generality of the Analytical Engine is discussed in Menabrea’s Italian article \textit{Sketch of the Analytical Engine} (1842). It was translated into English by Augustus Ada \cite{Menabrea1843}.

\textsc{Augustus Ada}, countess of Lovelace, proposed that the Babbage's design could be used to compute function of any number of functions. On Babbage's request, she wrote some additional notes to her ``memoir'', most famous one of them is the \textit{Note G}\index{note G - Augustus Ada}, in which, firstly, she anticipated an issue: whether computers can exhibit ``intelligence'', or, ``original thought'', and secondly, in this note she wrote a sequence of operations (an \textit{algorithm}) to compute Bernoulli numbers on the Analytical Engine.

After some decades, \textsc{Alan Mathison Turing} worked on construction of formal languages for any function (or a decision problem, as he presents in \cite{Turing1937}). He initiated the design of what we call the Turing Machine \index{turing machine: invention} which works on these formal languages to compute for any decision problem. We discuss decision problems and formal languages in this chapter; we do not touch the Turing Machine, the reader is advised to refer \cite{Turing1937,Garey1979} to study the Turing Machine in detail. We start with a brief discussion on graphs, on which the following chapters are majorly based.

\section{Graphs}

A \textit{graph}\index{graphs} is a representation of entities and their relations: generally, a graph tells which entities are related (unweighted graph); sometimes the relations may have some associated cost or weightage (weighted graphs). Formally, a graph is a mathematical object which represents entities as vertices, and edges as relations between those vertices: if two entities are related, then there will be an edge between the corresponding vertices in the graph. Here, we only discuss relations which are symmetric (if $a$ is related to $b$, then $b$ is related to $a$). So the edges are bidirectional; we do not show any directions for brevity.

Any cost or weightage related to a relation between a pair of entities is presented as weights on the edges. For example, if two computers are connected in a network, we can denote the frequency of communication between them in that network as a weight on the edge between their corresponding vertices in the graph representation of that network. Throughout the following chapters, we assume that all the edges have same weight; we consider the weight on all edges to be $1$, which we do not show explicitly for brevity; the value of the weight is to represent merely that a given pair of vertices are connected; there are only two weights associated with all the pairs of vertices which, $0$ or $1$, if otherwise the weight associated with some pair of vertices $0$, then we assume that they are not connected. If the vertices $a$ and $b$ are connected, then the weight associated to the pair $(a,b)$ is $1$, and if the vertices $a$ and $b$ are connected, then the weight associated to the pair $(a,b)$ is $0$. A rough example of an unweighted graph is presented in \Cref{figure:unweighted-graph}.

\begin{figure}
    \begin{minipage}{1\textwidth}
        \centering
        \begin{tikzpicture}[scale=.41]
		    \node[circle, draw=black, fill=black, inner sep=1pt, minimum size=3pt, label=left:{$p$}] (AA) at (-8,0) {};
		    \node[circle, draw=black, fill=black, inner sep=1pt, minimum size=3pt, label=right:{$q$}] (AB) at (-5,0) {};
	    	\node[circle, draw=black, fill=black, inner sep=1pt, minimum size=3pt, label=left:{$r$}] (AC) at (-8,-3) {};
	    	\node[circle, draw=black, fill=black, inner sep=1pt, minimum size=3pt, label=left:{$s$}] (AD) at (-5,-3) {};
			\node[circle, draw=black, fill=black, inner sep=1pt, minimum size=3pt, label=right:{$t$}] (AE) at (-3,-2) {};
		    \node[circle, draw=black, fill=black, inner sep=1pt, minimum size=3pt, label=left:{$u$}] (AF) at (-5,-6) {};
	    	\node[circle, draw=black, fill=black, inner sep=1pt, minimum size=3pt, label=right:{$v$}] (AG) at (-3,-5) {};
	    	\node[circle, draw=black, fill=black, inner sep=1pt, minimum size=3pt, label=right:{$w$}] (AH) at (-3,-7.5) {};
		    
		   	\draw (AA)--(AB);
		   	\draw (AA)--(AC);
		   	\draw (AB)--(AC);
		   	\draw (AB)--(AD);
		   	\draw (AB)--(AE);
		   	\draw (AE)--(AG);
		   	\draw (AD)--(AF);
	        \draw (AD)--(AG);
	        \draw (AG)--(AH);
		\end{tikzpicture}
    \end{minipage}
    \caption{An example of unweighted graph. Entities are represented as vertices $p, q, r, . . ., w$ and there is an edge between a pair of vertices if the corresponding edges are related as per some relation function. The weight associated to the pair $(s,u)$, for example, is $1$, and the weight associated to the pair $(u,w)$ is $0$.}
    \label{figure:unweighted-graph}
\end{figure}
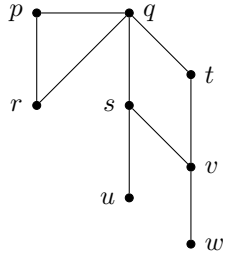

There are various problems which are related to graphs, most of them are computationally ``hard'' to solve on arbitrary graph inputs. In the following paragraphs, we discuss the description of the nature of problems in general, especially the problems which are hard. We shall keep all the descriptions close to the perspective of graphs: our main focus is on the problems which are related to graphs.

\section{Decision Problems}

Refer to the problems described under \Cref{section:problems-in-NP} in \Cref{chapter:list-of-definitions}. The problems like the graph isomorphism problem, whose solution is either ``yes'' or ``no'' are called \textbf{\textit{decision problems}}\index{decision problems}. The distinct 3-partition problem is also a decision problem. Other problems such as the largest clique problem, minimum dominating set problem, largest independent set problem, minimum vertex cover problem, graph coloring problem are optimization problems; they can also be converted into their respective equivalent decision problem versions. Also, we can add a mathematical bound $B$ as an additional parameter to a decision problem and and reformulate it. For example we can ask that given a graph $G$ with degree bound $B$, does there exist a clique of size at least $k$.

The optimization problems are at least as hard as decision problems problems. For example, if we can compute the largest clique in $G$, we can also compute if there exists a clique in $G$ whose size is at least $k$. Hardness of many decision problems is closely tied to their corresponding optimization versions \cite{Garey1979}. For example, decision version of the clique problem is no easier than the optimization version of the problem. Likewise, we can transform any problem into its corresponding decision version.

\section{Languages}

Let that $x$ is a sequence of symbols such that a particular decision problem $A$ returns ``yes'' as output. The \textbf{\textit{language}}\index{languages} $L$ of a decision problem $A$ is the set of sequences of symbols for such that for each sequence in that set, $A$ returns ``yes'' as output, given the alphabet, and an encoding scheme. For example, let the alphabet be $\Lambda = \{$`$0$', `$1$', `$($', `$)$', `$,$'$\}$. Let a graph $G$ be represented by an input string $x=\{00$, $01$, $10$, $11$, $(00,01)$, $(00,10)$, $(00,11)$, $(01,10)$, $(01,11)$, $(10,11)\}$ constructed from $\Lambda$. In $x$, $00,01,10$ and $11$ are the vertices and, for example, $(00,01)$ represents that there is an edge between vertices $00$ and $01$. Let $e$ be the encoding scheme, for example, used to encode $G$ into $x$. Similarly, we can represent any graph using the alphabet $\Lambda$ and the encoding $e$. Observe that in $e$, there is no unnecessary padding of symbols. All such encoding schemes which do not allow any unnecessary padding of symbols can represent any object (for example, $G$) with sequences whose lengths are polynomially bound to one another \cite{Garey1979}.

Let there be a problem $A$ as follows, given a graph as input, the task is to find if there is a clique of size at least $4$ in that graph. Observe that $x$ represents a graph which is a clique of size $4$, as presented in \Cref{figure:clique4}. $x$ is an element of $L_A$, the language of $A$ under the encoding scheme $e$. $L_A$ contains all the possible sequences from the alphabet $\Lambda$ (under the encoding scheme $e$) for which $A$ returns "yes". $A$ and $L_A$, for example, can be used interchangeably, and similarly, any decision problem with its corresponding language because they are computationally equivalent.

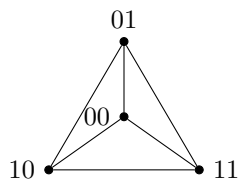
\begin{figure}
    \begin{minipage}{1\textwidth}
        \centering
        \begin{tikzpicture}
            \node [circle, draw=black, fill=black, inner sep=1pt, minimum size=3pt, label=left:{$00$}] (A) at (0,0) {};
            \node [circle, draw=black, fill=black, inner sep=1pt, minimum size=3pt, label=above:{$01$}] (B) at (0,1) {};
            \node [circle, draw=black, fill=black, inner sep=1pt, minimum size=3pt, label=left:{$10$}] (C) at (-1,-.707) {};
            \node [circle, draw=black, fill=black, inner sep=1pt, minimum size=3pt, label=right:{$11$}] (D) at (1,-.707) {};
            
            \draw (A) -- (B); \draw (A) -- (C); \draw (A) -- (D);
            \draw (B) -- (C); \draw (B) -- (D);
            \draw (C) -- (D);
        \end{tikzpicture}
    \end{minipage}
    \caption{A sample graph presented by a string $x=\{00$, $01$, $10$, $11$, $(00,01)$, $(00,10)$, $(00,11)$, $(01,10)$, $(01,11)$, $(10,11)\}$, where $00,01,10$ and $11$ are the vertices and, for example, $(00,01)$ represents that there is an edge between vertices $00$ and $01$. This graph is a clique of size 4.}
    \label{figure:clique4}
\end{figure}

\section{Algorithms}

Modern definition of the word algorithm is as follows. An \textbf{\textit{algorithm}}\index{algorithms: definition} is a step-by-step procedure used to solve a decision problem given that it halts in finite time given any input which may be from the language of that problem or not. \cite{Garey1979}

\section{Easy and Hard problems}\label{section:problems-classes}

If a problem $A$ can be reduced into another problem $B$ in polynomial time (with respect to the input length), it means that the problem $B$ is computationally at least as hard as the problem $A$ \cite{Garey1979}. There are several classes of problems depending on solvability, reducibility and computational hardness. Some of them are described in the following paragraphs. \index{problems solvability classification}

Problems in class \textbf{\textit{P}} can be solved in polynomial time if the host machine is allowed to execute only polynomial amount of instructions in one time unit. Problems in \textbf{\textit{NP}} class can be solved in polynomial time given that the host machine can execute arbitrarily any amount of instructions in one time unit. From here, it is clear that problems in class P can also be solved in polynomial time if the host machine is allowed to do arbitrarily any amount of instructions in one time unit. Hence it is conjectured that P $\subset$ NP.\index{P, NP} For example the problems with solution in $n^2$, $n^{100}$ or even $10^{9900}n^{10^{99}}$ come under P class. \textit{Exponential time algorithms} have, for example, time complexity functions like $3^n$, $n^n$, $n^{\sqrt{n}}$ or even $n^{\log n}$.

The problems in NP class can be verified in polynomial time. Here, verification means that given an instance $I$ for a problem $A$ and a structure $C$ in $I$, it is to be verified if $C$ fulfils the constraints of $A$. If $C$ passes the verification, it means that $A$ will return ``yes'' for $I$ as an input. For example, if we have a set of vertices $C$ and a graph $G$ as an instance input, we can easily verify if $C \in G.V$ is forming a clique in $G$ or not. In this way, we can assert that, given $C$ and $G$, if $C$ is a clique, then there exists a clique of size at least $|C|$ in $G$.

Problems in class P can be solved by deterministic algorithms in polynomial time. NP class of problems \textit{are} solved by a nondeterministic algorithm in polynomial time, which at each step, arbitrarily selects a structure from the instance and checks deterministically, in polynomial time, if that structure satisfies the constraints of the given problem. Here also, the conjecture that P $\subset$ NP follows. The arbitrary selection of a structure from the input instance (may also be called ``guessing'') by a nondeterministic algorithm is supposed to be computed in constant time, $O(1)$. The verification is always done detreministically; so the time complexity of any nondeterministic algorithm is always equal to the time taken to verify an arbitrary structure. However, the nondeterministic algorithm can, in practical, keep on guessing structures indefinitely and never terminate.

\textbf{\textit{NP-Complete}}\index{NP-Complete} is a class of problems to which each problem in NP can be reduced into in polynomial time. The problems which come in NP-Complete class are also reducible to each other in polynomial time. NP-Complete problems are are the hardest problems of NP class.

\textbf{\textit{NP-Hard}}\index{NP-hard} problems are the problems that are at least as hard as any problem in NP; they may or may not be in NP. So the NP-Complete problems are a subset of NP-Hard problems. In fact, considering the hardness, the NP-Complete problems come in the intersection between NP and NP-Hard problems, as shown in \Cref{figure:problems-classes}.

These classifications arise because we are still not able to determine whether the problems in $NP$ can be solved in polynomial time only, using some algorithm, that is, we still do not have a mathematical proof as to whether $P=NP$ or not. These classifications are based on the conjecture that $P\neq NP$. There are numerous other classes of problems, which hierarchically allow problems of more complex bounds of time complexity; apart from this, problems can also be classified on the basis of the ``extra'' space they require for computation. We do not discuss those classes of problems; such problems are discussed in detail in \cite{Garey1979} and other works.

\begin{figure}
    \begin{minipage}{1\textwidth}
        \centering
        \begin{tikzpicture}
            \draw (0,0) circle (1.5);
            \draw (0,-.75) circle (.5);
            \draw[thick] (0,.5) parabola (2,3);
            \draw[thick] (0,.5) parabola (-2,3);
            \draw[<-] (.5,2) parabola (1.2,1.5);
            \draw[<-] (-.5,2) parabola (-.8,1);
            
            \node at (0,0) {NP};
            \node at (0,-.75) {P};
            \node at (0, 1) {NPC};
            \node at (0, 2) {\textbf{NPH}};
        \end{tikzpicture}
    \end{minipage}
    \caption{Classification of algorithms / problems based on runtime complexity.}
    \label{figure:problems-classes}
\end{figure}
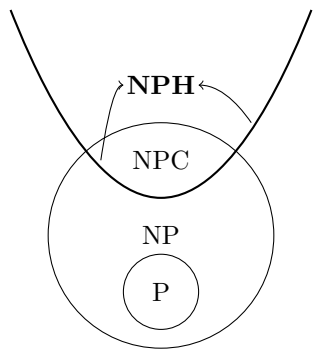

\section{Strong NP-Completeness}\label{section:strong-NPC}

Let there be a problem $B$ and $n$ be the length of an arbitrary input $x$ to $B$. A problem $B$ is NP-Complete (or NP-Hard) in the \textit{strong sense}\index{strong NP-Completeness} if it remains NP-Complete (or NP-Hard) even when its parameters are bounded by a polynomial $p$ of $n$.

To prove that a problem $B$ is NP-Complete (or NP-Hard) in the strong sense, we need to show \cite{Garey1979} that for some polynomial $p$, $B_p$ ($B$ constrained by $p$ of $n$) is NP-Complete (or NP-Hard).

\subsection{Brute force}

A \textit{brute force}\index{brute force algorithm} algorithm is an algorithm which tries all possibilities and then compares the output of each possibility to produce one possibility as an optimal result. For example, considering the rod cutting problem (see the definition in \Cref{section:problems-in-NP}), an algorithm which uses brute force to compute the optimal cuts on the rod to produce maximum profit, has time complexity exponential in the length of rod. We can rather use a dynamic programming approach, on the other hand, to solve any arbitrary rod cutting instance optimally in time quadratic in the length of the rod \cite{Cormen}.

Still, there are numerous problems which do not have a solution algorithm (yet) which runs in polynomial time. Some of the popular examples of such problems are finding the largest clique, coloring with minimum colors, finding maximum independent set, finding minimum vertex cover in an arbitrary graph. Such problems are yet NP-Hard because we have to try and search on every possibility. One of the reasons that these problems have no solution algorithm which gets executed in polynomial time because no overlapping subproblems have been defined (so far) for general graphs so that we could use a common dynamic programming approach and reduce the time complexity.

A \textit{pseudo-polynomial algorithm} is defined for number problems (we discuss examples of number problems shortly). A pseudo- polynomial algorithm runs in time polynomial in the value of the the input, rather than the input length. We generally use dynamic programming approach to design a pseudo-polynomial time algorithm. \Cref{obsertation:number-problem-polynomially-solvable} is stated in \cite{Garey1979}.

\begin{observation}\label{obsertation:number-problem-polynomially-solvable}
    If a problem $B$ is NP-Complete and $B$ is not a number problem, then $B$ cannot be solved by a pseudo-polynomial algorithm unless P $=$ NP.
\end{observation}

The problems like computation of a largest clique is an NP-Complete problem, and since it is not a number problem, a pseudo-polynomial time algorithm cannot be designed for it. On the other hand, the rod-cutting problem (see definition in \Cref{section:problems-in-NP}) is a number problem, but the length of the rod is always polynomial in the length of the input. So we have that it is polynomially solvable by the dynamic programming approach; we do not call the dynamic programming algorithm which we use to solve it a pseudo-polynomial time algorithm.

The are certain number problems which are NP-Complete (or NP-Hard) in the strong sense. These problems remain NP-Complete (or NP-Hard) even when we put a bound polynomial in the length of the input on its parameters and values in the input. For example, the distinct 3-partition problem is a number problem which is NP-Complete in the strong sense. If we try to bound-above each element in it by a polynomial in the length of the input, it still remains NP-Complete.

\subsection{Weak sense}

A problem is NP-Complete (or NP-Hard) in the \textit{weak sense}\index{weak sense NP-Completeness} if there is a solution of that problem which is polynomial in the \textit{magnitude} of the input value(s), given that its parameters are bounded above by the length of the input. Observe that if the values in the input are bounded by the polynomial in the length of the input and we obtain a solution which is polynomial in the \textit{magnitude} of the input value(s), then it also means that the solution is polynomial in the length of the input. Here, magnitude corresponds to the value of the input; for example, in the knapsack problem (see definition in \Cref{section:problems-in-NP}), if we bound the weight-capacity of the knapsack by a polynomial in the input length, we can obtain a pseudo-polynomial time algorithm \cite{Cormen} whose running time is a polynomial function of the input length. The complexity of the dynamic programming-based algorithm becomes $O(nk)$ where $n$ is the number of objects and $k$ is the weight-capacity of the knapsack. This solution is not necessarily a polynomial time solution because $k$ is not necessarily bounded polynomially by the size of the input (unlike $n$).

\section{Turing reduction}\label{section:ptr-pptr-npc}

One way of showing that a problem $B^\prime$ is computationally at least as hard as the problem $B$ is through Turing reduction. Let there be an arbitrary instance $I_B$ of $B$ (present in a set of instances which define a language $L$ of $B$, under an encoding scheme $e$). If for every instance $I_B$ in $B$, $I_B$ can be reduced in polynomial time to an instance $I_{B^\prime}$ of another problem $B^\prime$, we say that $B$ is \textit{Turing reducible} to $B^\prime$. It also implies that $B^\prime$ is at least as hard as $B$. Here, the languages and the underlying encoding schemes in which $I_B$ and $I_{B^\prime}$ are represented may as well be independent to each other (given that they do not accept any unnecessary padding).

Let that all the instances $I_B$ of a problem $B$ can be reduced to an instance $I_{B^\prime}$ of another problem $B^\prime$ by a one-to-one \textit{polynomial time reduction function}\index{polynomial time reduction function} $f$ ($\implies I_{B^\prime}=f(I_B)$). The characteristics of $f$ \cite{Garey1979} are as follows.
\begin{enumerate}
    \item $f$ can be computed deterministically in time polynomial in $|I_B|$.
    \item $B$ returns ``yes'' for $I_B$ as an input if and only if $B^\prime$ returns ``yes'' for $I_{B^\prime}$ as in input.
\end{enumerate}

\begin{lemma}\label{lemma:ptr-npc}
Let that $B$ is an NP-Complete problem and $B^\prime$ is in NP. Now if $B$ is \textit{Turing reducible} to $B^\prime$, we have that $B^\prime$ is also an NP-Complete problem.\cite{Garey1979}
\end{lemma}

If $B$, for example, is a number problem, then to prove that some problem is at least as hard as $B$, we can reduce $B$ using a pseudo polynomial time reduction function. Let that each instance $I_B$ of a problem $B$ can be reduced to an instance $I_{B^\prime}$ of another problem $B^\prime$ by a one-to-one \textit{pseudo-polynomial time reduction function}\index{pseudo polynomial time reduction function} $f_p$ ($\implies I_{B^\prime}=f_p(I_B)$). The characteristics of $f_p$ are as follows \cite{Garey1979}. Mark that for a pseudo-polynomial time reduction, $B$ has has to be a number problem.
\begin{enumerate}
    \item $B$ returns ``yes'' for an instance $I_B$ if and only if $B^\prime$ returns ``yes'' for an instance $f_p(I_B)$.
    \item $f_p$ can be computed in time polynomial in two variables $n=|I_B|$ and $m=\max(I_B)$.
    \item There exists a single-variable polynomial $q_1$ such that, for every instance $I_B$ for which $B$ returns ``yes'',
    $$q_1(|f_p(I_B)|)\geq |I_B|=n$$
    \item There exists a two-variable polynomial $q_2$ such that
    $$\max(f_p(I_B))\leq q_2(\max(I_B),|I_B|)=q_2(n,m)$$
\end{enumerate}

\begin{lemma}\label{lemma:pptr-npc}
Let that $B$ is an NP-Complete problem in the strong sense and $B^\prime$ is in NP. Now if $B$ is pseudo-polynomial time reducible to $B^\prime$, we have that $B^\prime$ is also an NP-Complete problem in the strong sense.\cite{Garey1979}
\end{lemma}

\section{Approximation, approximability and inapproximability}\label{section:inapproximability-NPH}

When a problem is NP-Complete, we theorize that (assuming the conjecture that P $\neq$ NP) we cannot solve the problem optimally in polynomial time. This arises the requirement of \textit{approximation algorithms}\index{approximation algorithms}: algorithms which can take us close enough to the optimal solution of a given problem; we generally call it ``acceptable'' solution. It computes a solution with a cost which is close enough to the optimal solution We generally use that solution in practical applications. Let that an optimization problem $P$ is NP-Complete, an algorithm $O_P$ which is optimally able to solve it (assume, in exponential time), and an algorithm $A_P$ which is an approximation algorithm to solve $P$. Let $R_A$ be the approximation ratio guaranteed by $A_P$ to solve $P$.

Let that $x$ be an arbitrary input to $P$. If $P$ is a minimization problem, we have that the approximation ratio $$R_A=\frac{A_P(x)}{O_P(x)}.$$
If otherwise $P$ is a maximization problem, we have that the approximation ratio $$R_A=\frac{O_P(x)}{A_P(x)}.$$
These (mathematical or intuition-based) guarantees are computed to hold for any arbitrary input.

The nature of \textit{approximability} sometimes changes with the cost of the optimal solution. For some problems, we have that if the cost of the optimal solution is more than a given arbitrary positive integer $N$, we can guarantee a different (generally, a better) approximation ratio, denoted as $R_A^\infty$.
\begin{center}
    $R_A^\infty=\inf\{r\geq 1:R_A(x)\leq r\ \forall\ x$ such that $O_P(x)\geq N\}$
\end{center}

Following the conjecture that P $\neq$ NP, we have that if a problem is NP-Complete, we cannot go on constructing approximation algorithms close to a ratio $1$ to the optimum. If P $\neq$ NP, then there must be a limit to the \textit{approximability} of an NP-Hard problem. The approximability directly depends on the problem itself. We have \Cref{theorem:approximability-general} \cite{Garey1979} stating a property regarding the design of approximation algorithms in general.

\begin{theorem}\label{theorem:approximability-general}
    If the solution for an NP-Hard problem $P$ has cost $k\in\mathbb{N}$, then no approximation algorithm $A_P$ can guarantee that the approximation ratio $R_A<1+(1/k)$, and $P$ cannot be solved by a polynomial time approximation scheme, given that P $\neq$ NP.
\end{theorem}

\section{Main objectives}

Graph burning has been recently introduced and has been identified as an NP-Complete problem. Our aim is to study graph burning on interval graphs, permutation graphs, and disk graphs and determine if graph burning can be solved on these graph classes in polynomial time. We have found that that burning of these graph classes is NP-Complete.

\section{Organization of the chapters}\label{section:organization-of-chapters}

\Cref{chapter:list-of-definitions} includes definitions, along with some basic theory, on some graph classes and problems in NP. It also includes elaborated definitions of some symbols (in \Cref{section:definitions-for-symbols}) along with definitions of some complexity notations (in \Cref{section:complexity-notation-functions}) that are used commonly in algorithms' texts. \Cref{chapter:literature} contains the results already present in the literature, which are related to this theory that we present in the following chapters. 

In \Cref{chapter:graph-burning}, we introduce \textit{graph burning}: we describe what the problem is, along with descriptive examples for better understanding. We also discuss some problems and games which were discovered earlier than \textit{graph burning}, but are closely related to it. We also look at some other works which have described some interesting applications related to \textit{graph burning}.

In \Cref{chapter:graph-burning-examples}, we discuss some more general and mathematically sound examples, and show optimal \textit{burning} procedures on several graph classes. We also describe an algorithm which can be used to \textit{burn} general graphs.

\Cref{chapter:other-games-and-problems} describes some other games and problems. We discuss the \textit{distinct 3-partition problem}; we utilize it in later chapters in deriving some useful proofs towards NP-Completeness of \textit{burning} several graph classes. We discuss the \textit{firefighter problem} which we later see (towards the conclusion, \Cref{chapter:conclusion}) that it can be utilized in controlling the spread of \textit{fire} throughout a graph along with some useful examples that may lead to good research developments.

In \Cref{chapter:why-hard}, we describe why optimal \textit{burning} of general graphs is computationally hard. We show that \textit{burning} several classes of graphs is NP-Complete. This is the chapter where we include some of our original findings that \textit{burning} certain subclasses of \textit{geometric graphs} is NP-Complete. On the other hand, in \Cref{chapter:where-easy}, we describe a few graph classes on which optimal \textit{burning} can be done in polynomial time.

In \Cref{chapter:approximation}, we describe a 3-approximation algorithm which can be used to derive a \textit{burning sequence} for an arbitrary graph in polynomial time. We also discuss how much we can get close to the \textit{burning number} in polynomial time while computing a \textit{burning sequence}.

We conclude in \Cref{chapter:conclusion} with some obvious, but interesting observations, along with the description of some prospective research opportunities related to the subject which we find useful and interesting.

\chapter{Preliminaries}\label{chapter:list-of-definitions}


\section{Complexity notation functions}\label{section:complexity-notation-functions}

The following functions \cite{Cormen} are used to denote the complexity \index{complexity notations} of algorithms in terms of the input size $n$. The exact runtime complexity of an algorithm is returned by $g(n)$, a function of $n$.\\

{\boldmath$\Theta$}

$\Theta(g(n)) = \{f(n) : \exists\ c_1>0, c_2>0$ and $n_0>0$ such that \[0 \leq c_1\ g(n) \leq f(n)\leq c_2\ g(n)\ \forall\ n\geq n_0\}
\]

{\boldmath$O$}

$O(g(n)) = \{f(n) : \exists\ c>0$ and $n_0>0$ such that \[0 \leq f(n)\leq c\ g(n)\ \forall\ n\geq n_0\}
\]

{\boldmath$\Omega$}

$\Omega(g(n)) = \{f(n) : \exists\ c>0$ and $n_0>0$ such that \[0 \leq c\ g(n) \leq f(n)\ \forall\ n>n_0\}
\]

{\boldmath$o$}

$o(g(n)) = \{f(n) : \forall\ c>0\ \exists\ n_0>0$ such that \[0 \leq f(n) < c\ g(n)\ \forall\ n\geq n_0\}
\]

{\boldmath$\omega$}

$\omega(g(n)) = \{f(n) : \forall\ c>0\ \exists\ n_0>0$ such that \[0 \leq c\ g(n) < f(n)\ \forall\ n\geq n_0\}
\]

\section{Definitions for symbols}\label{section:definitions-for-symbols}

Referring from the list of symbols.\\

\textbf{Left sequential union}\index{left sequential union}: If $P = (a,b)$, then after executing the statement $P = P \cup_{\setminus s} (c)$, $P$ becomes $(c, a, b)$. This operation can add a single element to a sequence, or merge two sequences.\\

\textbf{Right sequential union}\index{right sequential union}: If $P = (b,c)$, then after executing the statement $P = P \cup_{s/} (a)$, $P$ becomes $(b, c, a)$. This operation can add a single element to a sequence, or merge two sequences.\\

\textbf{Infimum}: The infimum of a subset $X$ of a set $X^\prime$ is the largest element of $X^\prime$ which is less than or equal to all the elements in $X$.\\

\textbf{Ceiling}: Let $x$ be a real number, then $i = \lceil x\rceil$ is the smallest integer such that $i \geq x$.\\

\textbf{Floor}: Let $x$ be a real number, then $i = \lfloor x\rfloor$ is the biggest integer such that $i \leq x$.\\

\textbf{Setminus}\index{setminus}: It removes the elements from the set preceding the operation symbol which are common to the set succeeding it. If $A = \{a, b, c\}$ and $B = \{b, c\}$ are two sets, then $A \setminus B = \{c\}$. For the sake of another example if $A = \{a, b, c\}$ and $B = \{c, d, e\}$ are two sets, then $A \setminus B = \{a, b\}$.\\

\textbf{Shortest distance between vertices}\index{shortest distance}: This statement returns the number of edges in a shortest path between two vertices $x$ and $y$.\\

\textbf{Adjacency at a distance {\boldmath$1$}}: This statement returns the set of vertices that are adjacent to $X$, excluding $X$. $X$ can be a single vertex, a set of vertices, or a subgraph of $G$.\\

\textbf{Adjacency at a distance {\boldmath$i$}}\index{adjacency}: This statement returns the set of vertices that are atmost at a distance $i$ from $X$, excluding $X$. $X$ can be a single vertex, a set of vertices, or a subgraph of $G$.\\

\textbf{Edge set in graph {\boldmath$G$}}\index{edge set in a graph}. This keyword acts as a variable which denotes the edge set in graph $G$.\\

\textbf{Neighbourhood at a distance {\boldmath$1$}}: This statement returns the set of vertices that are adjacent to $X$, including $X$. $X$ can be a single vertex, a set of vertices, or a subgraph of $G$.\\

\textbf{Neighbourhood at a distance {\boldmath$i$}}\index{neighbourhood}: This statement returns the set of vertices that are atmost at a distance $i$ from $X$, including $X$. $X$ can be a single vertex, a set of vertices, or a subgraph of $G$.\\

\textbf{Vertex set in graph {\boldmath$G$}}\index{vertex set in a graph}: This keyword acts as a variable which denotes the vertex set in graph $G$.\\

\textbf{Shortest path function}\index{shortest path}: A function that returns the shortest path $P$ from $u$ to $v$; the sequence of vertices in $P$ from $u$ to $v$, including $u$ and $v$.

\section{Problems in NP}\label{section:problems-in-NP}

The following problems are mentioned in the following chapters.\\

\textbf{\textit{Distinct 3-partition problem}}: In a \textit{distinct 3-partition problem} \index{distinct 3-partition problem}, given input is a set of positive integers, $X = \{a_1, a_2, . . ., a_{3n}\}$, and a positive integer $B$ such that $\sum_{i=1}^{3n}a_i = nB, \frac{B}{4}>a_i>\frac{B}{2}$; the task is to find if $X$ can be partitioned into $n$ sets, each containing $3$ integers, such that each set sums to $B$.\\

\textbf{\textit{Determination of a Hamiltonian cycle}}: A\textbf{\textit{Hamiltonian cycle}} \index{Hamiltonian cycle} of a graph $G$ is a path $P=(v_1,v_2,\dots,v_n,v_1)$ such that $n=|G.V|$ and for every pair of adjacent vertices $v_i$ and $v,j$ in $P$, $(v_i,v_j)\in G.E$. One approach to determine whether a hamiltonian cycle exists in a graph $G$ can be done as follows: we can check for every possible sequence of vertices $G.V$ that it satisfies the constraint or not. If there is at least one such sequence, the return value is $true$, otherwise $false$. This approach takes $O(n^n)$ time.\\

\textbf{\textit{Graph coloring problem}}\index{vertex coloring}: Given a graph $G$ and a set of infinite colors $C$, the task is to find the minimum number of colors in $C$ which can be assigned to each vertex in $G.V$, such that (1) each vertex is colored with only one color, and (2) $\forall\ a,b \in G.V$, if $(a,b) \in G.E$, then $color(a) \neq color(b)$. One solution is to color vertices sequentially, for each sequence of vertices in $G.V$. The sequence of vertices which utilizes the least amount of colors is the final solution. This approach takes $O(n^n)$ time.\\

\textbf{\textit{Graph isomorphism problem}}\index{isomorphism in graphs}: Given two graphs $G_1$ and $G_2$ such that $|G_1.V| = |G_2.V|$, the task is to determine if $\exists$ a sequence of vertices $C_1$ of $G_1.V$ and a sequence of vertices $C_2$ of $G_2.V$ such that $\forall\ 1\leq i,j\leq |G_1.V|: i\neq j$, $(C_2[i],C_2[j]) \in G_2.E$ if and only if $(C_1[i],C_1[j]) \in G_1.E$. One possible solution is to compare one sequence of $G_1.V$ with all the possible sequences of $G_2.V$. If the constraints get satisfies get satisfies at least once, then the value $true$ is returned, otherwise $false$. This approach takes $O(n^n)$ time.\\

\textbf{\textit{Knapsack problem}}: The input is a set of non-divisible objects have some associated weight and value, and a knapsack of a weight-capacity $k$. The objective is to fill the objects (repetition of one type to object is possible indefinitely) in the knapsack such that the total value is maximum.\\

\textbf{\textit{Largest Clique problem}}\index{largest clique problem}: Given a graph $G$, the task is to find the largest set of vertices $C\subseteq G.V$ such that $\forall\ a,b \in C$, if $a\neq b$ then $(a,b) \in G.E$. A possible solution is to check for each possible subset of $G.V$ that satisfy the constraint. The largest of all such sets is the solution. This approach takes $O(2^n)$ time.\\

\textbf{\textit{Maximum independent set problem}}\index{independent set maximum size}: Given a graph $G$, the task is to find the set of vertices $C\subseteq G.V$, of greatest possible size such that $\forall\ a,b \in G.V$, for all $(a,b) \in G.E$, if $a \in C$, then $b \not\in C$. A possible solution is to check for each possible subset of $G.V$ that satisfy the constraint. The largest of all such sets is the solution. This approach takes $O(2^n)$ time.\\

\textbf{\textit{Minimum vertex cover problem}}\index{minimum vertex cover problem}: Given a graph $G$, the task is to find the set of vertices $C\subseteq G.V$, of least possible size such that $\forall\ (a,b) \in G.E$, either $a \in C$ or $b\in C$ or both $a,b \in C$. A possible solution is to check for each possible subset of $G.V$ that satisfy the constraint. The smallest of all such sets is the solution. This approach takes $O(2^n)$ time.\\

\textbf{\textit{Minimum dominating set problem}}\index{minimum dominating set}: Given a graph $G$, the task is to compute a subset $D$ of $G.V$ of minimum size such that each vertex in $G.V\setminus D$ is connected to at least one vertex in $D$ by an edge.\\

\textbf{\textit{Rod cutting problem}}: Given is the length of a rod $L$ and a list of profit $p_i$ corresponding to all the possible lengths $i$ of the rod $C = \{p_i, i\}_{i=1}^L: i\in\mathbb{N}$. The task is to find how the rod should be cut in order to maximize the profit. One solution is to assume that the rod of length $L$ units can be cut at $L-1$ positions. Then compute the cost of each cut-decisions' sequence. The sequence which produces the maximum of the costs is the output. This procedure takes $O(2^{L-1})$ time.

\section{Graph Classes}

The following are a few graph classes with their definitions.

\subsection{Spider graphs}

\index{spider graphs}In a \textit{spider graph} (as per the usage in this text), only one vertex $c$, the head vertex, is of degree $d \geq 3$; the degree of all other vertices is less than $3$, it is either $2$ or $1$. An example spider graph is presented in \Cref{figure:example-spider-graph}. {\boldmath$SP(s,r)$} \index{$SP(s,r)$} is a spider graph with degree $s$ of head $c$, and length of each arm $r$.

\begin{figure}
    \centering
    \begin{tikzpicture}
        \node [circle, fill=black, inner sep=0pt, minimum size=3pt, label=left:{$v_1$}] (A) at (0,0) {};
        \node [circle, fill=black, inner sep=0pt, minimum size=3pt, label=above:{$v_2$}] (B) at (0,1) {};
        \node [circle, fill=black, inner sep=0pt, minimum size=3pt, label=below:{$v_3$}] (C) at (.67,.33) {};
        \node [circle, fill=black, inner sep=0pt, minimum size=3pt, label=below:{$v_4$}] (D) at (.33,-.67) {};
        \node [circle, fill=black, inner sep=0pt, minimum size=3pt, label=below:{$v_5$}] (E) at (-.33,-.67) {};
        \node [circle, fill=black, inner sep=0pt, minimum size=3pt, label=left:{$v_6$}] (F) at (-.67,.33) {};
        \node [circle, fill=black, inner sep=0pt, minimum size=3pt, label=right:{$v_7$}] (G) at (1.67,.33) {};
        
        \draw (A) -- (B); \draw (A) -- (C); \draw (A) -- (D); \draw (A) -- (E); \draw (A) -- (F);
        \draw (C) -- (G);
    \end{tikzpicture}
    \caption{An example spider graph.}
    \label{figure:example-spider-graph}
\end{figure}
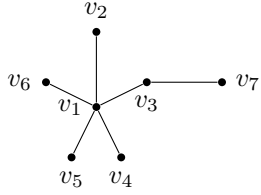

\subsection{\texorpdfstring{\boldmath$P_k$}{p-k}-free graphs}

{\boldmath$P_k$} \index{$P_k$} is a path of $k$ vertices ($k-1$ edges). A graph $G$ is $P_k-free$ \index{$P_k$-free graphs} if any induced subgraph of $G$ does not contain $P_k$.

\subsection{Cographs}\label{subsection:cographs}

\textit{Cographs} \index{cographs} can be recursively defined as follows. A single vertex is a cograph. A disjoint union of two cographs is a cograph. A complete join of two cographs is a cograph. Graph complement of a cograph is a cograph, which is why, cographs are also called \textbf{\textit{complement-reducible graphs}}\index{complement-reducible graphs}. Cographs are $P_4$-free graphs. A few examples of (recursively) constructed cographs are presented in \Cref{figure:example-cographs}.

\begin{figure}
    \centering
    \subfigure[]{
        \begin{tikzpicture}
            \node [circle, fill=black, inner sep=0pt, minimum size=3pt, label=below:{$v$}] at (0,0) {};
        \end{tikzpicture}
    }
    \subfigure[]{
        \begin{tikzpicture}
            \node [circle, fill=black, inner sep=0pt, minimum size=3pt, label=below:{$v_1$}] at (0,0) {};
            \node [circle, fill=black, inner sep=0pt, minimum size=3pt, label=below:{$v_2$}] at (1,0) {};
        \end{tikzpicture}
    }
    \subfigure[]{
        \begin{tikzpicture}
            \node [circle, fill=black, inner sep=0pt, minimum size=3pt, label=above:{$v_1$}] (A) at (0,0) {};
            \node [circle, fill=black, inner sep=0pt, minimum size=3pt, label=above:{$v_2$}] (B) at (1,0) {};
            
            \node [circle, fill=black, inner sep=0pt, minimum size=3pt, label=below:{$v_3$}] (C) at (1,-1) {};
            \node [circle, fill=black, inner sep=0pt, minimum size=3pt, label=below:{$v_4$}] (D) at (0,-1) {};
            
            \draw (A) -- (C); \draw (A) -- (D);
            \draw (B) -- (C); \draw (B) -- (D);
        \end{tikzpicture}
    }
    \subfigure[]{
        \begin{tikzpicture}
            \node [circle, fill=black, inner sep=0pt, minimum size=3pt, label=above:{$v_1$}] (A) at (0,0) {};
            \node [circle, fill=black, inner sep=0pt, minimum size=3pt, label=above:{$v_2$}] (B) at (1,0) {};
            
            \node [circle, fill=black, inner sep=0pt, minimum size=3pt, label=below:{$v_3$}] (C) at (1,-1) {};
            \node [circle, fill=black, inner sep=0pt, minimum size=3pt, label=below:{$v_4$}] (D) at (0,-1) {};
            
            \node [circle, fill=black, inner sep=0pt, minimum size=3pt, label=below:{$v_5$}] (E) at (2,0) {};
            \node [circle, fill=black, inner sep=0pt, minimum size=3pt, label=below:{$v_6$}] (F) at (2,-1) {};
            
            \draw (A) -- (C); \draw (A) -- (D);
            \draw (B) -- (C); \draw (B) -- (D);
        \end{tikzpicture}
    }
    \subfigure[]{
        \begin{tikzpicture}
            \node [circle, fill=black, inner sep=0pt, minimum size=3pt, label=above:{$v_1$}] (A) at (0,-.25) {};
            \node [circle, fill=black, inner sep=0pt, minimum size=3pt, label=above:{$v_2$}] (B) at (1,0) {};
            
            \node [circle, fill=black, inner sep=0pt, minimum size=3pt, label=below:{$v_3$}] (C) at (1,-1) {};
            \node [circle, fill=black, inner sep=0pt, minimum size=3pt, label=below:{$v_4$}] (D) at (0,-.75) {};
            
            \node [circle, fill=black, inner sep=0pt, minimum size=3pt, label=below:{$v_5$}] (E) at (2,0) {};
            \node [circle, fill=black, inner sep=0pt, minimum size=3pt, label=below:{$v_6$}] (F) at (2,-1) {};
            
            \draw (A) -- (B);
            \draw (A) -- (E); \draw (A) -- (F);
            \draw (B) -- (E); \draw (B) -- (F);
            
            \draw (C) -- (D);
            \draw (C) -- (E); \draw (C) -- (F);
            \draw (D) -- (E); \draw (D) -- (F);
        \end{tikzpicture}
    }
    \caption{a: a single vertex is a \textit{cograph}; b: a disjoint union of two copies of the \textit{cograph} in (a) is a \textit{cograph}; c: a complete join of two copies of the \textit{cograph} in (b) is a cograph; d: a disjoint union of the \textit{cographs} in (b) and (c) is a \textit{cograph}; e: the graph complement of the \textit{cograph} in (d) is a \textit{cograph}.}
    \label{figure:example-cographs}
\end{figure}
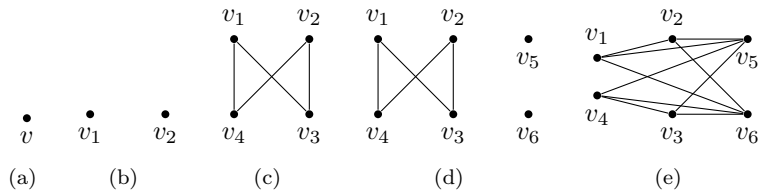

\subsection{Disk Graphs}

A graph $G$ is a \textit{disk graph} \index{disk graphs} if there is an edge between a pair of vertices iff the circles drawn on the plane with those vertices as centers overlap. These circles can generally be of arbitrary radius. If radius of all the circles overlap, the graph is called a \textbf{\textit{unit disk graph}}.

For example, let a circle $C$ of radius $R=2$ be positioned on the plane at $(0,0)$, three more disks $c_1$, $c_2$ and $c_3$, each of radius $r=1$, is placed with their centres respectively at $(2,0)$, $(0,2)$, and $(-2,0)$. Let us assume that a chain of $4$ disks $Ch_1=(c_1^1, c_1^2, c_1^3, c_1^4)$ is attached to $c_1$ such that $c_1^1$ overlaps with $c_1$ and $c_1^2$ only, $c_1^4$ overlaps with $c_1^3$ only, and $\forall\ 2 \leq j \leq 3$, $c_1^j$ overlap with only $c_1^{j-1}$ and $c_1^{j+1}$. Exactly in the similar way, there is a chain behind each of $c_2$ ($Ch_2=(c_2^1, c_2^2, c_2^3, c_2^4)$) and $c_3$ ($Ch_3=(c_3^1, c_3^2, c_3^3, c_3^4)$).

There are $q=3$ chains of disks, and $p=4$ more disks behind the first disk in each chain. Let $Ch = \{Ch_1, Ch_2, Ch_3\}$ and $Cir = \{c_1, c_2, c_3\}$ We denote this network of disks by $DK(R$, $r$, $q$, $p$, $C$, $Cir$, $Ch) = DK(2$, $1$, $3$, $4$, $C$, $Cir$, $Ch)$.

Let the vertex corresponding to $C$ be called head $h$, vertices corresponding to $c_i$ be called $v_i$, and the vertices corresponding to $c_i^j$ be called $v_i^j$, $\forall\ 1 \leq i \leq 3$, and $\forall\ 1 \leq j \leq 4$. The graph formed by this setting will be a spider graph $SP(3,5)$, as shown in \Cref{figure:example-disk-graph-1}.

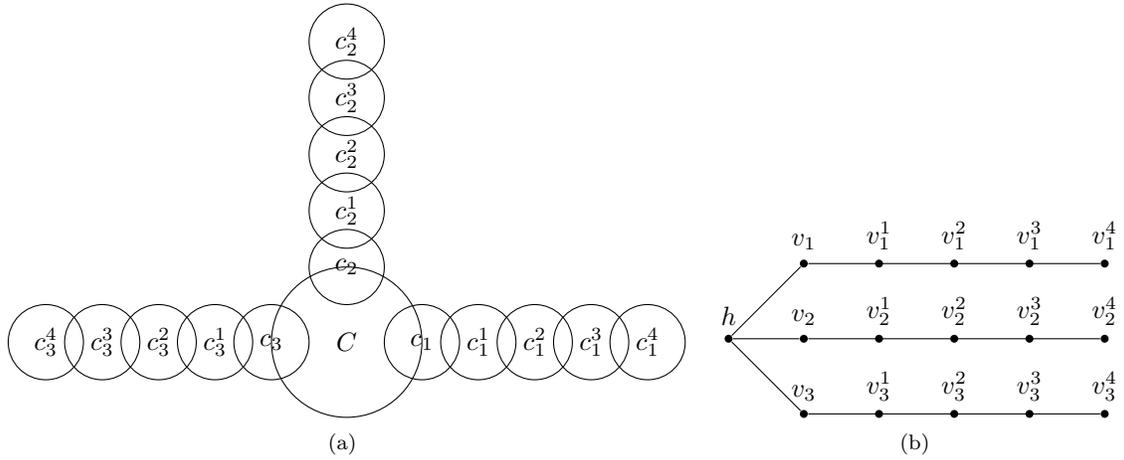
\begin{figure}
    \centering
    \subfigure[]{
        \begin{tikzpicture}[scale=0.5]
            \draw (0,0) circle (2); \node at (0,0) {$C$};
                
            \draw (2,0) circle (1); \node at (2,0) {$c_1$};
            \draw (3.5,0) circle (1); \node at (3.5,0) {$c_1^1$};
            \draw (5,0) circle (1); \node at (5,0) {$c_1^2$};
            \draw (6.5,0) circle (1); \node at (6.5,0) {$c_1^3$};
            \draw (8,0) circle (1); \node at (8,0) {$c_1^4$};
            
            \draw (0,2) circle (1); \node at (0,2) {$c_2$};
            \draw (0,3.5) circle (1); \node at (0,3.5) {$c_2^1$};
            \draw (0,5) circle (1); \node at (0,5) {$c_2^2$};
            \draw (0,6.5) circle (1); \node at (0,6.5) {$c_2^3$};
            \draw (0,8) circle (1); \node at (0,8) {$c_2^4$};
            
            \draw (-2,0) circle (1); \node at (-2,0) {$c_3$};
            \draw (-3.5,0) circle (1); \node at (-3.5,0) {$c_3^1$};
            \draw (-5,0) circle (1); \node at (-5,0) {$c_3^2$};
            \draw (-6.5,0) circle (1); \node at (-6.5,0) {$c_3^3$};
            \draw (-8,0) circle (1); \node at (-8,0) {$c_3^4$};
        \end{tikzpicture}
    }
    \subfigure[]{
        \begin{tikzpicture}
            \node [circle, fill=black, inner sep=0pt, minimum size=3pt, label=above:{$h$}] (C) at (0,0) {};
            
            \node [circle, fill=black, inner sep=0pt, minimum size=3pt, label=above:{$v_1$}] (A) at (1,1) {};
            \node [circle, fill=black, inner sep=0pt, minimum size=3pt, label=above:{$v_2$}] (B) at (1,0) {};
            \node [circle, fill=black, inner sep=0pt, minimum size=3pt, label=above:{$v_3$}] (D) at (1,-1) {};
            
            \draw (A) -- (C); \draw (B) -- (C); \draw (D) -- (C);
                
            \setcounter{c}{1}
            \loop
                \node [circle, fill=black, inner sep=0pt, minimum size=3pt, label=above:{$v_1^{\thec}$}] (AA) at (\value{c}+1,1) {};
                \node [circle, fill=black, inner sep=0pt, minimum size=3pt, label=above:{$v_2^{\thec}$}] (AB) at (\value{c}+1,0) {};
                \node [circle, fill=black, inner sep=0pt, minimum size=3pt, label=above:{$v_3^{\thec}$}] (AD) at (\value{c}+1,-1) {};
                
                \stepcounter{c}
                \ifnum \value{c} < 5
                    \repeat
                    
            \draw (AA) -- (A); \draw (AB) -- (B); \draw (AD) -- (D);
            
        \end{tikzpicture}
    }
    \caption{(a) arrangement of disks and (b) the corresponding disk graph, geometrically not according to the arrangement of the disks, but connections according to the overlap of respective disks.}
    \label{figure:example-disk-graph-1}
\end{figure}

Another example of an arrangement of disks is shown in \Cref{figure:example-disk-graph-2}.

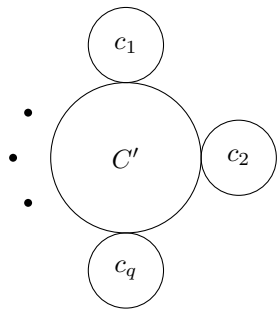
\begin{figure}
    \centering
    \begin{tikzpicture}
        \draw (0,0) circle (1); \node at (0,0) {$C^\prime$};
        
        \draw (0,1.5) circle (.5); \node at (0,1.5) {$c_1$};
        \draw (1.5,0) circle (.5); \node at (1.5,0) {$c_2$};
        \node [circle, fill=black, inner sep=0pt, minimum size=3pt] (D) at (-1.5+.2,.6) {};
        \node [circle, fill=black, inner sep=0pt, minimum size=3pt] (D) at (-1.5,0) {};
        \node [circle, fill=black, inner sep=0pt, minimum size=3pt] (D) at (-1.5+.2,-.6) {};
        \draw (0,-1.5) circle (.5); \node at (0,-1.5) {$c_q$};
    \end{tikzpicture}
    \caption{Central disk {$C^\prime$} and a set of $Cir = \{c_1, c_2, ..., c_q\}$ disks with their circumference touching the circumference of {$C^\prime$}, and not overlaping with each other, or with $C^\prime$}
    \label{figure:example-disk-graph-2}
\end{figure}

\subsection{Interval Graphs}\label{subsection:interval-graphs}

An \textit{interval graph} \index{interval graphs} is formed from a set of intervals on the real line where each interval is represented as a vertex and there is an edge between two vertices if an only if their corresponding intervals overlap on the real line.

\subsubsection{Interval graphs from a set of intervals}

The input is the list of intervals $L$, each interval $i$ has a starting time $s_i$ and an ending time $e_i$. Each interval in $L$ corresponds to a vertex in $G$.

To convert a set of intervals to an interval graph, for each interval $a$ $\in$ $L$, a vertex $v_a$ is added to $G$.$V$. Wherever there is an overlap between any two distinct intervals $a$ and $b$, $a$, $b$ $\in$ $L$, that is, if $s_b$ $\geq$ $s_a$ and $s_b < e_a$, we add an edge ($v_a$, $v_b$) in $G.E$. After following this procedure, $G$ represents the interval graph corresponding to $L$. This is demonstrated in \Cref{figure:example-interval-graph}.

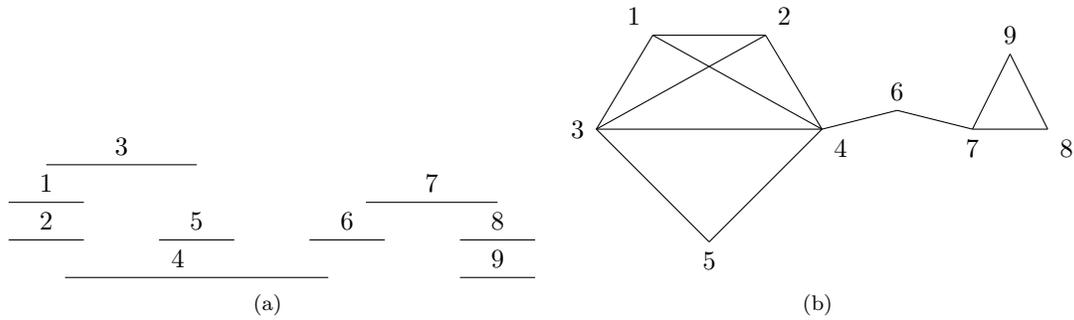
\begin{figure}
    \centering
    \subfigure[]{
		\begin{tikzpicture}
			\draw (-4, 0) -- (-3, 0);
			\draw (-4, 0.5) -- (-3, 0.5);
			\draw (-3.5, 1) -- (-1.5, 1);
			\draw (-3.25, -0.5) -- (0.25, -0.5);
			\draw (-2, 0) -- (-1, 0);
			
			\draw (0, 0) -- (1, 0);
			\draw (0.75, 0.5) -- (2.5, 0.5);
			\draw (2, 0) -- (3, 0);
			\draw (2, -0.5) -- (3, -0.5);
			
			\node at (-3.5, 0.75) {$1$};
			\node at (-3.5, 0.25) {$2$};
			\node at (-2.5, 1.25) {$3$};
			\node at (-1.75, -0.25) {$4$};
			\node at (-1.5, 0.25) {$5$};
			
			\node at (0.5, 0.25) {$6$};
			\node at (1.625, 0.75) {$7$};
			\node at (2.5, 0.25) {$8$};
			\node at (2.5, -0.25) {$9$};
		\end{tikzpicture}
	}
	\subfigure[]{
		\begin{tikzpicture}
			\draw (-4, 0) -- (-1, 0);
			
			\draw (-3.25, 1.25) -- (-1.75, 1.25);
			\draw (-4, 0) -- (-3.25, 1.25);
			\draw (-1.75, 1.25) -- (-1, 0);
			\draw (-4, 0) -- (-1.75, 1.25);
			\draw (-3.25, 1.25) -- (-1, 0);
			
			\draw (-4, 0) -- (-2.5, -1.5);
			\draw (-2.5, -1.5) -- (-1, 0);
			
			\draw (-1, 0) -- (0, 0.25);
			
			\draw (0, 0.25) -- (1, 0);
			
			\draw (1, 0) -- (2, 0);
			
			\draw (1, 0) -- (1.5, 1);
			\draw (1.5, 1) -- (2, 0);
			
			\node at (-3.5, 1.5) {$1$};
			\node at (-1.5, 1.5) {$2$};
			\node at (-4.25, 0) {$3$};
			\node at (-0.75, -0.25) {$4$};
			\node at (-2.5, -1.75) {$5$};
			
			\node at (0, 0.5) {$6$};
			\node at (1, -0.25) {$7$};
			\node at (2.25, -0.25) {$8$};
			\node at (1.5, 1.25) {$9$};
		\end{tikzpicture}
	}
	\caption{An example set of intervals (a) and the corresponding interval graph (b).}
	\label{figure:example-interval-graph}
\end{figure}

As demonstrated in \Cref{figure:example-interval-graph-frames}, we can also imagine this procedure as follows. We imagine that a vertical line that traverses the intervals from left to right (we could otherwise do right to left). The vertical line will intersect the horizontal intervals while traversal, it may intersect zero, one, or more intervals at a particular instant. We draw a clique between the corresponding vertices for each distinct set of intervals that the vertical line intersects. Each distinct instance that can possibly be presented by the position of the vertical line is called a \textit{frame}.

\begin{figure}
	\centering
	\begin{tikzpicture}
		\draw (-4, 0) -- (-3, 0);
		\draw (-4, 0.5) -- (-3, 0.5);
		\draw (-3.5, 1) -- (-1.5, 1);
		\draw (-3.25, -0.5) -- (0.25, -0.5);
		\draw (-2, 0) -- (-1, 0);
		
		\draw (0, 0) -- (1, 0);
		\draw (0.75, 0.5) -- (2.5, 0.5);
		\draw (2, 0) -- (3, 0);
		\draw (2, -0.5) -- (3, -0.5);
		
		\node at (-3.5, 0.75) {$1$};
		\node at (-3.5, 0.25) {$2$};
		\node at (-2.5, 1.25) {$3$};
		\node at (-1.75, -0.25) {$4$};
		\node at (-1.5, 0.25) {$5$};
		
		\node at (0.5, 0.25) {$6$};
		\node at (1.625, 0.75) {$7$};
		\node at (2.5, 0.25) {$8$};
		\node at (2.5, -0.25) {$9$};
		
		\draw[red] (-3.75, 1.5) -- (-3.75, -1);
		\draw[red] (-3.375, 1.5) -- (-3.375, -1);
		\draw[red] (-3.125, 1.5) -- (-3.125, -1);
		\draw[red] (-1.75, 1.5) -- (-1.75, -1);
		\draw[red] (0.125, 1.5) -- (0.125, -1);
		\draw[red] (0.875, 1.5) -- (0.875, -1);
		\draw[red] (2.25, 1.5) -- (2.25, -1);
	\end{tikzpicture}
	\caption{Positions of the vertical line while traversing an example set on intervals on the real line. This figure presents each frame where a new untraversed interval is encountered and a clique in the graph (in \Cref{figure:example-interval-graph}-(b)) is added.}
	\label{figure:example-interval-graph-frames}
\end{figure}
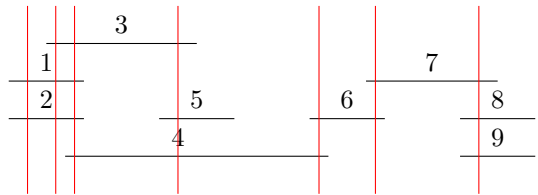

\subsection{Permutation graphs}

A \textit{permutation graph} \index{permutation graphs} $G$ is constructed from an original sequence of numbers $O = (1, 2, 3, . . ., k)$ and its permutation $P = (p_1, p_2, p_3, . . ., p_k)$ such that there is an edge between vertices $v_i$ and $v_j$, corresponding to the numbers $i$ and $j$ in $O$, if $i < j$ in $O$, but $j$ occurs before $i$ in $P$. For example, let $O = (1, 2, 3, 4, 5, 6, 7, 8)$ and $P = (3, 1, 5, 2, 7, 4, 8, 6)$ be the subject permutation of $O$. Then the permutation graph $G$ formed from this pair $(O, P)$ is shown as \Cref{figure:example-permutation-graph}.

\begin{figure}
	\centering
	\begin{tikzpicture}
	    \draw (0,0) -- (7,0);
	    
	    \node [circle, fill=black, inner sep=0pt, minimum size=3pt, label=below:{$v_1$}] at (0,0) {};
		\node [circle, fill=black, inner sep=0pt, minimum size=3pt, label=below:{$v_3$}] at (1,0) {};
	    \node [circle, fill=black, inner sep=0pt, minimum size=3pt, label=below:{$v_2$}] at (2,0) {};
		\node [circle, fill=black, inner sep=0pt, minimum size=3pt, label=below:{$v_5$}] at (3,0) {};
	    \node [circle, fill=black, inner sep=0pt, minimum size=3pt, label=below:{$v_4$}] at (4,0) {};
	    \node [circle, fill=black, inner sep=0pt, minimum size=3pt, label=below:{$v_7$}] at (5,0) {};
	    \node [circle, fill=black, inner sep=0pt, minimum size=3pt, label=below:{$v_6$}] at (6,0) {};
		\node [circle, fill=black, inner sep=0pt, minimum size=3pt, label=below:{$v_8$}] at (7,0) {};
	\end{tikzpicture}
	\caption{Representation of permutation graph corresponding to $(O, P)$, where $O = (1, 2, 3, 4, 5, 6, 7, 8)$ and $P = (3, 1, 5, 2, 7, 4, 8, 6)$.}
	\label{figure:example-permutation-graph}
\end{figure}
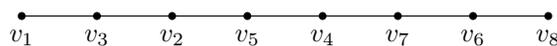

\subsection{Split graphs}\label{subsection:split-graphs}

\index{split graphs}A graph $G$ is a \textit{split graph} if its vertices can be partitioned into a clique and an independent set. A split graph is $P_5$ free. An example split graph is presented in \Cref{figure:example-split-graph}

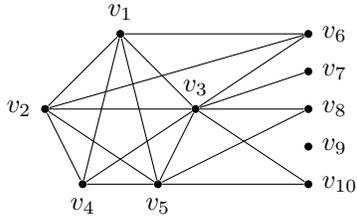
\begin{figure}
    \centering
    \begin{tikzpicture}
        \node [circle, fill=black, inner sep=0pt, minimum size=3pt, label=above:{$v_1$}] (A) at (0,0) {};
        \node [circle, fill=black, inner sep=0pt, minimum size=3pt, label=left:{$v_2$}] (B) at (-1,-1) {};
        \node [circle, fill=black, inner sep=0pt, minimum size=3pt, label=above:{$v_3$}] (C) at (1,-1) {};
        \node [circle, fill=black, inner sep=0pt, minimum size=3pt, label=below:{$v_4$}] (D) at (-.5,-2) {};
        \node [circle, fill=black, inner sep=0pt, minimum size=3pt, label=below:{$v_5$}] (E) at (.5,-2) {};
        
        \draw (A) -- (B); \draw (A) -- (C); \draw (A) -- (D); \draw (A) -- (E);
        \draw (B) -- (C); \draw (B) -- (D); \draw (B) -- (E);
        \draw (C) -- (D); \draw (C) -- (E);
        \draw (D) -- (E);
        
        \node [circle, fill=black, inner sep=0pt, minimum size=3pt, label=right:{$v_6$}] (F) at (2.5,0) {};
        \node [circle, fill=black, inner sep=0pt, minimum size=3pt, label=right:{$v_7$}] (G) at (2.5,-.5) {};
        \node [circle, fill=black, inner sep=0pt, minimum size=3pt, label=right:{$v_8$}] (H) at (2.5,-1) {};
        \node [circle, fill=black, inner sep=0pt, minimum size=3pt, label=right:{$v_9$}] (I) at (2.5,-1.5) {};
        \node [circle, fill=black, inner sep=0pt, minimum size=3pt, label=right:{$v_{10}$}] (J) at (2.5,-2) {};
        
        \draw (A) -- (F); \draw (B) -- (F); \draw (C) -- (F);
        \draw (C) -- (G);
        \draw (C) -- (H); \draw (E) -- (H);
        \draw (C) -- (J); \draw (E) -- (J);
    \end{tikzpicture}
    \caption{An example split graph. $C=\{v_1,v_2,\dots,v_5\}$ is the clique and $I=\{v_6,v_7,\dots,v_{10}\}$ is the independent set excerpt of this split graph.}
    \label{figure:example-split-graph}
\end{figure}

\chapter{Literature Survey}\label{chapter:literature}

\section{Graph burning}\label{section:burning-literature}

The burning number was introduced in \cite{Bonato2016}. This work showed that the burning number of a path or cycle of length $n$ is $\lceil\sqrt{n}\rceil$. They have presented some other properties and results also related to the graph burning problem. Bessy et al. \cite{Bessy2017} showed that general graph burning is NP-Complete. They showed that optimal burning spider graphs, trees, and path-forests is NP-Complete. A 3-approximation algorithm for burning general graphs was described in \cite{Bessy2017}. Bonato et al. \cite{Bonato2019} proposed a $2$-approximation algorithm for burning trees. A 2-approximation algorithm for graphs bounded by a diameter of constant length was described in \cite{Kamali2019,Kamali2020}. A 1.5 approximation algorithm for burning path forests was described in \cite{Bonato2019a}. This article also showed that burning number of spider graphs of order $n$ is atmost $\sqrt{n}$. Simon et al. \cite{Simon2019} presented systems that utilize burning in the spread of an alarm through a network.
Bessy et al. \cite{Bessy2018} proved that the burning number of a connected graph of order $n$ is at most $\sqrt{\frac{12n}{7}}+3 \approx 1.309\sqrt{n}+3$. They also showed that the burning number of trees with $n_2$ vertices of degree $2$, and $n (\geq3)$ vertices of degree at least $3$ is at most $\sqrt{(n+n_2)+\frac{1}{4}}+\frac{1}{2}$. \cite{Simon2019a} have provided heuristics to minimize the time steps in burning a graph. They have studied that which vertices should be selected to be burnt from ``outside'' and in which time-steps.

\section{Cographs}

A cotree \cite{Corneil1981}, constructed from a cograph, assists in the computation of properties of a cograph in linear time, such as \textit{maximum independent set}, \textit{graph coloring number}, and determining if there exists a Hamiltonian cycle. A cograph can be recognized in linear time, and cotrees can be constructed using modular decomposition \cite{Corneil1985}, partition refinement \cite{Habib2005}, LexBFS \cite{Bretscher2008}, or split decomposition \cite{Gioan2012}.  An induced subgraph of a cograph is a cograph itself \cite{Damaschke1990}.

\section{Disk graphs}

Computation of various properties on disk graph is NP-Hard. Clark et al. (1990) \cite{Clark1990} showed that finding the chromatic number of disk graph is NP-Complete. They also showed that 3-coloring problem is NP-Complete on unit disk graph, even if their underlying disk representation is given. Although, in contrast, they also gave a polynomial time algorithm to find maximal cliques in unit disk graph when the geometrical representation of the underlying disks is given.

\section{Interval graphs}

There are several linear time algorithms available to solve different problems on interval graph. Olariu \cite{Olariu1991} discovered linear time algorithm for coloring.  Marathe et al. \cite{Marathe1992} gave a linear time algorithm to compute minimum vertex cover. Similarly, a linear time algorithm to compute interval graph isomosphism was given in \cite{Lueker1979}. Ibarra \cite{Ibarra2017} proposed a linear time algorithm to compute the clique separator graph of a given interval graph. Fomin et al. \cite{Fomin2016} described an algorithm which solves the firefighter problem on interval graphs in $O(n^7)$ time. Interval graph ordering was introduced in \cite{Ramalingam1988}. As per interval graph ordering, the vertices are ordered according to the increasing order of the ending time of their corresponding intervals. Ravi et al. \cite{Ravi1992} proposed an algorithm to solve all pairs shortest path in $(O(n^2))$ time, where $n$ is the number of vertices. Authors in \cite{Kamali2019,Kamali2020}, along with \cite{Kare2019} discussed bounds on the burning number of interval graph. Although they have not provided any algorithm to find an optimal burning sequence.

\subsection{Similarity between a path and an interval graph}\label{subsection:IG-path-similar}

\begin{proposition}\label{proposition:LongestPathLIG}
Let $l,r\in L$ be two intervals such that $P_L$ (shortest path between $l$ and $r$) is of maximum length as compared to shortest path between any pair of intervals in $L$. Then, for each interval $i\in L$, if $i\not\in P_L$, then $i$ shall overlap with at least one interval of $P_L$.
\end{proposition}

\begin{proof}
Let $s = \min s_j:j\in P_L$ and $e = \max e_j:j\in P_L$. Let that $i$ does not overlap with any interval in $L$, then either (1) $e_i < s$, or (2) $s_i > e$.

First let $e_i < s$ for contradiction. Let $P_L^\prime$ be the shortest path between $i$ and $l$, and $P_L^{\prime\prime}$ be the shortest path between $i$ and $r$. Here, we obtain a contradiction because $P_L^{\prime\prime} \geq P_L + P_L^\prime-1$ and $P^\prime \geq 2$ $|P_L^{\prime\prime}|$, so $P_L$ is not of maximum length. Otherwise, $L$ may be disconnected.

Similarly, we can obtain a contradiction for an interval $i$ such that $s_i > e_r$.
\end{proof}

\begin{corollary}\label{corollary:NHClosurePIG}
Let $P$ be the shortest path between a pair of vertices in $G$, which is of maximum length as compared to the shortest path between any other pair of vertices in $G.V$, that is, $P$ is the diameter of $G$, then all vertices in $G\setminus P$ are connected by a single edge with at least one of the vertices in $P$.
\end{corollary}

\Cref{corollary:NHClosurePIG} has been proved earlier in \cite{Kare2019} and \cite{Kamali2019,Kamali2020}. It can also be observed that the interval graphs do not contain a cycle of size more than $3$ \cite{Gilmore1964,Golumbic2004}; we show this in \Cref{proposition:IG-no-cycle}.

\begin{proposition}\label{proposition:IG-no-cycle}
For any induced subgraph $G^\prime$ of an interval graph $G$ with $|G^\prime.V| \neq 3,\ G^\prime$ will not be a cycle.
\end{proposition}

\begin{proof}
Let the interval graph be a cycle of $k \geq 4$ vertices $v_1$, $v_2$, $v_3$, $...$, $v_k$ corresponding to the intervals $i_1, i_2, i_3, ..., i_k$ respectively. The interval $i_j$ overlaps with interval $i_{j+1},\ \forall 1\leq j\leq k-1$. Now to validate this cycle, interval $i_1$ must overlap with $i_k$. This is only possible if either or both of the following occur.

(a) $i_1$ overlaps with all intervals $i_2, i_3, i_4, ..., i_{k-1}$ also.

(b) $i_k$ overlaps with all intervals $i_2, i_3, i_4, ..., i_{k-1}$ also.\\
Here we obtain the contradiction that there will not exist cycle of length greater than $3$ in any induced subgraph of $G$. This is demonstrated in \Cref{figure:IG-no-cycle}.
\end{proof}

\begin{figure}
	\centering
	\subfigure[]{
		\begin{tikzpicture}
			\draw (0, 0) -- (1, 0); \draw[dashed] (1, 0) -- (3.75, 0);
			\draw (0.75, -1) -- (1.75, -1);
			\draw (2.5, -3) -- (3.5, -3);
			\draw (3.25, -4) -- (4.25, -4);  \draw[dashed] (3.25, -4) -- (0.5, -4);
			
			\node at (0.5, -0.5) {$i_1$};
			\node at (1.3, -1.5) {$i_2$};
			
			\node [circle, fill=black, inner sep=0pt, minimum size=3pt] at (2,-1.75) {};
			\node [circle, fill=black, inner sep=0pt, minimum size=3pt] at (2.125,-2) {};
			\node [circle, fill=black, inner sep=0pt, minimum size=3pt] at (2.25,-2.25) {};
			
			\node at (3, -3.5) {$i_{k-1}$};
			\node at (3.8, -4.5) {$i_k$};
			
			\node at (2, -5) {\textit{Set of intervals}};
		\end{tikzpicture}
	}
	\subfigure[]{
		\begin{tikzpicture}
		    \node [circle, fill=black, inner sep=0pt, minimum size=3pt, label=below:{$v_1$}] at (0,0) {};
		    \node [circle, fill=black, inner sep=0pt, minimum size=3pt, label=below:{$v_2$}] at (1,0) {};
		    
		    \node [circle, fill=black, inner sep=0pt, minimum size=2pt] at (1.75,0) {};
		    \node [circle, fill=black, inner sep=0pt, minimum size=2pt] at (2,0) {};
		    \node [circle, fill=black, inner sep=0pt, minimum size=2pt] at (2.25,0) {};
		    
		    \node [circle, fill=black, inner sep=0pt, minimum size=3pt, label=above:{$v_{k-1}$}] at (3,0) {};
		    \node [circle, fill=black, inner sep=0pt, minimum size=3pt, label=above:{$v_k$}] at (4,0) {};
		    
			\draw (0,0) -- (1.5,0);
			\draw (4,0) -- (2.5,0);
			
			\draw (2,1) parabola (0,0); \draw (2,1) parabola (4,0);
			\draw[dashed] (1.5,.75) parabola (0,0); \draw[dashed] (1.5,.75) parabola (3,0);
			\draw[dashed] (.5,.25) parabola (0,0); \draw[dashed] (.5,.25) parabola (1,0);
			
			\draw[dashed] (2.5,-.75) parabola (4,0); \draw[dashed] (2.5,-.75) parabola (1,0);
			\draw[dashed] (3.5,-.25) parabola (4,0); \draw[dashed] (3.5,-.25) parabola (3,0);
			
			\node at (2,-1) {$Corresponding\ interval\ graph$};
		\end{tikzpicture}
	}
	\caption{Demonstration of what happens if we try to create an interval graph cycle and the corresponding set of intervals. The dotted lines are the prospective edges which shall form, apart from the original cycle of $k$ vertices. Either or both of the terminal vertices $v_1$ or $v_k$, on the original path $v_1, ..., v_k$ shall be connected to all the vertices from $v_2,v_3,\dots,v_{k-1}$.}
	\label{figure:IG-no-cycle}
\end{figure}
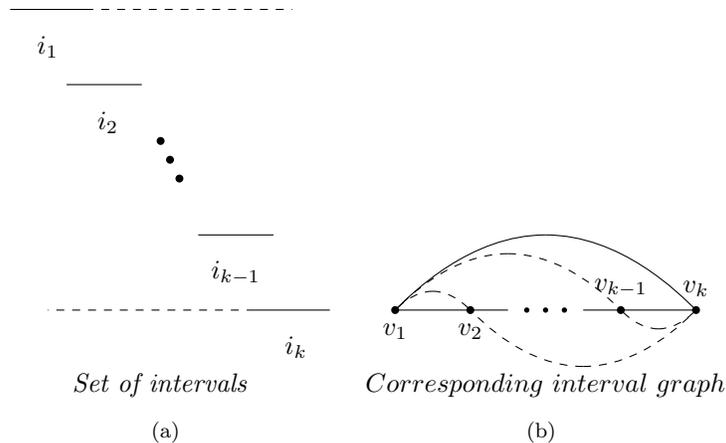

\begin{corollary}
We can conclude two simple facts, which are as follows.

(1) If a vertex $v$ of $G$ is not in $P$, then it is adjacent to some vertex in $P$, it is part of a clique involving that vertex of $P$.

(2) Each clique involves exactly

\quad (a) one vertex from $P$ or

\quad (b) two vertices from $P$ which are adjacent according to the sequence of $P$.
\end{corollary}

\section{Permutation graphs}

Polynomial time algorithms exist for various properties in permutation graphs. If $r$ is the size of longest decreasing subsequence in a permutation $P$, then the chromatic number and the size of largest clique in the corresponding permutation graph $G$, both are equal to $r$ \cite{Golumbic2004}.  Atallah et al. (1998) \cite{Atallah1988} proposed an algorithm that finds the minimum dominating set in an arbitrary permutation graph with $n$ vertices in $O(n\ \log^2 n)$ time.

\section{Split graphs}

Clearly, it is easy to compute the maximum clique on split graphs, and complementarily, its maximum independent set \cite{Golumbic2004,Hammer1981}; along with coloring. \cite{Kare2019} showed that split graphs can be burned in polynomial time. Determining if a hamiltonian cycle exists remains NP-Complete for split graphs \cite{Mueller1996}, along with the minimum dominating set problem \cite{Bertossi1984}.
\chapter{Graph burning: Problem statement}\label{chapter:graph-burning}

\section{Introduction}\label{section:burning-intro}

Works like \cite{Banerjee2012,Domingos2001,Kempe2003,Kempe2005,Mossel2007,Richardson2002} have studied the spread of social influence in order to analyze a social network. \cite{Kramer2014} have highlighted that the underlying network plays an essential role in the spread of an emotional contagion; they have nullified the necessity of in-person interaction and non-verbal cues.

With the aim to being able to model such problems, graph burning was introduced in \cite{Bonato2016}. Graph burning runs on discrete time-steps. During a graph burning process, each vertex is either burned or unburned. We choose one unburned vertex in each step as a ``fire source''. If a node is burned, then it remains in that state until the end of the game. Once a node is burned in time-step $t$, it spreads fire to its neighbouring vertices in step $t+1$ and each of its unburned neighbours are also burned. The aim of the \textit{graph burning} problem is to burn all the vertices in a given graph $G$ in least amount of time-steps. Formally, we describe graph burning as follows.

\noindent\textbf{Arbitrary graph burning}\index{graph burning}: At each step $t$, first (a) an unburned vertex is burned (as a \textit{fire source}) from ``outside'', and then (b) the fire spreads to vertices adjacent to the vertices which are burned till step $t-1$. This process stops after all the vertices of $G$ have been burned.

Burning process on an example graph has been demonstrated in \Cref{figure:arbitrary-burn-example}. Observe that there are $4$ fire sources that burn this graph in this particular procedure. We define burning sequence of an arbitrary graph $G$ in \Cref{definition:burning-sequence} as follows.

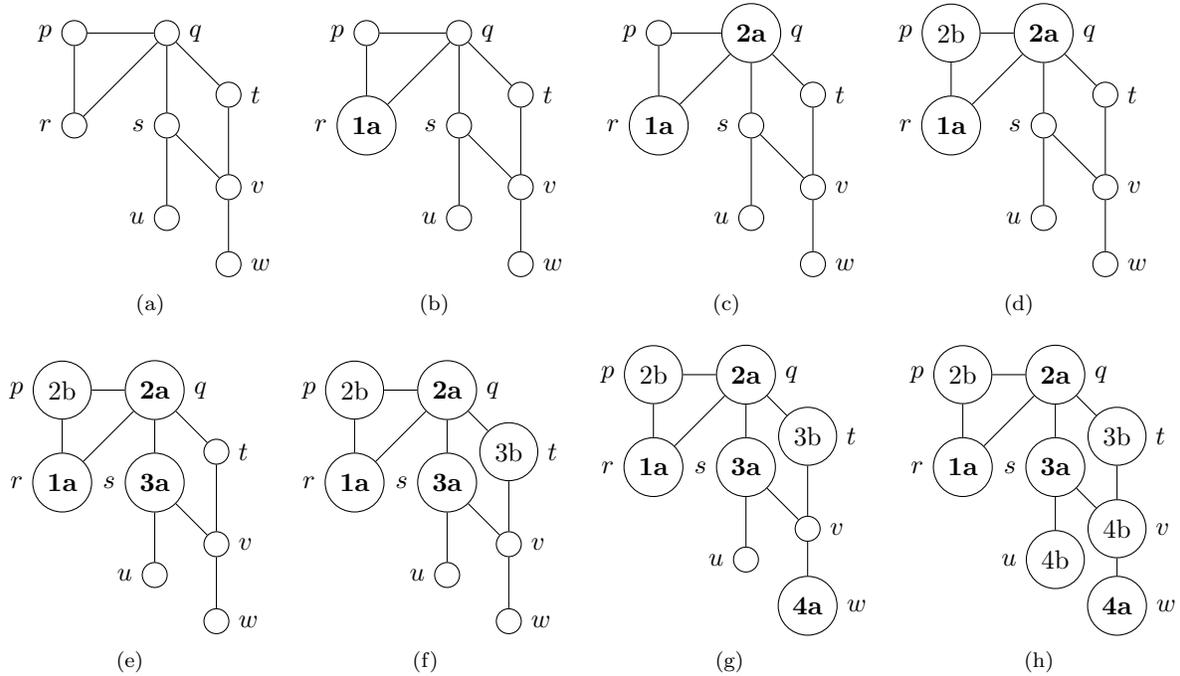
\begin{figure}
	\centering
		\subfigure[]{
		    \begin{tikzpicture}[scale=.41]
			    \node[draw,shape=circle, label=left:{$p$}] (AA) at (-8,0) {};
		    	\node[draw,shape=circle, label=right:{$q$}] (AB) at (-5,0) {};
		    	\node[draw,shape=circle, label=left:{$r$}] (AC) at (-8,-3) {};
		    	\node[draw,shape=circle, label=left:{$s$}] (AD) at (-5,-3) {};
		    	\node[draw,shape=circle, label=right:{$t$}] (AE) at (-3,-2) {};
		    	\node[draw,shape=circle, label=left:{$u$}] (AF) at (-5,-6) {};
		    	\node[draw,shape=circle, label=right:{$v$}] (AG) at (-3,-5) {};
		    	\node[draw,shape=circle, label=right:{$w$}] (AH) at (-3,-7.5) {};
		    	
		    	\draw (AA)--(AB);
		    	\draw (AA)--(AC);
		    	\draw (AB)--(AC);
		    	\draw (AB)--(AD);
		    	\draw (AB)--(AE);
		    	\draw (AE)--(AG);
		    	\draw (AD)--(AF);
	    	    \draw (AD)--(AG);
	    	    \draw (AG)--(AH);
		    \end{tikzpicture}
		}
		\subfigure[]{
		    \begin{tikzpicture}[scale=.41]
			    \node[draw,shape=circle, label=left:{$p$}] (AA) at (-8,0) {};
		    	\node[draw,shape=circle, label=right:{$q$}] (AB) at (-5,0) {};
		    	\node[draw,shape=circle, label=left:{$r$}] (AC) at (-8,-3) {\textbf{1a}};
		    	\node[draw,shape=circle, label=left:{$s$}] (AD) at (-5,-3) {};
		    	\node[draw,shape=circle, label=right:{$t$}] (AE) at (-3,-2) {};
		    	\node[draw,shape=circle, label=left:{$u$}] (AF) at (-5,-6) {};
		    	\node[draw,shape=circle, label=right:{$v$}] (AG) at (-3,-5) {};
		    	\node[draw,shape=circle, label=right:{$w$}] (AH) at (-3,-7.5) {};
		    	
		    	\draw (AA)--(AB);
		    	\draw (AA)--(AC);
		    	\draw (AB)--(AC);
		    	\draw (AB)--(AD);
		    	\draw (AB)--(AE);
		    	\draw (AE)--(AG);
		    	\draw (AD)--(AF);
	    	    \draw (AD)--(AG);
	    	    \draw (AG)--(AH);
		    \end{tikzpicture}
		}
		\subfigure[]{
		    \begin{tikzpicture}[scale=.41]
			    \node[draw,shape=circle, label=left:{$p$}] (AA) at (-8,0) {};
		    	\node[draw,shape=circle, label=right:{$q$}] (AB) at (-5,0) {\textbf{2a}};
		    	\node[draw,shape=circle, label=left:{$r$}] (AC) at (-8,-3) {\textbf{1a}};
		    	\node[draw,shape=circle, label=left:{$s$}] (AD) at (-5,-3) {};
		    	\node[draw,shape=circle, label=right:{$t$}] (AE) at (-3,-2) {};
		    	\node[draw,shape=circle, label=left:{$u$}] (AF) at (-5,-6) {};
		    	\node[draw,shape=circle, label=right:{$v$}] (AG) at (-3,-5) {};
		    	\node[draw,shape=circle, label=right:{$w$}] (AH) at (-3,-7.5) {};
		    	
		    	\draw (AA)--(AB);
		    	\draw (AA)--(AC);
		    	\draw (AB)--(AC);
		    	\draw (AB)--(AD);
		    	\draw (AB)--(AE);
		    	\draw (AE)--(AG);
		    	\draw (AD)--(AF);
	    	    \draw (AD)--(AG);
	    	    \draw (AG)--(AH);
		    \end{tikzpicture}
		}
		\subfigure[]{
		    \begin{tikzpicture}[scale=.41]
			    \node[draw,shape=circle, label=left:{$p$}] (AA) at (-8,0) {2b};
		    	\node[draw,shape=circle, label=right:{$q$}] (AB) at (-5,0) {\textbf{2a}};
		    	\node[draw,shape=circle, label=left:{$r$}] (AC) at (-8,-3) {\textbf{1a}};
		    	\node[draw,shape=circle, label=left:{$s$}] (AD) at (-5,-3) {};
		    	\node[draw,shape=circle, label=right:{$t$}] (AE) at (-3,-2) {};
		    	\node[draw,shape=circle, label=left:{$u$}] (AF) at (-5,-6) {};
		    	\node[draw,shape=circle, label=right:{$v$}] (AG) at (-3,-5) {};
		    	\node[draw,shape=circle, label=right:{$w$}] (AH) at (-3,-7.5) {};
		    	
		    	\draw (AA)--(AB);
		    	\draw (AA)--(AC);
		    	\draw (AB)--(AC);
		    	\draw (AB)--(AD);
		    	\draw (AB)--(AE);
		    	\draw (AE)--(AG);
		    	\draw (AD)--(AF);
	    	    \draw (AD)--(AG);
	    	    \draw (AG)--(AH);
		    \end{tikzpicture}
		}
		\subfigure[]{
		    \begin{tikzpicture}[scale=.41]
			    \node[draw,shape=circle, label=left:{$p$}] (AA) at (-8,0) {2b};
		    	\node[draw,shape=circle, label=right:{$q$}] (AB) at (-5,0) {\textbf{2a}};
		    	\node[draw,shape=circle, label=left:{$r$}] (AC) at (-8,-3) {\textbf{1a}};
		    	\node[draw,shape=circle, label=left:{$s$}] (AD) at (-5,-3) {\textbf{3a}};
		    	\node[draw,shape=circle, label=right:{$t$}] (AE) at (-3,-2) {};
		    	\node[draw,shape=circle, label=left:{$u$}] (AF) at (-5,-6) {};
		    	\node[draw,shape=circle, label=right:{$v$}] (AG) at (-3,-5) {};
		    	\node[draw,shape=circle, label=right:{$w$}] (AH) at (-3,-7.5) {};
		    	
		    	\draw (AA)--(AB);
		    	\draw (AA)--(AC);
		    	\draw (AB)--(AC);
		    	\draw (AB)--(AD);
		    	\draw (AB)--(AE);
		    	\draw (AE)--(AG);
		    	\draw (AD)--(AF);
	    	    \draw (AD)--(AG);
	    	    \draw (AG)--(AH);
		    \end{tikzpicture}
		}
		\subfigure[]{
		    \begin{tikzpicture}[scale=.41]
			    \node[draw,shape=circle, label=left:{$p$}] (AA) at (-8,0) {2b};
		    	\node[draw,shape=circle, label=right:{$q$}] (AB) at (-5,0) {\textbf{2a}};
		    	\node[draw,shape=circle, label=left:{$r$}] (AC) at (-8,-3) {\textbf{1a}};
		    	\node[draw,shape=circle, label=left:{$s$}] (AD) at (-5,-3) {\textbf{3a}};
		    	\node[draw,shape=circle, label=right:{$t$}] (AE) at (-3,-2) {3b};
		    	\node[draw,shape=circle, label=left:{$u$}] (AF) at (-5,-6) {};
		    	\node[draw,shape=circle, label=right:{$v$}] (AG) at (-3,-5) {};
		    	\node[draw,shape=circle, label=right:{$w$}] (AH) at (-3,-7.5) {};
		    	
		    	\draw (AA)--(AB);
		    	\draw (AA)--(AC);
		    	\draw (AB)--(AC);
		    	\draw (AB)--(AD);
		    	\draw (AB)--(AE);
		    	\draw (AE)--(AG);
		    	\draw (AD)--(AF);
	    	    \draw (AD)--(AG);
	    	    \draw (AG)--(AH);
		    \end{tikzpicture}
		}
		\subfigure[]{
		    \begin{tikzpicture}[scale=.41]
			    \node[draw,shape=circle, label=left:{$p$}] (AA) at (-8,0) {2b};
		    	\node[draw,shape=circle, label=right:{$q$}] (AB) at (-5,0) {\textbf{2a}};
		    	\node[draw,shape=circle, label=left:{$r$}] (AC) at (-8,-3) {\textbf{1a}};
		    	\node[draw,shape=circle, label=left:{$s$}] (AD) at (-5,-3) {\textbf{3a}};
		    	\node[draw,shape=circle, label=right:{$t$}] (AE) at (-3,-2) {3b};
		    	\node[draw,shape=circle, label=left:{$u$}] (AF) at (-5,-6) {};
		    	\node[draw,shape=circle, label=right:{$v$}] (AG) at (-3,-5) {};
		    	\node[draw,shape=circle, label=right:{$w$}] (AH) at (-3,-7.5) {\textbf{4a}};
		    	
		    	\draw (AA)--(AB);
		    	\draw (AA)--(AC);
		    	\draw (AB)--(AC);
		    	\draw (AB)--(AD);
		    	\draw (AB)--(AE);
		    	\draw (AE)--(AG);
		    	\draw (AD)--(AF);
	    	    \draw (AD)--(AG);
	    	    \draw (AG)--(AH);
		    \end{tikzpicture}
		}
		\subfigure[]{
		    \begin{tikzpicture}[scale=.41]
			    \node[draw,shape=circle, label=left:{$p$}] (AA) at (-8,0) {2b};
		    	\node[draw,shape=circle, label=right:{$q$}] (AB) at (-5,0) {\textbf{2a}};
		    	\node[draw,shape=circle, label=left:{$r$}] (AC) at (-8,-3) {\textbf{1a}};
		    	\node[draw,shape=circle, label=left:{$s$}] (AD) at (-5,-3) {\textbf{3a}};
		    	\node[draw,shape=circle, label=right:{$t$}] (AE) at (-3,-2) {3b};
		    	\node[draw,shape=circle, label=left:{$u$}] (AF) at (-5,-6) {4b};
		    	\node[draw,shape=circle, label=right:{$v$}] (AG) at (-3,-5) {4b};
		    	\node[draw,shape=circle, label=right:{$w$}] (AH) at (-3,-7.5) {\textbf{4a}};
		    	
		    	\draw (AA)--(AB);
		    	\draw (AA)--(AC);
		    	\draw (AB)--(AC);
		    	\draw (AB)--(AD);
		    	\draw (AB)--(AE);
		    	\draw (AE)--(AG);
		    	\draw (AD)--(AF);
	    	    \draw (AD)--(AG);
	    	    \draw (AG)--(AH);
		    \end{tikzpicture}
		    }
	    \caption{Burning a graph. See that part $b$ does not run in step 1 because there is no vertex that was burnt before step 1, so the fire did not spread. Once we label the vertex in some step, we do not change the label in any future steps.}
	    \label{figure:arbitrary-burn-example}
\end{figure}

\begin{definition}\label{definition:burning-sequence}
\textbf{Burning sequence.}\index{burning sequence} The \textit{burning sequence} of
a graph $G$ is the sequence of vertices that were chosen as \textit{fire sources} in part (a) of each time-step to burn a given graph, such that this sequence of vertices is able to burn all the vertices of $G$.
\end{definition}

In the graph demonstrated in \Cref{figure:arbitrary-burn-example}, the burning sequence is $S = (r, q, s, w)$. This means that $r$ was chosen as a fire source in step $1a$, $q$ was chosen as a fire source in step $2a$, and so on.

\section{The underlying problem}\label{section:burn-problem}

Graph burning aims to burn all the vertices in a graph as quickly as possible and has been inspired by other contact processes like firefighting \cite{Hartnell1995}, graph cleaning \cite{Alon2009}, and graph bootstrap percolation \cite{Balogh2012}. The underlying decision problem is described as follows.

\textbf{The decision problem}: Given input is an arbitrary graph $G$ and a constant $k$. The problem is to determine if $G$ can be burnt using a burning sequence of length $k$ or less (or equivalently, in $k$ or less time-steps).\index{graph burning: decision and optimization versions}

Equivalently, we have the optimization version of the graph burning problem.

\textbf{The optimization problem}: Given input is an arbitrary graph $G$. The problem is to compute the minimum number of fire sources (or equivalently, time-steps), that can (collectively in the form of a burning sequence) burn $G$ completely.

With reference to the theory that we established in \Cref{chapter:introduction}, it can be observed that if an optimization algorithm returns a positive integer $k$ as output, the decision algorithm will return $true$ if the input $(G,k)$ is passed to it. Note that the main task in the burning problem is to find a burning sequence of minimum length such that it is able to to burn all the vertices of a graph. Here, we introduce a (new) property of a graph $G$ in \Cref{definition:burning-number}, the \textit{burning number} of a graph $G$ \cite{Bonato2016}, which we denote as $b(G)$.

\begin{definition}\label{definition:burning-number}
\textbf{Burning number.}\index{burning number} The least amount of time steps (or equivalently, fire sources) which are required to burn a graph $G$ is called the \textit{burning number} of $G,\ b(G)$.
\end{definition}

Clearly, the burning number of a graph $G$ tells that how fast $G$ ``can'' be burnt. Observe from \Cref{figure:optimal-burn-example} that the same graph that we burnt in four time-steps in \Cref{figure:arbitrary-burn-example}, can also be burnt in only three time-steps. From \Cref{figure:optimal-burn-example}, we have that the burning sequence $S^{\prime} = (q, v, u)$ of size $3$ is also able to burn this graph completely.

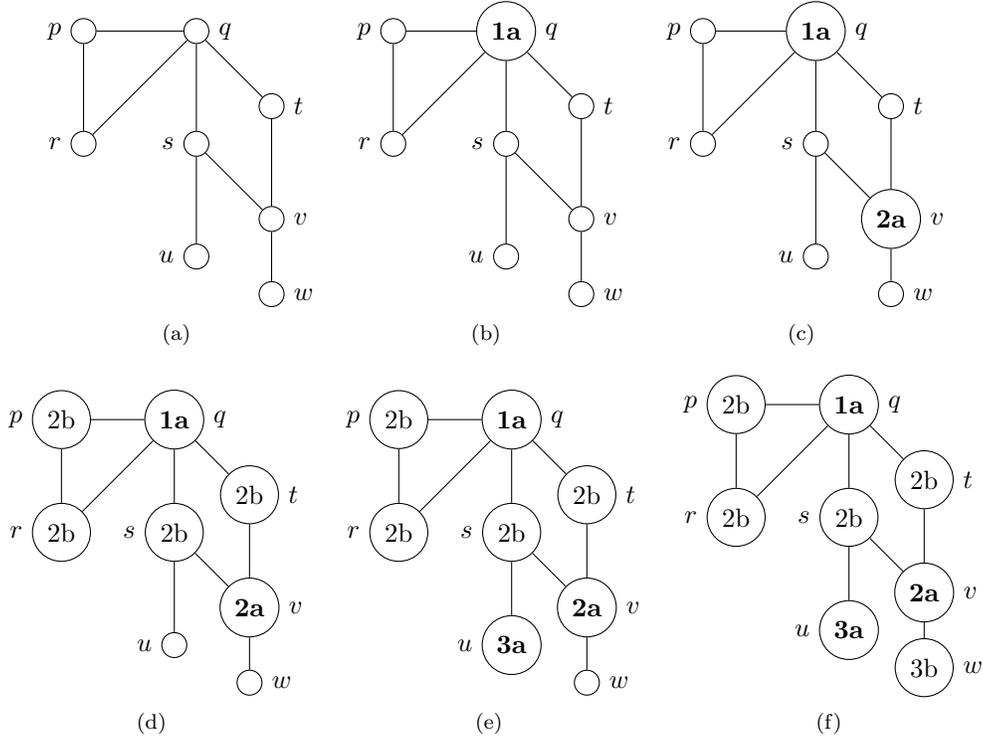
\begin{figure}
	\begin{minipage}{1\textwidth}
		\centering
		    \subfigure[]{
		    \begin{tikzpicture}[scale=0.5]
		    	\node[draw,shape=circle, label=left:{$p$}] (BA) at (8-8,0) {};
		    	\node[draw,shape=circle, label=right:{$q$}] (BB) at (8-5,0) {};
		    	\node[draw,shape=circle, label=left:{$r$}] (BC) at (8-8,-3) {};
		    	\node[draw,shape=circle, label=left:{$s$}] (BD) at (8-5,-3) {};
		    	\node[draw,shape=circle, label=right:{$t$}] (BE) at (8-3,-2) {};
		    	\node[draw,shape=circle, label=left:{$u$}] (BF) at (8-5,-6) {};
		    	\node[draw,shape=circle, label=right:{$v$}] (BG) at (8-3,-5) {};
		    	\node[draw,shape=circle, label=right:{$w$}] (BH) at (8-3,-7) {};
		    	
		    	\draw (BA)--(BB);
		    	\draw (BA)--(BC);
		    	\draw (BB)--(BC);
		    	\draw (BB)--(BD);
		    	\draw (BB)--(BE);
		    	\draw (BE)--(BG);
		    	\draw (BD)--(BF);
	    	    \draw (BD)--(BG);
	    	    \draw (BG)--(BH);
		    \end{tikzpicture}
		    }
		    \subfigure[]{
		    \begin{tikzpicture}[scale=0.5]
		    	\node[draw,shape=circle, label=left:{$p$}] (BA) at (8-8,0) {};
		    	\node[draw,shape=circle, label=right:{$q$}] (BB) at (8-5,0) {\textbf{1a}};
		    	\node[draw,shape=circle, label=left:{$r$}] (BC) at (8-8,-3) {};
		    	\node[draw,shape=circle, label=left:{$s$}] (BD) at (8-5,-3) {};
		    	\node[draw,shape=circle, label=right:{$t$}] (BE) at (8-3,-2) {};
		    	\node[draw,shape=circle, label=left:{$u$}] (BF) at (8-5,-6) {};
		    	\node[draw,shape=circle, label=right:{$v$}] (BG) at (8-3,-5) {};
		    	\node[draw,shape=circle, label=right:{$w$}] (BH) at (8-3,-7) {};
		    	
		    	\draw (BA)--(BB);
		    	\draw (BA)--(BC);
		    	\draw (BB)--(BC);
		    	\draw (BB)--(BD);
		    	\draw (BB)--(BE);
		    	\draw (BE)--(BG);
		    	\draw (BD)--(BF);
	    	    \draw (BD)--(BG);
	    	    \draw (BG)--(BH);
		    \end{tikzpicture}
		    }
		    \subfigure[]{
		    \begin{tikzpicture}[scale=0.5]
		    	\node[draw,shape=circle, label=left:{$p$}] (BA) at (8-8,0) {};
		    	\node[draw,shape=circle, label=right:{$q$}] (BB) at (8-5,0) {\textbf{1a}};
		    	\node[draw,shape=circle, label=left:{$r$}] (BC) at (8-8,-3) {};
		    	\node[draw,shape=circle, label=left:{$s$}] (BD) at (8-5,-3) {};
		    	\node[draw,shape=circle, label=right:{$t$}] (BE) at (8-3,-2) {};
		    	\node[draw,shape=circle, label=left:{$u$}] (BF) at (8-5,-6) {};
		    	\node[draw,shape=circle, label=right:{$v$}] (BG) at (8-3,-5) {\textbf{2a}};
		    	\node[draw,shape=circle, label=right:{$w$}] (BH) at (8-3,-7) {};
		    	
		    	\draw (BA)--(BB);
		    	\draw (BA)--(BC);
		    	\draw (BB)--(BC);
		    	\draw (BB)--(BD);
		    	\draw (BB)--(BE);
		    	\draw (BE)--(BG);
		    	\draw (BD)--(BF);
	    	    \draw (BD)--(BG);
	    	    \draw (BG)--(BH);
		    \end{tikzpicture}
		    }
		    \subfigure[]{
		    \begin{tikzpicture}[scale=0.5]
		    	\node[draw,shape=circle, label=left:{$p$}] (BA) at (8-8,0) {2b};
		    	\node[draw,shape=circle, label=right:{$q$}] (BB) at (8-5,0) {\textbf{1a}};
		    	\node[draw,shape=circle, label=left:{$r$}] (BC) at (8-8,-3) {2b};
		    	\node[draw,shape=circle, label=left:{$s$}] (BD) at (8-5,-3) {2b};
		    	\node[draw,shape=circle, label=right:{$t$}] (BE) at (8-3,-2) {2b};
		    	\node[draw,shape=circle, label=left:{$u$}] (BF) at (8-5,-6) {};
		    	\node[draw,shape=circle, label=right:{$v$}] (BG) at (8-3,-5) {\textbf{2a}};
		    	\node[draw,shape=circle, label=right:{$w$}] (BH) at (8-3,-7) {};
		    	
		    	\draw (BA)--(BB);
		    	\draw (BA)--(BC);
		    	\draw (BB)--(BC);
		    	\draw (BB)--(BD);
		    	\draw (BB)--(BE);
		    	\draw (BE)--(BG);
		    	\draw (BD)--(BF);
	    	    \draw (BD)--(BG);
	    	    \draw (BG)--(BH);
		    \end{tikzpicture}
		    }
		    \subfigure[]{
		    \begin{tikzpicture}[scale=0.5]
		    	\node[draw,shape=circle, label=left:{$p$}] (BA) at (8-8,0) {2b};
		    	\node[draw,shape=circle, label=right:{$q$}] (BB) at (8-5,0) {\textbf{1a}};
		    	\node[draw,shape=circle, label=left:{$r$}] (BC) at (8-8,-3) {2b};
		    	\node[draw,shape=circle, label=left:{$s$}] (BD) at (8-5,-3) {2b};
		    	\node[draw,shape=circle, label=right:{$t$}] (BE) at (8-3,-2) {2b};
		    	\node[draw,shape=circle, label=left:{$u$}] (BF) at (8-5,-6) {\textbf{3a}};
		    	\node[draw,shape=circle, label=right:{$v$}] (BG) at (8-3,-5) {\textbf{2a}};
		    	\node[draw,shape=circle, label=right:{$w$}] (BH) at (8-3,-7) {};
		    	
		    	\draw (BA)--(BB);
		    	\draw (BA)--(BC);
		    	\draw (BB)--(BC);
		    	\draw (BB)--(BD);
		    	\draw (BB)--(BE);
		    	\draw (BE)--(BG);
		    	\draw (BD)--(BF);
	    	    \draw (BD)--(BG);
	    	    \draw (BG)--(BH);
		    \end{tikzpicture}
		    }
		    \subfigure[]{
		    \begin{tikzpicture}[scale=0.5]
		    	\node[draw,shape=circle, label=left:{$p$}] (BA) at (8-8,0) {2b};
		    	\node[draw,shape=circle, label=right:{$q$}] (BB) at (8-5,0) {\textbf{1a}};
		    	\node[draw,shape=circle, label=left:{$r$}] (BC) at (8-8,-3) {2b};
		    	\node[draw,shape=circle, label=left:{$s$}] (BD) at (8-5,-3) {2b};
		    	\node[draw,shape=circle, label=right:{$t$}] (BE) at (8-3,-2) {2b};
		    	\node[draw,shape=circle, label=left:{$u$}] (BF) at (8-5,-6) {\textbf{3a}};
		    	\node[draw,shape=circle, label=right:{$v$}] (BG) at (8-3,-5) {\textbf{2a}};
		    	\node[draw,shape=circle, label=right:{$w$}] (BH) at (8-3,-7) {3b};
		    	
		    	\draw (BA)--(BB);
		    	\draw (BA)--(BC);
		    	\draw (BB)--(BC);
		    	\draw (BB)--(BD);
		    	\draw (BB)--(BE);
		    	\draw (BE)--(BG);
		    	\draw (BD)--(BF);
	    	    \draw (BD)--(BG);
	    	    \draw (BG)--(BH);
		    \end{tikzpicture}
		}
	    \end{minipage}
	    \caption{(Better) Burning procedure (optimal) on an example graph. The subject graph is same as that on which an arbitrary burning procedure is demonstrated in \Cref{figure:arbitrary-burn-example}.}
	    \label{figure:optimal-burn-example}
    \end{figure}

Observe that the example graph, on which two different burning procedures are demonstrated, in \Cref{figure:arbitrary-burn-example} and \Cref{figure:optimal-burn-example} respectively, requires at least $3$ time steps to be burnt; it cannot be burned in less than $3$ time steps. So $S^{\prime}$ is an optimal burning sequence which can burn this graph. Burning number of this graph is $3$.

Let that a burning procedure burns a graph in $k$ time steps. While burning a graph, it is noticeable that if at a step $i$ a vertex $v \in G$ is chosen as a fire source, then at each subsequent step, the set of vertices to which $v$ spreads fire to keeps on increasing till the $k^{th}$ step: at a time step $i+t$, $v$ is able to burn all the vertices in $G.N_t[v]$. Here, we define the \textit{burning cluster} of a fire source in \Cref{definition:burning-cluster} as follows.

\begin{definition}\label{definition:burning-cluster}
\textbf{Burning cluster.}\index{burning cluster} Let that a burning sequence $S=(x_1,x_2,\dots,x_k)$ of size $k$ is able to burn a graph $G$. The burning cluster of a fire source $x_i$ (the fire source chosen at the time-step $i$) is the set of vertices to which $x_i$ is able to spread fire till the end of the burning process (till the $k^{th}$ time-step). This set contains all the vertices in $G.N_{k-i}[x_i]$.
\end{definition}

If $S = (x_1, x_2, x_3, . . ., x_k)$ is the burning sequence which is capable of burning $G$, \Cref{equation:burn-verify} \cite{Bessy2017} must follow.

\begin{equation}\label{equation:burn-verify}
    G.N_{k-1}[x_1] \cup G.N_{k-2}[x_2] \cup . . . \cup G.N_0[x_k] = G.V.
\end{equation}

Some results that have been discovered with respect to graph burning are present in \Cref{section:burning-literature}.

\section{Verification of a burning sequence}\label{section:burning-verify}

As discussed in \Cref{section:problems-classes}, a problem which can be verified deterministically in polynomial time is in NP class. We can verify the validity of a burning sequence in polynomial time. Every burning sequence which satisfies \Cref{equation:burn-verify} is is able to burn $G$ completely, but apart from this, here we also verify that no fire source in a burning sequence should be placed on the vertex which has already been burnt. \Cref{algorithm:burn-verify}\index{burning sequence - verifying correctness} verifies if a given burning sequence $S$ is a valid burning sequence for a graph $G$.

\begin{algorithm}\label{algorithm:burn-verify}
Given an input graph $G$ and a burning sequence $S = (x_1, x_2, x_3, . . ., x_k)$ of length $k$, perform the following steps.
\end{algorithm}

\textbf{\textit{Stage 1.}} If $S$ does not satisfy \Cref{equation:burn-verify}, then return $false$.

\textbf{\textit{Stage 2.}} $\forall\ 1\leq i\leq k-1$, perform the following steps.

\textbf{\textit{Stage 2.1.}} $\forall\ i+1\leq j\leq k$, if $x_j \in G.N_{j-i-1}[x_i]$, then return $false$.

\textbf{\textit{Stage 3.}} Return $true$.\\

If \Cref{algorithm:burn-verify} returns $true$ for the input $(G,S)$, it means that $S$ is a valid burning sequence and is able to burn $G$ completely. \Cref{algorithm:burn-verify} can be implemented in $O(n^2)$ time. This also implies that the graph burning problem is in NP. We state the formally in \Cref{lemma:burning-in-NP}. In fact, optimal burning of general graphs is NP-Hard. We discuss this in the following chapters in detail.

\begin{lemma}\label{lemma:burning-in-NP}
    The (optimal) graph burning problem is in NP.
\end{lemma}

\section{Related problems and games}\label{section:related-games}

As discussed in \Cref{section:burn-problem}, the procedures of certain other problems such as firefighter problem, graph cleaning and graph bootstrap percolation are similar to the procedure of the graph burning problem. In the following few paragraphs in this section, we shall discuss them in brief.

\subsection{Firefighter problem}

The aim of the \textit{firefighter}\index{firefighter problem} problem \cite{Fomin2016} is to save as many vertices of a given graph $G$ as possible from a fire that starts from a single vertex. At step $1$, an arbitrary vertex is burned. At each step $t$, $t \geq 2$, first (a) a firefighter can be placed on an unburned vertex and this firefighter protects that node from fire till the last time-step, and then (b) the fire spreads to the unprotected vertices adjacent to the vertices which are burned till step $t-1$. This process continues till fire cannot spread to any more vertices.

The input is the subject graph $G$ and one fire source $s$. The task is to save the maximum possible number of vertices from fire.

\subsection{Firefighter reserve problem}

\textit{Firefighter reserve deployment}\index{firefighter reserve deployment problem} \cite{Fomin2016} problem proceeds as follows. The fire gets initiated from a single fire source. Initially, there is one firefighter in the firefighter reserve. In each step $t, t \geq 2$, (a) some or no firefighters (subject to availability in the firefighter reserve) are placed, each on an unburned vertex, (b) the fire spreads from the burned vertices to their adjacent vertices which are not protected (by a firefighter), (c) one firefighter increases in the firefighter reserve. This process continues until fire can spread no further.

\subsection{Graph cleaning}

In the \textit{graph cleaning}\index{cleaning: graph} problem, at the beginning, all the vertices and the edges are considered ``dirty''. There are a fixed number of available cleaning brushes. At each step, one vertex $v$ and all the edges incident to $v$ which are dirty may be cleaned if the number of brushes on $v$ are equal to the dirty edges incident on $v$. No brush cleans any edge which is already clean. If a brush cleans an edge, the edge is considered to be traversed. A vertex is cleaned if all the edges incident to it are cleaned, and the cleaning process is done on a vertex $v$ only if we can clean each edge incident on it. A graph $G$ is considered cleaned when each vertex of $G$ has been cleaned. Graph cleaning was introduced in \cite{Mckeil2007,Messinger2008}.

The input is an arbitrary graph $G$, and a constant $k$ number of brushes. The task is to determine if $G$ can be cleaned by $k$ brushes. Equivalently, the optimization problem can be to compute the minimum number of brushes that are required to clean an arbitrary graph.

\subsection{Graph bootstrap percolation}

A phenomenon in graphs called \textit{weak saturation}\index{bootstrap percolation: graph} was introduced by Bollobas in 1968. Given a graph $H$, another graph $G$ of $n$ vertices is called \textit{weakly $H$-saturated} if no subgraph of $G$ can make $H$, but $\exists$ a non-empty set $A$ consisting of edges missing in $G$ such that $\forall\ e \in A$, $H$ is a subgraph of $G+e$ \cite{Faudree2013,Balogh2012,Bollobas2017}. \cite{Balogh2012} observed that weak saturation is strongly related to bootstrap percolation, which was introduced in \cite{Chalupa1979}. \nocite{Faudree2013}

In \textit{bootstrap percolation}, the inputs are an arbitrary graph $G$ and an infection threshold $r>2$. We choose a set of initially ``infected'' vertices $A \subseteq G.V$; we declare the remaining vertices $G.V \setminus A$ ``healthy''.

Then, in consecutive time steps, we infect all healthy vertices which have at least $r$ infected neighbours.  We say that $A$ percolates if, starting from $A$ we are able to finally infect every vertex in $G.V$. More precisely, we set $A_0=A$ and for $t= 1,2,3,...$, we compute $A_t$ according to \Cref{equation:bootstrap-percolation} \cite{Faudree2013} as follows.

\begin{equation}\label{equation:bootstrap-percolation}
    A_t = A_{t-1} \cup \{v\in V: | G.N[v]\cap A_{t-1}| \geq r\}
\end{equation}

Hence $A$ percolates if after computing on equation \thech.2 indefinitely, we infect all the vertices, that is, $\mathop{\cup}_{t=0}^\infty A_t = G.V$.

\section{Overview of possible applications}

Burning a graph can be used to model the spread of a meme, social gossip, emotion, or a social contagion. It can also be used to model spread of viral infections, or otherwise the exposure to infections and proliferation of virus in the body.
Graph burning is relatively a newly introduced procedure. Currently, not much works have come in public domain which utilize the graph burning in modelling of practically applicable systems.

{\~S}imon et al. \cite{Simon2019} have provided heuristics for usage in for spreading an alarm or other critical information in minimum amount of time steps. This may include spread of information via satellite, or throughout a terrestrial network, for example.
They have assumed that a satellite can spread information only sequentially: to one target node (person, device, etc) at a time. On the other hand, each node in their system is able to spread the information parallelly through the available technologies.

Simon et al. \cite{Simon2019a} have simulated a network which tries to spread alarms to all the nodes in the least amount of time steps. The nodes are connected to each other, and at the start of each time step, a new node is alarmed from ``outside''. Also during each time step, the alarmed nodes alarm their neighbour nodes, same as the graph burning procedure. This process stops when all the nodes are alarmed.

\section{Computational limitations}

Optimal graph burning is hard \cite{Bessy2017} for general graphs, like several other problems such as coloring a graph, finding the largest clique in a graph, or finding maximum sized independent set of a graph. It has been shown that graph burning is NP-Complete for several graph classes, even when computing other properties, which are NP-Hard for general graph classes, is ``easy'': graph classes on which other NP-Hard problems can be solved in polynomial time. It has been proved in \cite{Bessy2017} that burning a spider graph, path forests, and trees with maximum degree 3 is NP-Complete.

On the other hand, we have a 3-approximation algorithm for burning general graphs, a 2-approximation algorithm for burning trees, a 2-approximation algorithm for burning the graphs bounded by a diameter of constant length, and a 1.5-approximation algorithm for burning path forests, as discussed in \Cref{section:burning-literature}.

In the following chapters, we discuss these characteristics of the graph burning problem in detail from the perspective of certain graph classes. For the overview of chapters, see \Cref{section:organization-of-chapters}.

\section{Summary of this chapter}

Graph burning represents the multiplication and spread of an object or phenomenon throughout a network under some strict constraints. One of the important aspect is, in each time-step, an object/phenomenon infects one of the uninfected nodes. Despite of its high computational complexity, it can be further utilized to model some useful processes in a computer.

\chapter{The burning process in graphs}\label{chapter:graph-burning-examples}

\section{Burning a simple path or cycle}\label{section:burn-path}

A path is simple to burn. A path and a cycle (of equal length) are equivalent from the burning perspective. In each step, we select one new unburnt ``victim'' vertex (a fire source) and burn vertices adjacent to other clusters. This follows from the definition of graph burning that we discussed in \Cref{section:burning-intro}. Each cluster burns its neighbour vertices and includes them in itself. The interesting fact here is, a burning cluster can spread fire to atmost two more vertices. Each burning cluster starts with on single vertex (a fire source) and in successive time steps, it can get enlarged in size by atmost two more vertices. A newly burned fire source also acts as a distinct cluster.

\subsection{Burning a path (or a cycle) of infinite length}

\index{burning simple path or cycle} Let us consider that a path is so long (of infinite length) that we can put fire sources in such a way that any two fire sources are at infinite distance from one another, as shown in \Cref{figure:burn-path-inf}. We select a fire source, and the fire in all other clusters spreads to their respective neighbouring vertices, i.e., each burning cluster has $2$ more vertices; this is shown in \Cref{table:burn-path-inf}.

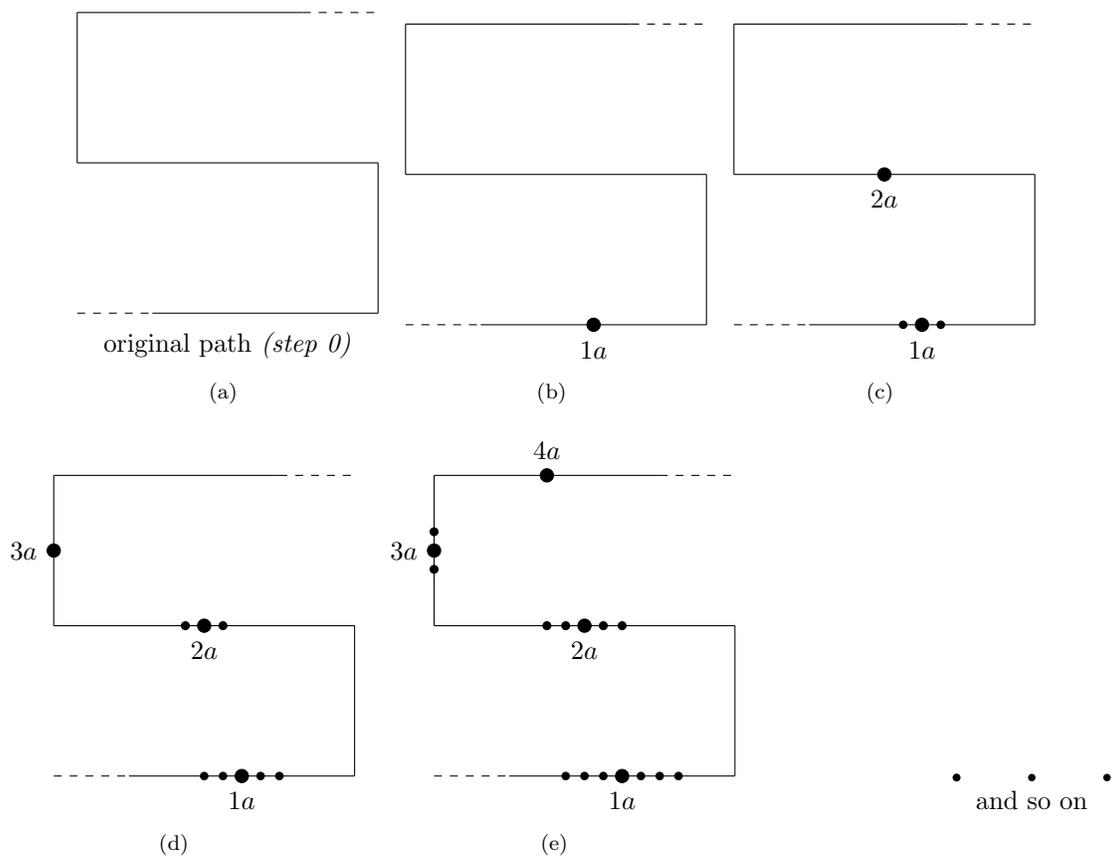
\begin{figure}
    \begin{minipage}{1\textwidth}
        \centering
        \subfigure[]{
            \begin{tikzpicture}
                \draw[dashed] (0,0) -- (1,0);
                \draw(1,0) -- (4,0);
                \draw(4,0) -- (4,2);
                \draw(4,2) -- (0,2);
                \draw(0,2) -- (0,4);
                \draw(0,4) -- (3,4);
                \draw[dashed] (3,4) -- (4,4);
                
                \node [label=below:{original path \textit{(step 0)}}] (BA) at (2,0) {};
            \end{tikzpicture}
        }
        \subfigure[]{
            \begin{tikzpicture}
                \draw[dashed] (0,0) -- (1,0);
                \draw(1,0) -- (4,0);
                \draw(4,0) -- (4,2);
                \draw(4,2) -- (0,2);
                \draw(0,2) -- (0,4);
                \draw(0,4) -- (3,4);
                \draw[dashed] (3,4) -- (4,4);
                
                \node [circle, draw=black, fill=black, inner sep=1pt, minimum size=5pt, label=below:{$1a$}] (BA) at (2.5,0) {};
            \end{tikzpicture}
        }
        \subfigure[]{
            \begin{tikzpicture}
                \draw[dashed] (0,0) -- (1,0);
                \draw(1,0) -- (4,0);
                \draw(4,0) -- (4,2);
                \draw(4,2) -- (0,2);
                \draw(0,2) -- (0,4);
                \draw(0,4) -- (3,4);
                \draw[dashed] (3,4) -- (4,4);
                
                \node [circle, draw=black, fill=black, inner sep=1pt, minimum size=5pt, label=below:{$1a$}] (BA) at (2.5,0) {};
                \node [circle, draw=black, fill=black, inner sep=1pt, minimum size=2pt, label=below:{}] (BA) at (2.25,0) {};
                \node [circle, draw=black, fill=black, inner sep=1pt, minimum size=2pt, label=below:{}] (BA) at (2.75,0) {};
                
                \node [circle, draw=black, fill=black, inner sep=1pt, minimum size=5pt, label=below:{$2a$}] (BA) at (2,2) {};
            \end{tikzpicture}
        }
        \subfigure[]{
            \begin{tikzpicture}
                \draw[dashed] (0,0) -- (1,0);
                \draw(1,0) -- (4,0);
                \draw(4,0) -- (4,2);
                \draw(4,2) -- (0,2);
                \draw(0,2) -- (0,4);
                \draw(0,4) -- (3,4);
                \draw[dashed] (3,4) -- (4,4);
                
                \node [circle, draw=black, fill=black, inner sep=1pt, minimum size=5pt, label=below:{$1a$}] (BA) at (2.5,0) {};
                \node [circle, draw=black, fill=black, inner sep=1pt, minimum size=2pt, label=below:{}] (BA) at (2,0) {};
                \node [circle, draw=black, fill=black, inner sep=1pt, minimum size=2pt, label=below:{}] (BA) at (2.25,0) {};
                \node [circle, draw=black, fill=black, inner sep=1pt, minimum size=2pt, label=below:{}] (BA) at (2.75,0) {};
                \node [circle, draw=black, fill=black, inner sep=1pt, minimum size=2pt, label=below:{}] (BA) at (3,0) {};
                
                \node [circle, draw=black, fill=black, inner sep=1pt, minimum size=5pt, label=below:{$2a$}] (BA) at (2,2) {};
                \node [circle, draw=black, fill=black, inner sep=1pt, minimum size=3pt, label=below:{}] (BA) at (1.75,2) {};
                \node [circle, draw=black, fill=black, inner sep=1pt, minimum size=3pt, label=below:{}] (BA) at (2.25,2) {};
                
                \node [circle, draw=black, fill=black, inner sep=1pt, minimum size=5pt, label=left:{$3a$}] (BA) at (0,3) {};
            \end{tikzpicture}
        }
        \subfigure[]{
            \begin{tikzpicture}
                \draw[dashed] (0,0) -- (1,0);
                \draw(1,0) -- (4,0);
                \draw(4,0) -- (4,2);
                \draw(4,2) -- (0,2);
                \draw(0,2) -- (0,4);
                \draw(0,4) -- (3,4);
                \draw[dashed] (3,4) -- (4,4);
                
                \node [circle, draw=black, fill=black, inner sep=1pt, minimum size=5pt, label=below:{$1a$}] (BA) at (2.5,0) {};
                \node [circle, draw=black, fill=black, inner sep=1pt, minimum size=2pt, label=below:{}] (BA) at (1.75,0) {};
                \node [circle, draw=black, fill=black, inner sep=1pt, minimum size=2pt, label=below:{}] (BA) at (2,0) {};
                \node [circle, draw=black, fill=black, inner sep=1pt, minimum size=2pt, label=below:{}] (BA) at (2.25,0) {};
                \node [circle, draw=black, fill=black, inner sep=1pt, minimum size=2pt, label=below:{}] (BA) at (2.75,0) {};
                \node [circle, draw=black, fill=black, inner sep=1pt, minimum size=2pt, label=below:{}] (BA) at (3,0) {};
                \node [circle, draw=black, fill=black, inner sep=1pt, minimum size=2pt, label=below:{}] (BA) at (3.25,0) {};
                
                \node [circle, draw=black, fill=black, inner sep=1pt, minimum size=5pt, label=below:{$2a$}] (BA) at (2,2) {};
                \node [circle, draw=black, fill=black, inner sep=1pt, minimum size=3pt, label=below:{}] (BA) at (1.5,2) {};
                \node [circle, draw=black, fill=black, inner sep=1pt, minimum size=3pt, label=below:{}] (BA) at (1.75,2) {};
                \node [circle, draw=black, fill=black, inner sep=1pt, minimum size=3pt, label=below:{}] (BA) at (2.25,2) {};
                \node [circle, draw=black, fill=black, inner sep=1pt, minimum size=3pt, label=below:{}] (BA) at (2.5,2) {};
                
                \node [circle, draw=black, fill=black, inner sep=1pt, minimum size=5pt, label=left:{$3a$}] (BA) at (0,3) {};
                \node [circle, draw=black, fill=black, inner sep=1pt, minimum size=3pt, label=left:{}] (BA) at (0,2.75) {};
                \node [circle, draw=black, fill=black, inner sep=1pt, minimum size=3pt, label=left:{}] (BA) at (0,3.25) {};
                
                \node [circle, draw=black, fill=black, inner sep=1pt, minimum size=5pt, label=above:{$4a$}] (BA) at (1.5,4) {};
            \end{tikzpicture}
        }
        \subfigure{
            \begin{tikzpicture}
                \node [circle, inner sep=1pt, minimum size=5pt] (BA) at (-.5,0) {};
                \node [circle, inner sep=1pt, minimum size=4pt] (BA) at (1,0) {};
                \node [circle, fill=black, inner sep=1pt, minimum size=3pt] (BA) at (2,0) {};
                \node [circle, fill=black, inner sep=1pt, minimum size=2pt, label=below:{and so on}] (BA) at (3,0) {};
                \node [circle, fill=black, inner sep=1pt, minimum size=2pt] (BA) at (4,0) {};
            \end{tikzpicture}
        }
    \end{minipage}
    \caption{Burning process on a path of $\infty$ length; the fire sources are placed at $\infty$ distance from each other.}
    \label{figure:burn-path-inf}
\end{figure}

\begin{table}
    \centering
    \begin{tabular}{c|c|c|c|c|l}
          & $C_1$ & $C_2$ & $C_3$ & $C_4$ & $C_5$ $\ \ \ \bullet \bullet \bullet$\\
        \textit{Step 1} & 1 &   &   &   &  \\
        \textit{Step 2} & 2 & 1 &   &   &  \\
        \textit{Step 3} & 2 & 2 & 1 &   &  \\
        \textit{Step 4} & 2 & 2 & 2 & 1 &  \\
        \textit{Step 5} & 2 & 2 & 2 & 2 & 1\\
          &   &   &   &   & $\bullet$\\
          &   &   &   &   & \ \ $\bullet$\\
          &   &   &   &   & \ \ \ \ $\bullet$\\
         Total $25$ & 9 & 7 & 5 & 3 & 1 $\ \ \ \bullet \bullet \bullet$\\
    \end{tabular}
    \caption{Step-wise number of burned vertices in all the clusters $\{C_i\}$ in a path of infinite length, as shown in figure 1. Total burning vertices shown at bottom left.}
    \label{table:burn-path-inf}
\end{table}

For a burning cluster, initially it is born with a single vertex, and $2$ more vertices keep on adding to it in all the subsequent steps.

\subsection{Governing dynamics}

\Cref{equation:burn-path-inf-1} and \Cref{equation:burn-path-inf-2} show the phenomenon of initiating one fire source at a step $t$ along with the addition of two more vertices to the clusters created in the previous steps.

\begin{equation}\label{equation:burn-path-inf-1}
f_0 = 0
\end{equation}
\begin{equation}\label{equation:burn-path-inf-2}
f_t = f_{t-1} + 1 + 2(t - 1) = 2t - 1 + f_{t-1}
\end{equation}

For a path of finite length, in each step, we must burn an unburnt vertex such that the burning clusters do not coincide uptil the maximum number of possible time steps. $f_t$ in \Cref{equation:burn-path-inf-2} presents the maximum of vertices that can be burnt in a given path until after the completion of step $t$.

\Cref{equation:burn-path-inf-1} and \Cref{equation:burn-path-inf-2} imply that $f_t$ is the sum of first $t$ odd numbers, which is equal to $t^2$; $f_t = 1 + 3 + 5 + . . . + (2t-1) = t^2$. So we have that a path of length $n$ should take at least $\big\lceil\sqrt{n}\big\rceil$ time steps to burn completely. While burning a path of finite length, after step $t$ we shall have burnt atmost $t^2$ vertices. \cite{Bonato2016} have shown that it takes $\big\lceil\sqrt{n}\big\rceil$ steps to burn a path of $n$ vertices. They also gave an algorithm to compute an optimal burning sequence for a path of finite length, which we describe as \Cref{algorithm:burn-path-finite}.

\begin{algorithm}\label{algorithm:burn-path-finite}
Given the input $(G$, $P)$ where $G$ is the graph denoting a path, $P$ is a sequence of vertices in the path represented by $G$ from (any) one end to the other, and $P[i]$ stands for $i^{th}$ vertex in $P$, perform the following steps.
\end{algorithm}

\textbf{\textit{Stage 1.}} $n=|P|.\ S=\phi.\ k=\big\lceil\sqrt{n}\big\rceil$.

\textbf{\textbf{Stage 2.}} $\forall\ 0\leq i\leq k-2$, $v = n-i^2-i$; $S=S\cup_{\setminus s}\{P[v]\}$.

\textbf{\textit{Stage 3.}} If $n>(k-1)^2+k$, then $v = n-(k-1)^2-(k-1)$;\\
else, then $v=1$.\\
$S = S \cup_{\setminus s} \{P[v]\}$.

\textbf{\textit{Stage 4.}} Return $S$.\\

Time complexity of \Cref{algorithm:burn-path-finite} is $O(\sqrt{n})$, where $n$ is the number of vertices in a given path. 


The optimal burning process of a path of $9$ vertices is demonstrated in \Cref{figure:burn-path-9}; the burning number of this path is $3$.

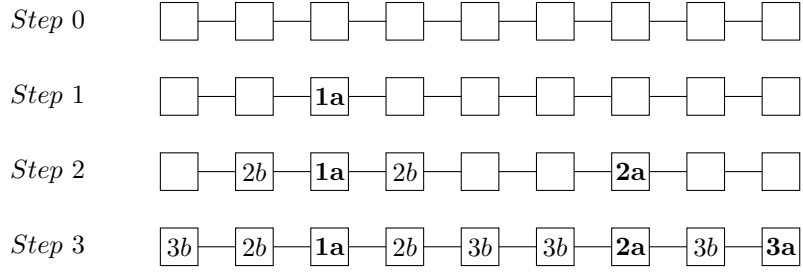
\begin{figure}
	\centering
	\begin{tikzpicture}
	    \draw (-4,0) rectangle (-3.5, 0.5);
	    \draw (-3,0) rectangle (-2.5, 0.5);\draw (-3.5,0.25) -- (-3, 0.25);
	    \draw (-2,0) rectangle (-1.5, 0.5);\draw (-2.5,0.25) -- (-2, 0.25);
	    \draw (-1,0) rectangle (-0.5, 0.5);\draw (-1.5,0.25) -- (-1, 0.25);
	    \draw (0,0) rectangle (0.5, 0.5);\draw (-0.5,0.25) -- (0, 0.25);
    	\draw (1,0) rectangle (1.5, 0.5);\draw (0.5,0.25) -- (1, 0.25);
	    \draw (2,0) rectangle (2.5, 0.5);\draw (1.5,0.25) -- (2, 0.25);
		\draw (3,0) rectangle (3.5, 0.5);\draw (2.5,0.25) -- (3, 0.25);
		\draw (4,0) rectangle (4.5, 0.5);\draw (3.5,0.25) -- (4, 0.25);
		
		\node at (-5.5,0.25) {$Step\ 3$};
    	\node at (-3.75,0.25) {$3b$};
    	\node at (-2.75,0.25) {$2b$};
		\node at (-1.75,0.25) {$\textbf{1a}$};
		\node at (-0.75,0.25) {$2b$};
		\node at (0.25,0.25) {$3b$};
    	\node at (1.25,0.25) {$3b$};
		\node at (2.25,0.25) {$\textbf{2a}$};
		\node at (3.25,0.25) {$3b$};
		\node at (4.25,0.25) {$\textbf{3a}$};
    	
		\draw (-4,1) rectangle (-3.5, 1.5);
		\draw (-3,1) rectangle (-2.5, 1.5);\draw (-3.5,1.25) -- (-3, 1.25);
    	\draw (-2,1) rectangle (-1.5, 1.5);\draw (-2.5,1.25) -- (-2, 1.25);
    	\draw (-1,1) rectangle (-0.5, 1.5);\draw (-1.5,1.25) -- (-1, 1.25);
    	\draw (0,1) rectangle (0.5, 1.5);\draw (-0.5,1.25) -- (0, 1.25);
    	\draw (1,1) rectangle (1.5, 1.5);\draw (0.5,1.25) -- (1, 1.25);
    	\draw (2,1) rectangle (2.5, 1.5);\draw (1.5,1.25) -- (2, 1.25);
		\draw (3,1) rectangle (3.5, 1.5);\draw (2.5,1.25) -- (3, 1.25);
		\draw (4,1) rectangle (4.5, 1.5);\draw (3.5,1.25) -- (4, 1.25);
		
		\node at (-5.5,1.25) {$Step\ 2$};
		\node at (-2.75,1.25) {$2b$};
		\node at (-1.75,1.25) {$\textbf{1a}$};
    	\node at (-0.75,1.25) {$2b$};
   		\node at (2.25,1.25) {$\textbf{2a}$};
		
		\draw (-4,2) rectangle (-3.5, 2.5);
		\draw (-3,2) rectangle (-2.5, 2.5);\draw (-3.5,2.25) -- (-3, 2.25);
		\draw (-2,2) rectangle (-1.5, 2.5);\draw (-2.5,2.25) -- (-2, 2.25);
		\draw (-1,2) rectangle (-0.5, 2.5);\draw (-1.5,2.25) -- (-1, 2.25);
		\draw (0,2) rectangle (0.5, 2.5);\draw (-0.5,2.25) -- (0, 2.25);
		\draw (1,2) rectangle (1.5, 2.5);\draw (0.5,2.25) -- (1, 2.25);
		\draw (2,2) rectangle (2.5, 2.5);\draw (1.5,2.25) -- (2, 2.25);
		\draw (3,2) rectangle (3.5, 2.5);\draw (2.5,2.25) -- (3, 2.25);
		\draw (4,2) rectangle (4.5, 2.5);\draw (3.5,2.25) -- (4, 2.25);
		
		\node at (-5.5,2.25) {$Step\ 1$};
		\node at (-1.75,2.25) {$\textbf{1a}$};
		
		\draw (-4,3) rectangle (-3.5, 3.5);
		\draw (-3,3) rectangle (-2.5, 3.5);\draw (-3.5,3.25) -- (-3, 3.25);
		\draw (-2,3) rectangle (-1.5, 3.5);\draw (-2.5,3.25) -- (-2, 3.25);
		\draw (-1,3) rectangle (-0.5, 3.5);\draw (-1.5,3.25) -- (-1, 3.25);
		\draw (0,3) rectangle (0.5, 3.5);\draw (-0.5,3.25) -- (0, 3.25);
    	\draw (1,3) rectangle (1.5, 3.5);\draw (0.5,3.25) -- (1, 3.25);
    	\draw (2,3) rectangle (2.5, 3.5);\draw (1.5,3.25) -- (2, 3.25);
    	\draw (3,3) rectangle (3.5, 3.5);\draw (2.5,3.25) -- (3, 3.25);
    	\draw (4,3) rectangle (4.5, 3.5);\draw (3.5,3.25) -- (4, 3.25);
    	    
	    \node at (-5.5,3.25) {$Step\ 0$};
    \end{tikzpicture}
	\caption{Burning of a path of $9$ vertices (optimal procedure). The vertices are marked by the step number in which they are burned. For each step, the vertex which is selected (arbitrarily) as a victim to be burned (the fire source) is marked by \textbf{\textit{a}}, and those which are burned by the fire spread from already burned vertices are marked by \textbf{\textit{b}}.}
	\label{figure:burn-path-9}
\end{figure}

\section{Optimal burning of some general graphs}\label{section:burn-general-graphs-optimally}

We start with defining isometric subgraph. An isometric subgraph of a graph is defined as follows.

\textbf{Isometric subgraph}\index{isometric subgraph}: A subgraph $H$ of a graph $G$ is called an \textit{isometric subgraph} if for every pair of nodes $u$ and $v$ in $H$, we have that the shortest distance between them in $H$ is equal to the shortest distance between them in $G$. If $G$ is a tree, and then we have that $H$ is an \textit{isometric tree}\index{isometric subtree} of $G$ is the same conditions satisfy.

It can be easily observed that any connected subgraph of a tree is its isometric subtree. Bonato et al. \cite{Bonato2016} have showed that if the burning number of a graph $G$ is $b(G)=k^\prime$, then the burning number of $H$, an isometric subtree of $G$ is $b(H)\leq b(G)=k^\prime$. We show this in \Cref{theorem:burn-isomorphic-subtree}.

\begin{theorem}\label{theorem:burn-isomorphic-subtree}
    Let that the burning number of a graph $G$ is $b(G)=k^\prime$, then the burning number of $H$, an isometric subtree of $G$ is $b(H)\leq b(G)=k^\prime$.
\end{theorem}

\begin{proof}
    Let for contradiction that $b(G)<b(H)$. Let the burning number of $H$ be $k^\prime$. From our assumption, we have that $b(G)\leq k^\prime-1$. This implies that there exists a burning sequence $S^\prime=(y_1,y_2,...,y_{k^\prime-1})$ which is able to burn $G$ completely.
    
    $H$ cannot be burned by any burning sequence of length less than $k^\prime$. It means that for any burning sequence of length less than $k^\prime$, at least one vertex in $H$ remains unburned. Since $H$ is an isometric subtree of $G$, we have that for any burning sequence of length less than $k^\prime$, at least one vertex in $G$ also remains unburned. So we need at least one more fire source in $S^\prime$ to burn $G$ which is a contradiction to our assumption.
\end{proof}

Similarly, we can obtain a proof that for any graph $G$, the burning number of an isometric subgraph of $G$ is not more then its own burning number $b(G)$. \Cref{theorem:burn-connected-graph-radius-r} and \Cref{theorem:burn-k-component-radius-r} we proved in \cite{Bonato2016}.

\begin{theorem}\label{theorem:burn-connected-graph-radius-r}
    If a connected graph $G$ is of radius $r$, then $b(G)\leq r+1$.
\end{theorem}

\begin{proof}
    Let $v$ be the vertex from which the distance of each vertex in $G.V$ is atmost $r$. Let that we have put a fire source on $v$ in the first time step. $\because G.N_{r}[v]=G.V$, $G$ will eventually get burnt within $r$ more time-steps irrespective of the position of the fire sources that are put in all other time-steps.
\end{proof}

\begin{theorem}\label{theorem:burn-k-component-radius-r}
    Let that $G$ is a graph containing $k$ connected components, each having a radius atmost $r$, then $b(G)\leq r+k$.
\end{theorem}

\begin{proof}
    Let there be $k$ components $C_1,C_2,\dots,C_k$ in a graph $G$, each of maximum radius $r$. Let $v_i$ be a vertex in $C_i$, $\forall\ 1\leq i\leq k$, such that it is atmost at a distance $r$ from all the other vertices in $C_i$.
    
    Let that in each time-step we put a fire source such that in step $i$, we put a fire source on $v_i$. Within next $r$ time-steps $G$ will be burnt.
\end{proof}

Next we discuss the burning procedure on some examples of spider graphs\index{burning spider graphs}. $SP(3,4)$ can be burnt optimally in 4 steps. As demonstrated in \Cref{figure:burn-SP34-optimal}, an optimal burning sequence of the presented graph $SP(3,\ 4)$ is $S^\prime_{SP(3,4)}=(b,\ h,\ l,\ m)$, which is of size $4$.

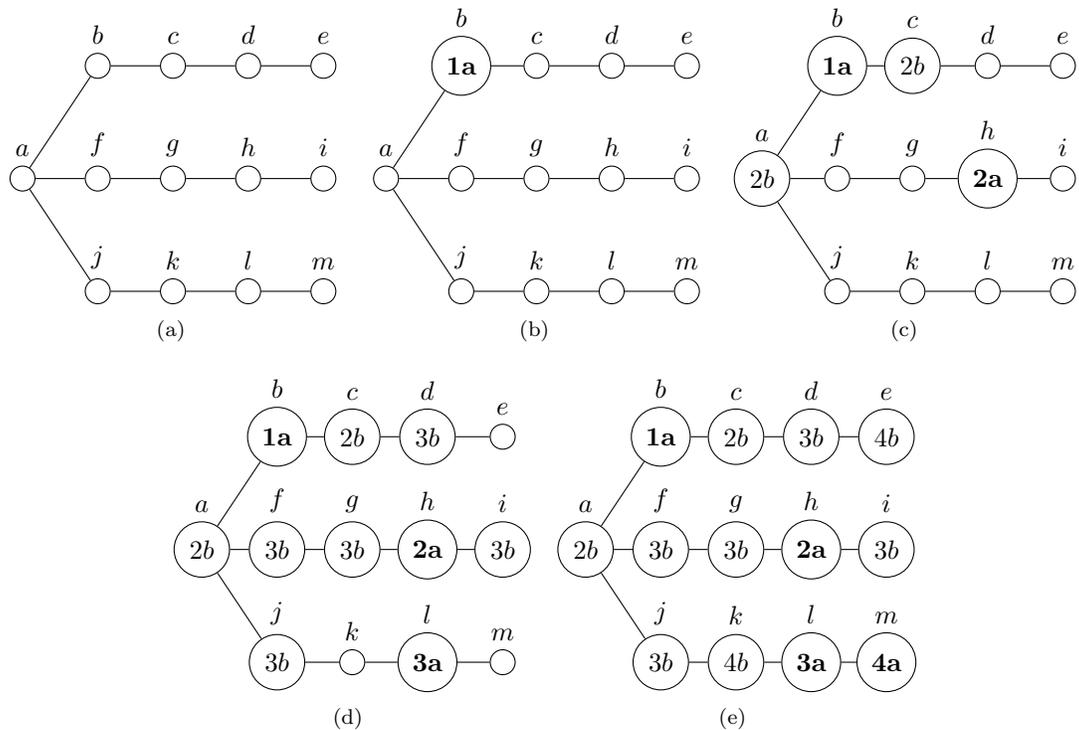
\begin{figure}
    \centering
    \subfigure[]{
        \begin{tikzpicture}
            \node[draw, shape=circle, label=above:{$a$}] (AA) at (0,0) {};
            
            \node[draw, shape=circle, label=above:{$b$}] (AB) at (1, 1.5) {};
            \node[draw, shape=circle, label=above:{$c$}] (AC) at (2, 1.5) {};
            \node[draw, shape=circle, label=above:{$d$}] (AD) at (3, 1.5) {};
            \node[draw, shape=circle, label=above:{$e$}] (AE) at (4, 1.5) {};
                
            \draw (AA) -- (AB);
            \draw (AB) -- (AC);
            \draw (AC) -- (AD);
            \draw (AD) -- (AE);
            
            \node[draw, shape=circle, label=above:{$f$}] (AB) at (1, 0) {};
            \node[draw, shape=circle, label=above:{$g$}] (AC) at (2, 0) {};
            \node[draw, shape=circle, label=above:{$h$}] (AD) at (3, 0) {};
            \node[draw, shape=circle, label=above:{$i$}] (AE) at (4, 0) {};
                
            \draw (AA) -- (AB);
            \draw (AB) -- (AC);
            \draw (AC) -- (AD);
            \draw (AD) -- (AE);
            
            \node[draw, shape=circle, label=above:{$j$}] (AB) at (1, -1.5) {};
            \node[draw, shape=circle, label=above:{$k$}] (AC) at (2, -1.5) {};
            \node[draw, shape=circle, label=above:{$l$}] (AD) at (3, -1.5) {};
            \node[draw, shape=circle, label=above:{$m$}] (AE) at (4, -1.5) {};
                
            \draw (AA) -- (AB);
            \draw (AB) -- (AC);
            \draw (AC) -- (AD);
            \draw (AD) -- (AE);
        \end{tikzpicture}
    }
    \subfigure[]{
        \begin{tikzpicture}
            \node[draw, shape=circle, label=above:{$a$}] (AA) at (0,0) {};
            
            \node[draw, shape=circle, label=above:{$b$}] (AB) at (1, 1.5) {\textbf{1a}};
            \node[draw, shape=circle, label=above:{$c$}] (AC) at (2, 1.5) {};
            \node[draw, shape=circle, label=above:{$d$}] (AD) at (3, 1.5) {};
            \node[draw, shape=circle, label=above:{$e$}] (AE) at (4, 1.5) {};
                
            \draw (AA) -- (AB);
            \draw (AB) -- (AC);
            \draw (AC) -- (AD);
            \draw (AD) -- (AE);
            
            \node[draw, shape=circle, label=above:{$f$}] (AB) at (1, 0) {};
            \node[draw, shape=circle, label=above:{$g$}] (AC) at (2, 0) {};
            \node[draw, shape=circle, label=above:{$h$}] (AD) at (3, 0) {};
            \node[draw, shape=circle, label=above:{$i$}] (AE) at (4, 0) {};
                
            \draw (AA) -- (AB);
            \draw (AB) -- (AC);
            \draw (AC) -- (AD);
            \draw (AD) -- (AE);
            
            \node[draw, shape=circle, label=above:{$j$}] (AB) at (1, -1.5) {};
            \node[draw, shape=circle, label=above:{$k$}] (AC) at (2, -1.5) {};
            \node[draw, shape=circle, label=above:{$l$}] (AD) at (3, -1.5) {};
            \node[draw, shape=circle, label=above:{$m$}] (AE) at (4, -1.5) {};
                
            \draw (AA) -- (AB);
            \draw (AB) -- (AC);
            \draw (AC) -- (AD);
            \draw (AD) -- (AE);
        \end{tikzpicture}
    }
    \subfigure[]{
        \begin{tikzpicture}
            \node[draw, shape=circle, label=above:{$a$}] (AA) at (0,0) {$2b$};
            
            \node[draw, shape=circle, label=above:{$b$}] (AB) at (1, 1.5) {\textbf{1a}};
            \node[draw, shape=circle, label=above:{$c$}] (AC) at (2, 1.5) {$2b$};
            \node[draw, shape=circle, label=above:{$d$}] (AD) at (3, 1.5) {};
            \node[draw, shape=circle, label=above:{$e$}] (AE) at (4, 1.5) {};
                
            \draw (AA) -- (AB);
            \draw (AB) -- (AC);
            \draw (AC) -- (AD);
            \draw (AD) -- (AE);
            
            \node[draw, shape=circle, label=above:{$f$}] (AB) at (1, 0) {};
            \node[draw, shape=circle, label=above:{$g$}] (AC) at (2, 0) {};
            \node[draw, shape=circle, label=above:{$h$}] (AD) at (3, 0) {\textbf{2a}};
            \node[draw, shape=circle, label=above:{$i$}] (AE) at (4, 0) {};
                
            \draw (AA) -- (AB);
            \draw (AB) -- (AC);
            \draw (AC) -- (AD);
            \draw (AD) -- (AE);
            
            \node[draw, shape=circle, label=above:{$j$}] (AB) at (1, -1.5) {};
            \node[draw, shape=circle, label=above:{$k$}] (AC) at (2, -1.5) {};
            \node[draw, shape=circle, label=above:{$l$}] (AD) at (3, -1.5) {};
            \node[draw, shape=circle, label=above:{$m$}] (AE) at (4, -1.5) {};
                
            \draw (AA) -- (AB);
            \draw (AB) -- (AC);
            \draw (AC) -- (AD);
            \draw (AD) -- (AE);
        \end{tikzpicture}
    }
    \subfigure[]{
        \begin{tikzpicture}
            \node[draw, shape=circle, label=above:{$a$}] (AA) at (0,0) {$2b$};
            
            \node[draw, shape=circle, label=above:{$b$}] (AB) at (1, 1.5) {\textbf{1a}};
            \node[draw, shape=circle, label=above:{$c$}] (AC) at (2, 1.5) {$2b$};
            \node[draw, shape=circle, label=above:{$d$}] (AD) at (3, 1.5) {$3b$};
            \node[draw, shape=circle, label=above:{$e$}] (AE) at (4, 1.5) {};
                
            \draw (AA) -- (AB);
            \draw (AB) -- (AC);
            \draw (AC) -- (AD);
            \draw (AD) -- (AE);
            
            \node[draw, shape=circle, label=above:{$f$}] (AB) at (1, 0) {$3b$};
            \node[draw, shape=circle, label=above:{$g$}] (AC) at (2, 0) {$3b$};
            \node[draw, shape=circle, label=above:{$h$}] (AD) at (3, 0) {\textbf{2a}};
            \node[draw, shape=circle, label=above:{$i$}] (AE) at (4, 0) {$3b$};
                
            \draw (AA) -- (AB);
            \draw (AB) -- (AC);
            \draw (AC) -- (AD);
            \draw (AD) -- (AE);
            
            \node[draw, shape=circle, label=above:{$j$}] (AB) at (1, -1.5) {$3b$};
            \node[draw, shape=circle, label=above:{$k$}] (AC) at (2, -1.5) {};
            \node[draw, shape=circle, label=above:{$l$}] (AD) at (3, -1.5) {\textbf{3a}};
            \node[draw, shape=circle, label=above:{$m$}] (AE) at (4, -1.5) {};
                
            \draw (AA) -- (AB);
            \draw (AB) -- (AC);
            \draw (AC) -- (AD);
            \draw (AD) -- (AE);
        \end{tikzpicture}
    }
    \subfigure[]{
        \begin{tikzpicture}
            \node[draw, shape=circle, label=above:{$a$}] (AA) at (0,0) {$2b$};
            
            \node[draw, shape=circle, label=above:{$b$}] (AB) at (1, 1.5) {\textbf{1a}};
            \node[draw, shape=circle, label=above:{$c$}] (AC) at (2, 1.5) {$2b$};
            \node[draw, shape=circle, label=above:{$d$}] (AD) at (3, 1.5) {$3b$};
            \node[draw, shape=circle, label=above:{$e$}] (AE) at (4, 1.5) {$4b$};
                
            \draw (AA) -- (AB);
            \draw (AB) -- (AC);
            \draw (AC) -- (AD);
            \draw (AD) -- (AE);
            
            \node[draw, shape=circle, label=above:{$f$}] (AB) at (1, 0) {$3b$};
            \node[draw, shape=circle, label=above:{$g$}] (AC) at (2, 0) {$3b$};
            \node[draw, shape=circle, label=above:{$h$}] (AD) at (3, 0) {\textbf{2a}};
            \node[draw, shape=circle, label=above:{$i$}] (AE) at (4, 0) {$3b$};
                
            \draw (AA) -- (AB);
            \draw (AB) -- (AC);
            \draw (AC) -- (AD);
            \draw (AD) -- (AE);
            
            \node[draw, shape=circle, label=above:{$j$}] (AB) at (1, -1.5) {$3b$};
            \node[draw, shape=circle, label=above:{$k$}] (AC) at (2, -1.5) {$4b$};
            \node[draw, shape=circle, label=above:{$l$}] (AD) at (3, -1.5) {\textbf{3a}};
            \node[draw, shape=circle, label=above:{$m$}] (AE) at (4, -1.5) {\textbf{4a}};
                
            \draw (AA) -- (AB);
            \draw (AB) -- (AC);
            \draw (AC) -- (AD);
            \draw (AD) -- (AE);
        \end{tikzpicture}
    }
    \caption{Optimal burning of $SP(3,\ 4)$.}
    \label{figure:burn-SP34-optimal}
\end{figure}

If we increase one more arm in $SP(3,\ 4)$, it becomes $SP(4,\ 4)$, and now we need at least $5$ time steps to burn it; a burnt $SP(4,\ 4)$ is shown in figure \Cref{figure:burn-SP44-optimal}. An optimal burning sequence of the presented graph $SP(4,4)$ is $S^\prime_{SP(4,4)}=(b,\ h,\ l,\ m,\ q)$.

\begin{figure}
    \centering
    \begin{tikzpicture}
        \node[draw, shape=circle, label=above:{$a$}] (AA) at (0,0) {$2b$};
        
        \node[draw, shape=circle, label=above:{$b$}] (AB) at (1, .75) {\textbf{1a}};
        \node[draw, shape=circle, label=above:{$c$}] (AC) at (2, .75) {$2b$};
        \node[draw, shape=circle, label=above:{$d$}] (AD) at (3, .75) {$3b$};
        \node[draw, shape=circle, label=above:{$e$}] (AE) at (4, .75) {$4b$};
            
        \draw (AA) -- (AB);
        \draw (AB) -- (AC);
        \draw (AC) -- (AD);
        \draw (AD) -- (AE);
        
        \node[draw, shape=circle, label=above:{$f$}] (AB) at (1, -.75) {$3b$};
        \node[draw, shape=circle, label=above:{$g$}] (AC) at (2, -.75) {$3b$};
        \node[draw, shape=circle, label=above:{$h$}] (AD) at (3, -.75) {\textbf{2a}};
        \node[draw, shape=circle, label=above:{$i$}] (AE) at (4, -.75) {$3b$};
            
        \draw (AA) -- (AB);
        \draw (AB) -- (AC);
        \draw (AC) -- (AD);
        \draw (AD) -- (AE);
        
        \node[draw, shape=circle, label=above:{$j$}] (AB) at (-1, .75) {$3b$};
        \node[draw, shape=circle, label=above:{$k$}] (AC) at (-2, .75) {$4b$};
        \node[draw, shape=circle, label=above:{$l$}] (AD) at (-3, .75) {\textbf{3a}};
        \node[draw, shape=circle, label=above:{$m$}] (AE) at (-4, .75) {\textbf{4a}};
        
        \draw (AA) -- (AB);
        \draw (AB) -- (AC);
        \draw (AC) -- (AD);
        \draw (AD) -- (AE);
            
        \node[draw, shape=circle, label=above:{$n$}] (AB) at (-1, -.75) {$3b$};
        \node[draw, shape=circle, label=above:{$o$}] (AC) at (-2, -.75) {$4b$};
        \node[draw, shape=circle, label=above:{$p$}] (AD) at (-3, -.75) {$5b$};
        \node[draw, shape=circle, label=above:{$q$}] (AE) at (-4, -.75) {\textbf{5a}};
        
        \draw (AA) -- (AB);
        \draw (AB) -- (AC);
        \draw (AC) -- (AD);
        \draw (AD) -- (AE);
    \end{tikzpicture}
    \caption{Optimal burning of $SP(4,\ 4)$}
    \label{figure:burn-SP44-optimal}
\end{figure}
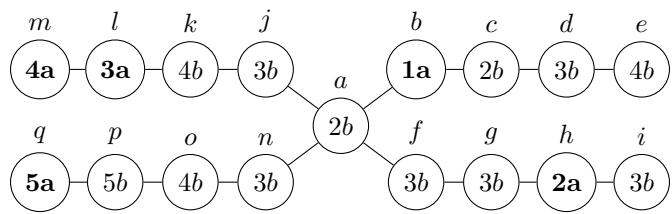

Bessy et al \cite{Bessy2017} have proved that for a graph $SP(s,\ r)$ if $s \geq r$, then $b(SP(s,\ r))$ $=$ $r+1$, and if $s \geq r+2$, then it is necessary for any optimal burning sequence to have first fire source on the spider head. We show this as \Cref{lemma:b(SP)}. \Cref{figure:burn-1-component} is an example of this phenomenon; it shows burning of $SP(d,\ k^\prime-1)$, where $d \geq k^\prime+1$. In its optimal burning sequence $S^\prime_{SP(d,k^\prime-1)} = (y_1, y_2, y_3, . . ., y_{k^\prime})$, the first fire source $y_1$ is placed on the spider head.

\begin{lemma}\label{lemma:b(SP)}
    For a graph $SP(s,\ r)$ if $s \geq r$, then $b(SP(s,\ r))$ $=$ $r+1$, and if $s \geq r+2$, then it is necessary for any optimal burning sequence to have first fire source on the spider head.
\end{lemma}

\begin{proof}
    We start with discussing a spider graph $SP(r,r)$ whose number of arms $r$ is equal to the length of each arm. The radius of $SP(r,r)$ is $r$. So by \Cref{theorem:burn-connected-graph-radius-r} we have that $b(SP(r,r))\leq r+1$.
    
    Now we show that $b(SP(r,r))\geq r+1$.
    Let for contradiction that $b(SP(r,r))\leq r$. This implies that there is a w-burning sequence $S^\prime=(x_1,x_2,...,x_{k^\prime})$ such that $k^\prime\leq r$ which is able to burn $SP(s,r)$. No fire source in $S^\prime$ shall be able to burn more than one leaf node in $G$ because the distance between any two leaf nodes is $2r$. We will be able to burn atmost $k^\prime \leq r$ leaf nodes in $k^\prime$ rounds. This implies that at least $1$ leaf node will be left unburned. This implies that $S^\prime$ must contain at least one more fire source which is a contradiction.
    
    Thus, $b(SP(r,r))=r+1$.
    
    Given that $s \geq r+2$, by \Cref{theorem:burn-isomorphic-subtree} $b(SP(s,\ r))\geq r+1$. Also since the radius of $SP(s,r)$ is of radius $r$, we have that $SP(s,r)\leq r+1$. So we have that $SP(s,r)=r+1$.
    
    Next we show that if $s \geq r+2$, then it is necessary for any optimal burning sequence to have first fire source on the spider head. Let that an optimal burning sequence $S^\prime=(y_1,y_2,...,y_{r+1})$ of length $r+1$ is able to burn $SP(s,r)$. Let for contradiction that the first fire source is not placed at the spider head. So we have that at least $r+1$ leaf nodes cannot be burned by $y_1$. Also, no other fire source will be able to burn two leaf nodes simultaneously because the distance between any two leaf nodes is $2r$. So we have that the fire sources $y_2,...,y_{r+1}$ will be able to burn at most $r$ leaf nodes. We have that at least one leaf node will remain unburned which is a contradiction. So the first fire source $y_1\in S^\prime$ must be placed at the spider head. This concludes the proof.
\end{proof}

\begin{figure}
	\centering
	\begin{tikzpicture}
	    \node [circle, fill=black, inner sep=0pt, minimum size=3pt, label=left:{$y_1$}] (AA) at (0,0) {};
	    
	    \node [circle, fill=black, inner sep=0pt, minimum size=3pt] at (0,.33) {};
	    \node [circle, fill=black, inner sep=0pt, minimum size=3pt] at (.236,.236) {};
	    \node [circle, fill=black, inner sep=0pt, minimum size=3pt] at (.33,0) {};
	    \node [circle, fill=black, inner sep=0pt, minimum size=3pt] at (.236,-.236) {};
	    \node [circle, fill=black, inner sep=0pt, minimum size=3pt] at (0,-.33) {};
	    \node [circle, fill=black, inner sep=0pt, minimum size=3pt] at (-.236,-.236) {};
	    
	    \node [circle, fill=black, inner sep=0pt, minimum size=3pt] at (0,.66) {};
	    \node [circle, fill=black, inner sep=0pt, minimum size=3pt] at (.472,.472) {};
	    \node [circle, fill=black, inner sep=0pt, minimum size=3pt] at (.66,0) {};
	    \node [circle, fill=black, inner sep=0pt, minimum size=3pt] at (.472,-.472) {};
	    \node [circle, fill=black, inner sep=0pt, minimum size=3pt] at (0,-.66) {};
	    \node [circle, fill=black, inner sep=0pt, minimum size=3pt] at (-.472,-.472) {};
	    
	    \node [circle, fill=black, inner sep=0pt, minimum size=3pt] (AB) at (0,1) {};
	    \node [circle, fill=black, inner sep=0pt, minimum size=3pt] (AC) at (.707,.707) {};
	    \node [circle, fill=black, inner sep=0pt, minimum size=3pt] (AD) at (1,0) {};
	    \node [circle, fill=black, inner sep=0pt, minimum size=3pt] (AE) at (.707,-.707) {};
	    \node [circle, fill=black, inner sep=0pt, minimum size=3pt] (AF) at (0,-1) {};
	    \node [circle, fill=black, inner sep=0pt, minimum size=3pt] (AG) at (-.707,-.707) {};
	    
	    \node [circle, fill=black, inner sep=0pt, minimum size=2pt] at (0,1.33) {};
	    \node [circle, fill=black, inner sep=0pt, minimum size=2pt] at (.707+.236,.707+.236) {};
	    \node [circle, fill=black, inner sep=0pt, minimum size=2pt] at (1.33,0) {};
	    \node [circle, fill=black, inner sep=0pt, minimum size=2pt] at (.707+.236,-.707-.236) {};
	    \node [circle, fill=black, inner sep=0pt, minimum size=2pt] at (0,-1.33) {};
	    \node [circle, fill=black, inner sep=0pt, minimum size=2pt] at (-.707-.236,-.707-.236) {};
	    
	    \node [circle, fill=black, inner sep=0pt, minimum size=2pt] at (0,1.66) {};
	    \node [circle, fill=black, inner sep=0pt, minimum size=2pt] at (.707+.472,.707+.472) {};
	    \node [circle, fill=black, inner sep=0pt, minimum size=2pt] at (1.66,0) {};
	    \node [circle, fill=black, inner sep=0pt, minimum size=2pt] at (.707+.472,-.707-.472) {};
	    \node [circle, fill=black, inner sep=0pt, minimum size=2pt] at (0,-1.66) {};
	    \node [circle, fill=black, inner sep=0pt, minimum size=2pt] at (-.707-.472,-.707-.472) {};
	    
	    \node [circle, fill=black, inner sep=0pt, minimum size=2pt, label=above:{$y_{k^{\prime}}$}] at (0,2) {};
	    \node [circle, fill=black, inner sep=0pt, minimum size=2pt] at (1.411,1.414) {};
	    \node [circle, fill=black, inner sep=0pt, minimum size=2pt] at (2,0) {};
	    \node [circle, fill=black, inner sep=0pt, minimum size=2pt, label=right:{$y_2$}] at (1.414,-1.414) {};
	    \node [circle, fill=black, inner sep=0pt, minimum size=2pt, label=below:{$y_3$}] at (0,-2) {};
	    \node [circle, fill=black, inner sep=0pt, minimum size=2pt, label=left:{$y_4$}] at (-1.414,-1.414) {};
	    
	    \node [circle, fill=black, inner sep=0pt, minimum size=3pt] at (-1,0) {};
	    \node [circle, fill=black, inner sep=0pt, minimum size=3pt] at (-1,-0.2) {};
	    \node [circle, fill=black, inner sep=0pt, minimum size=3pt] at (-1,0.2) {};
	    
		\draw (AA) -- (AB);
		\draw (AA) -- (AC);
		\draw (AA) -- (AD);
		\draw (AA) -- (AE);
		\draw (AA) -- (AF);
		\draw (AA) -- (AG);
		
	\end{tikzpicture}
	\caption{Optimal burning of a graph $SP(d,\ k^\prime-1)$, where $d \geq k^\prime+1$. $S^\prime_{SP(d,k^\prime-1)} = (y_1, y_2, . . ., y_{k^{\prime}})$ is the optimal burning sequence.}
	\label{figure:burn-1-component}
\end{figure}
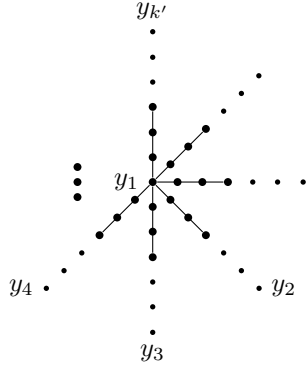

\Cref{figure:burn-1-component} shows burning of a graph containing only one connected component, whereas \Cref{figure:burn-many-components} represents the burning of a graph $G_{SP}$ which is a disjoint union of $k^{\prime}$ components\index{burning (example) graph containing more than one components}. Each component $comp_i: 1\leq i\leq k^\prime$ is a spider graph where $comp_i = SP(d_i, i-1),\ d_i \geq i+1$. This implies the graph $G_{SP} = SP(d_1,\ 0)$ $\cup$ $SP(d_2,\ 1)$ $\cup$ $SP(d_3,\ 2)$ $\cup$ $. . .$ $\cup$ $SP(d_{k^\prime},\ k^\prime-1)$.

The burning number of both the graphs $SP(d,\ k^\prime-1)$ and $G_{SP}$ is $k^{\prime}$. It can be easily observed \cite{Bonato2016} that if an arbitrary graph $G^\prime$ contains $k$ connected components, then the burning number of this graph, $b(G^\prime) \geq k$.

\begin{figure}
	\begin{minipage}{1\textwidth}
		\centering
		\begin{tikzpicture}
		    \node [circle, fill=black, inner sep=0pt, minimum size=3pt, label=left:{$y_{k^{\prime}}$}] at (-2,0) {};

		    \node [circle, fill=black, inner sep=0pt, minimum size=3pt, label=left:{$y_{k^{\prime}-1}$}] (AA) at (0,0) {};
		    \node [circle, fill=black, inner sep=0pt, minimum size=3pt] (AB) at (0,1) {};
		    \node [circle, fill=black, inner sep=0pt, minimum size=3pt] (AC) at (0.707,0.707) {};
		    \node [circle, fill=black, inner sep=0pt, minimum size=3pt] (AD) at (1,0) {};
		    \node [circle, fill=black, inner sep=0pt, minimum size=3pt] (AE) at (0.707,-0.707) {};
		    \node [circle, fill=black, inner sep=0pt, minimum size=3pt] (AF) at (0,-1) {};
		    \node [circle, fill=black, inner sep=0pt, minimum size=3pt] (AG) at (-0.707,-0.707) {};
		    
		    \node [circle, fill=black, inner sep=0pt, minimum size=2pt] at (-1.1,0) {};
		    \node [circle, fill=black, inner sep=0pt, minimum size=2pt] at (-1.1,-0.2) {};
		    \node [circle, fill=black, inner sep=0pt, minimum size=2pt] at (-1.1,0.2) {};
		    
			\draw (AA) -- (AB);
			\draw (AA) -- (AC);
			\draw (AA) -- (AD);
			\draw (AA) -- (AE);
			\draw (AA) -- (AF);
			\draw (AA) -- (AG);

		    \node [circle, fill=black, inner sep=0pt, minimum size=3pt, label=left:{$y_{k^{\prime}-2}$}] (BA) at (3,0) {};
		    \node [circle, fill=black, inner sep=0pt, minimum size=3pt] (BB) at (3,0.5) {};
		    \node [circle, fill=black, inner sep=0pt, minimum size=3pt] (BC) at (3,1) {};
		    \node [circle, fill=black, inner sep=0pt, minimum size=3pt] (BD) at (3.3535,0.3535) {};
		    \node [circle, fill=black, inner sep=0pt, minimum size=3pt] (BE) at (3.707,0.707) {};
		    \node [circle, fill=black, inner sep=0pt, minimum size=3pt] (BF) at (3.5,0) {};
		    \node [circle, fill=black, inner sep=0pt, minimum size=3pt] (BG) at (3+1,0) {};
		    \node [circle, fill=black, inner sep=0pt, minimum size=3pt] (BH) at (3.3535,-0.3535) {};
		    \node [circle, fill=black, inner sep=0pt, minimum size=3pt] (BI) at (3.707,-0.707) {};
		    \node [circle, fill=black, inner sep=0pt, minimum size=3pt] (BJ) at (3,-0.5) {};
		    \node [circle, fill=black, inner sep=0pt, minimum size=3pt] (BK) at (3,-1) {};
		    \node [circle, fill=black, inner sep=0pt, minimum size=3pt] (BL) at (3-0.3535,-0.3535) {};
		    \node [circle, fill=black, inner sep=0pt, minimum size=3pt] (BM) at (3-0.707,-0.707) {};
		    
		    \node [circle, fill=black, inner sep=0pt, minimum size=2pt] at (3-1.1,0) {};
		    \node [circle, fill=black, inner sep=0pt, minimum size=2pt] at (3-1.1,-0.2) {};
		    \node [circle, fill=black, inner sep=0pt, minimum size=2pt] at (3-1.1,0.2) {};
		    
			\draw (BA) -- (BC);
			\draw (BA) -- (BE);
			\draw (BA) -- (BG);
			\draw (BA) -- (BI);
			\draw (BA) -- (BK);
			\draw (BA) -- (BM);

			\node [circle, fill=black, inner sep=0pt, minimum size=2pt] at (5,-.5) {};
		    \node [circle, fill=black, inner sep=0pt, minimum size=2pt] at (5.2,-.5) {};
		    \node [circle, fill=black, inner sep=0pt, minimum size=2pt] at (5-0.2,-.5) {};

		    \node [circle, fill=black, inner sep=0pt, minimum size=3pt, label=left:{$y_1$}] (CA) at (7,0) {};
		    
		    \node [circle, fill=black, inner sep=0pt, minimum size=2pt] at (7-1,0) {};
		    \node [circle, fill=black, inner sep=0pt, minimum size=2pt] at (7-1,-0.2) {};
		    \node [circle, fill=black, inner sep=0pt, minimum size=2pt] at (7-1,0.2) {};
		    
			\draw (CA) -- (7,1);
			\draw (CA) -- (7.707,0.707);
			\draw (CA) -- (7+1,0);
			\draw (CA) -- (7.707,-0.707);
			\draw (CA) -- (7,-1);
			\draw (CA) -- (7-0.707,-0.707);

			\node at (-2,-2) {$SP(d_1,\ 0)$};
			\node at (0,-2) {$SP(d_2,\ 1)$};
			\node at (3,-2) {$SP(d_3,\ 2)$};
			\node at (5,-2) {$. . .$};
			\node at (7,-2) {$SP(d_{k^\prime},\ k^\prime-1)$};
		\end{tikzpicture}
	\end{minipage}
	
	\caption{Optimal burning of a graph $G_{SP}$ with $k^\prime$ connected components, where $G_{SP} = SP(d_1,\ 0)$ $\cup$ $SP(d_2,\ 1)$ $\cup$ $SP(d_3,\ 2)$ $\cup$ $. . .$ $\cup$ $SP(d_{k^\prime},\ k^\prime-1)$ and $d_i \geq i+1$.}
	\label{figure:burn-many-components}
\end{figure}
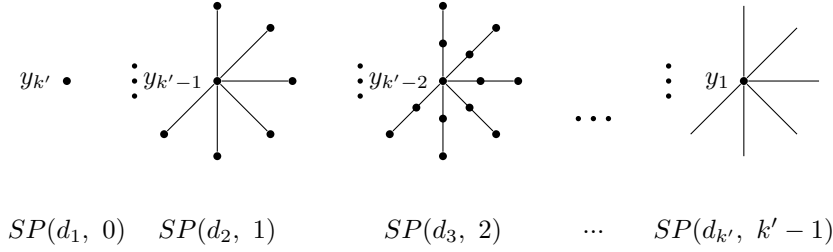

\section{Algorithm for burning general graphs optimally}

\index{algorithm: optimal burning of general graphs}An algorithm to compute a burning sequence for an arbitrary graph is given in \cite{Bessy2017}. \Cref{algorithm:burn-general-optimal} is similar to that algorithm; it computes an optimal burning sequence for an arbitrary input graph $G$ in exponential time. It is a brute force algorithm which checks for every possible permutation of vertices in $G$, sequentially, if it is a valid burning sequence.\\

\begin{algorithm}\label{algorithm:burn-general-optimal}
Given the input graph $G$ perform the following steps.
\end{algorithm}

\textbf{\textit{Stage 1.}} $n = |G.V|$. Let that the vertices in $G.V$ be labelled as $v_1,v_2,\dots,v_n$. $i=1$. if $i\leq n$, perform the following steps.

\textbf{\textit{Stage 1.1.}} For every permutation $P$ of size $i$ in the sequence $(1,2,\dots,n)$ (where $P$ itself is a sequence of numbers), perform the following steps.

\textbf{\textit{Stage 1.1.2.}} Let $S = (v_{P[1]},v_{P[2]},\dots,v_{P[|P|]})$. If $S$ is able to burn $G$, that is, if $(G,S)$ is passed as input to \Cref{algorithm:burn-verify}, it returns $true$, then return $S$.

\textbf{\textit{Stage 1.2.}} $i=i+1$.

\section{Summary of this chapter}

We visit burning properties of some graph classes such as paths, cycles, spider graphs with all arms of equal lengths. We also discuss some of the relations that can be determined between two distinct graphs with respect to graph burning, for example the burning number of a graph is always greater than or equal to any of its isometric subgraph. We utilize these properties in the following chapters to lay down the graph burning theory in more detail.

\chapter{Other games and problems}\label{chapter:other-games-and-problems}

In this chapter, we are going to describe some problems which are related to the graph burning problem, and some problems which we utilize in the following chapters to prove various properties of graph burning.

\section{Distinct 3-partition problem}\label{section:d3pp}

\index{distinct 3-partition problem}For definition of the distinct 3-partition problem, see \Cref{section:problems-in-NP}. \cite{Hulett2008} have showed that the distinct $3$-partition problem is NP-Complete in the strong sense. So as per our discussion in \Cref{section:strong-NPC}, we have that if each element of the input set is bounded above by the polynomial in the length of the input, then the distict 3-partition problem still remains NP-Complete.

We have from \Cref{lemma:ptr-npc} and \cref{lemma:pptr-npc} in \Cref{section:ptr-pptr-npc} that if a problem $B$ (which is NP-Complete in the strong sense) is reducible (in polynomial or pseudo-polynomial time) to another problem $B^\prime$ which is in NP, we have that $B^\prime$ is also an NP-Complete problem in the strong sense. We reduce several problems from the distinct 3-partition problem in the following chapters to show the NP-Completeness of various graph burning subproblems (in the strong sense).

\section{Firefighter problem}

\index{firefighter problem}The firefighting game, in combination with graph burning, can be used to model, for example, various ``real-life'' social scenarios where (communicable) diseases are spread. We also discuss (briefly) majorly in the conclusion how the firefighting game can be used to model the hindrance to the spread of a disease (the spread of which, we can model through graph burning). In fact, firefighter was introduced (in \cite{Hartnell1995}) to capture several important applications which included immunization of a population against a virus. In this section, we briefly discuss the firefighter problem. For the definition of the firefighter problem, see \Cref{section:related-games}.

In the firefighter problem, we input an arbitrary graph $G$, and a vertex $s\in G.V$ which is the initial fire source. After the computation of a firefighting process, a sequence of vertices $S$ is generated which is used, (sequentially) one vertex in each round, to put the firefighter on, and thus, protect it the spread of fire. Once a vertex is protected, it cannot be burnt (and so it can also not spread the fire). \Cref{lemma:verify-firefighter} was proved in \cite{Fomin2016}.

\begin{lemma}\label{lemma:verify-firefighter}
We can verify the solution for a firefighter problem on a graph $G$ of $n$ vertices and $m$ edges with a fire source $s$ in $O(n+m)$ time. \cite{Fomin2016}
\end{lemma}

\begin{proof}
The input is $G$, $s$, and a sequence $S$ of vertices where the firefighter is placed in the (same) sequence of time steps. The vertices on which firefighter is placed till step $i$ is $S_i = S[1:i-1]$, an algorithm is described as \Cref{lemma:verify-firefighter}.

\begin{algorithm}\label{algorithm:verify-firefighter}
Given the input $(G,s,S)$, perform the following steps.
\end{algorithm}

\textbf{\textit{Stage 1.}} $B=\{s\}$ (stores the set of vertices which have been burnt). $k=|S|$.

\textbf{\textit{Stage 2.}} $\forall\ 2\leq i\leq k$, perform the following steps.

\textbf{\textit{Stage 2.1.}} If $S[i]\in T$, then return $false$.

\textbf{\textit{Stage 2.2.}} $B=B\cup (G\setminus S[1:i-1]).Adj[B]$.

\textbf{\textit{Stage 3.}} If $B\cup (G\setminus S).Adj[B]=B$ (that is, if fire can spread no more), then return $true$, else return $false$.
\end{proof}

\begin{corollary}
    The firefighter problem is in NP.
\end{corollary}

\begin{lemma}
    The firefighter and the firefighter reserve deployment problem are the same. \cite{Fomin2016}
\end{lemma}

\begin{proof}
    Followed by the definitions (see definitions of both in \Cref{section:related-games}).
\end{proof}

\begin{lemma}\label{lemma:firefighter-path-from-s}
Given the instance $(G, s)$ of a firefighter problem if $l$ is the number of vertices in a longest induced path in $G$, starting in $s$, then no optimal strategy can place firefighter on more than $l-1$ vertices. \cite{Fomin2016}
\end{lemma}

\begin{proof}
Let the firefighter sequence be $S=(v_1, v_2, . . ., v_t)$. Let $P$ be a path between $v_t$ and $s$ such that all vertices in $P$ burn, except $v_t$. Let that $P$ is the only path with this property. $\implies |P|\geq t+1$, or otherwise not have protected $v_t$ in $t^{th}$ step. Therefore, $t\leq l-1$.
\end{proof}

\begin{theorem}\label{theorem:firefight-P-k-free}
Firefighter problem can be solved in $O(n^{k-2}(n+e)) = O(n^k)$ time for $P_k$-free graphs. \cite{Fomin2016}
\end{theorem}

\begin{proof}
In a $P_k$-free graph $G$, the longest induced path is of length $k-1$. By \Cref{lemma:firefighter-path-from-s}, any optimal strategy protects at most $k-2$ vertices.
\end{proof}

Proof of \Cref{theorem:firefight-P-k-free} suggests the algorithm \Cref{algorithm:firefight-P-k-free} as follows, which is able to perform the firefighting problem in a $P_k$-free graph $G$ in $O(n^{k-2}(n+m))$ time (where $n$ is the number of vertices in $G$ and $m$ is the number of edges in it).

\begin{algorithm}\label{algorithm:firefight-P-k-free}
    Given an input graph $G$, perform the following steps.
\end{algorithm}

\textbf{\textit{Stage 1.}} Enumerate all ordered subsets $S \subseteq V(G)$ of size at most $k-2$ in $n^{k-2}$ time.

\textbf{\textit{Stage 2.}} For each $S$, perform the following steps.

\textbf{\textit{Stage 2.1.}} Verify by lemma \thech.2 that it is a valid firefighter sequence in $O(n+m)$ time.

\textbf{\textit{Stage 2.3.}} If $S$ is a valid firefighter sequence, then return $S$.\\

\Cref{corollary:firefight-split-graphs} has been proved in \cite{Fomin2016}.

\begin{corollary}\label{corollary:firefight-split-graphs}
    Firefighter problem can be solved on split graphs in polynomial time.\cite{Fomin2016}.
\end{corollary}

\begin{proof}
    Since split graphs are $P_5$ free graphs, this corollary follows from \Cref{theorem:firefight-P-k-free}.
\end{proof}

\section{Firefighting general graphs}

\Cref{algorithm:firefight-general-graphs} computes an optimal firefighting sequence for an arbitrary input graph $G$ in exponential time. It checks for every possible permutation of vertices in $G$, sequentially, if it is a valid firefighting sequence. It is a brute force algorithm.\\

\begin{algorithm}\label{algorithm:firefight-general-graphs}
Given the input $(G,\ s)$, where $G$ is an arbitrary graph, and $s\in G.V$ is the initial fire source, perform the following steps.
\end{algorithm}

\textbf{\textit{Stage 1.}} $n = |G.V|$. Let that the vertices in $G.V$ be labelled as $v_1,v_2,\dots,v_n$. $i=1$. $S_{max}$ is the firefighting sequence which saves maximum vertices and $m$ is the number of vertices saved by $S_{max}$, initially $m=0$, $S_{max}=\phi$. If $i\leq n$, perform the following steps.

\textbf{\textit{Stage 1.1.}} For every permutation $P$ of size $i$ in the sequence $(1,2,\dots,n)$ (where $P$ itself is a sequence of numbers), perform the following steps.

\textbf{\textit{Stage 1.1.2.}} Let $S = (v_{P[1]},v_{P[2]},\dots,v_{P[|P|]})$. If $S$ is is a valid firefighting sequence by \Cref{algorithm:verify-firefighter} (\Cref{lemma:verify-firefighter}), then if $m>$ the total number of vertices protected and saved by $S$, then $S_{max}=S$ and $m=$ the number of vertices protected and saved by $S$.

\textbf{\textit{Stage 1.2.}} $i=i+1$.

\textbf{\textit{Stage 2.}} Return $S_{max}$.\\

Using \Cref{algorithm:firefight-general-graphs}, we compute a valid firefighting sequence which is able to protect and save the maximum number of vertices in $G$. In-spite of computing this throughout all the permutations of all sizes in $v_1,v_2,\dots,v_n$, we could also return the first valid firefighting sequence that we encounter, but that would not suffice the task of saving the maximum vertices.

\section{Summary of this chapter}

We discuss the distinct 3-partition problem, which we are going to reduce to the burning several graph classes to show that the problem is NP-Complete for these graph classes. We also discuss the firefighter problem which we can utilize to use with graph burning to model protection of a graph against graph burning. We discuss some of the models that can be built using this concept towards the conclusion, \Cref{chapter:conclusion}.

\chapter{Why is burning hard?}\label{chapter:why-hard}

In this chapter, we describe some of our novel results in \Cref{section:burn-interval-graphs}, \Cref{section:burn-permutation-graphs}, and \Cref{section:burn-disk-graphs}. Respectively, we show that burning interval graphs, permutation graphs and disk graphs is NP-Complete. In fact, NP-Completeness of burning disk graphs follows from the NP-Completeness of interval graphs, but our construction in the proof for disk graphs is such that another class of graphs follows to be NP-Complete from it.

As corollaries of our results in these sections, it follows that burning several other graph classes remains NP-Complete. These corollaries are obtained for those graph classes for which NP-Completeness results have been shown earlier. We provide citations to the resources where the respective results were first discussed.

\section{Burning interval graphs}\label{section:burn-interval-graphs}

\subsection{Similarity in burning paths and interval graphs}\label{subsection:similar-burn-path-IG}

Refer to the definition of the burning cluster of a fire source, described as \Cref{definition:burning-cluster} in \Cref{section:burn-problem}. While demonstrating the burning of a path or cycle in \Cref{section:burn-path}, we obtain an observation described as \Cref{observation:path} as follows.

\begin{observation}\label{observation:path}
The burning clusters of each of the $n$ fire sources of any optimal burning sequence of a path of order $n^2$ are pairwise disjoint.
\end{observation}

From \Cref{corollary:NHClosurePIG} in \Cref{subsection:interval-graphs}, \Cref{observation:bound} follows, which we describe as follows.

\begin{observation}\label{observation:bound}
Let $P$ be a shortest path of maximum length among shortest paths between all pairs of vertices in an interval graph $G$. Then $b(P) \leq b(G) \leq b(P)+1$.
\end{observation}

\Cref{observation:bound} has been shown earlier in \cite{Kare2019,Kamali2019,Kamali2020}.

Finding such $P$ is easy to do in polynomial time. We can simply compute all pair shortest path and choose the maximum length path among all. So burning an interval graph in $(b(P)+1)$ is trivial. We study if $b(G)=b(P)$, whether it is always possible to determine the same and compute a burning sequence of length $b(P)$. 
We show that this is an NP-Complete problem.
We use the distinct 3-partition problem (see \Cref{section:d3pp}) and reduce it to interval graph burning to show NP-Completeness of this graph burning subproblem.

\subsection{Interval graph construction for NP-completeness}

Let $X$ be an input set to the distinct 3-partition problem; let $n=\frac{|X|}{3}$, $m = \max (X)$, $B = \frac{s(X)}{n}$, and $k=m-3n$. Let $F_m$ be the set of first $m$ natural numbers, $F_m = \{1,2,3,...,m\}$; and $F^{\prime}_m$ be the set of first $m$ odd numbers, $F^\prime_m = \{2\ f_i-1: f_i \in F_m\} = \{1, 3, 5, . . ., 2m-1\}$. Let $X^\prime = \{2\ a_i-1:a_i \in X\}$, $B^\prime = \frac{s(X^\prime)}{n}$. Observe that $s(X^\prime) = \sum_{i=1}^{3n} 2\ a_i -1 = 2nB-3n$, so $B^\prime = 2B-3$. Let $Y=F_m^\prime\setminus X^\prime$.

Let there be $n$ paths $Q_1,Q_2...,Q_n$, each of order $B^\prime$; $k$ paths $Q_1^\prime,Q_2^\prime,...,Q_k^\prime$ such that each $Q_j^\prime (1\leq j\leq k)$ is of order of $j^{th}$ largest number in $Y$, and $m+1$ paths $T_1,T_2,...,T_{m+1}$ such that each $T_j (1\leq j\leq m+1)$ is of order of $2(2m+1-j)+1$. We join these paths in the following order to form a larger path:\\
$Q_1$, $T_1$, $Q_2$, $T_2$, $...$, $Q_n$, $T_n$, $Q_1^\prime,T_{n+1}$, $Q_2^\prime$, $T_{n+2}$, $...$, $Q_k^\prime$, $T_{n+k}$, $T_{n+k+1}$, $...$, $T_{m+1}$.\\
We call this path $P_I$. Number of vertices in $P_I$ is $(2m+1)^2$. Hence $b(P_I)=(2m+1)$.

Now we add a few more vertices to $P_I$. 
We add a distinct vertex connected to each vertex from $2nd$ to $2nd$-last vertices of $T_j$ $\forall\ 1\leq j\leq m+1$.
Let this graph be called $IG(X)$. Total number of vertices in $IG(X)$ is $(2m+1)^2+((2m)^2-(m-1)^2)=4m^2+1+4m+4m^2-m^2-1+2m=7m^2+6m$. 
Observe that, in $IG(X)$, $P_I$ is a shortest path of maximum length among shortest paths between all pairs of vertices, and all other vertices not in $P_I$ are at most at a distance one from some vertex of $P_I$. Also, $IG(X)$ is a tree: it contains no cycles.
So $IG(X)$ is a valid interval graph (an example illustrated in \Cref{subsection:ExConsIG}, and \Cref{figure:IGNPCEIG}).
We show in \Cref{lemma:BNb(G)IG} that burning number of $IG(X)$ is still $2m+1$, i.e., $b(IG(X))=b(P_I)$. 
Then we show that burning an interval graph is NP-Complete because to burn $IG(X)$, we must solve the distinct 3-partition problem on $X$.

\Cref{figure:TStructureIG} illustrates a particular $T_j$ along with the added vertices and edges (vertically upwards w.r.t. $T_j$). This forms a comb structure, we call it $A_{T_j}$.

\newcounter{i}
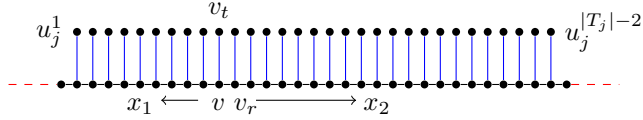
\begin{figure}
    \centering
    \begin{tikzpicture}[scale=.7]
        \setcounter{c}{0}
        \setcounter{d}{-1}
        \setcounter{i}{1}
        \loop
            \ifthenelse{\value{c}=5}{
                \node [circle, fill=black, inner sep=0pt, minimum size=3pt, label=below:{$x_1$}]      (A\thec) at (\value{c}*.3,0) {};
            }{
                \ifthenelse{\value{c}=11}{
                   \node [circle, fill=black, inner sep=0pt, minimum size=3pt, label=below:{$v\ v_r$}] (A\thec) at (\value{c}*.3,0) {};
                }{
                    \ifthenelse{\value{c}=20}{
                        \node [circle, fill=black, inner sep=0pt, minimum size=3pt, label=below:{$x_2$}] (A\thec) at (\value{c}*.3,0) {};
                    }{
                        \node [circle, fill=black, inner sep=0pt, minimum size=3pt] (A\thec) at (\value{c}*.3,0) {};
                    }
                }
            }

            \ifthenelse{\value{c}>0 \AND \value{c}<32}{
                \ifthenelse{\value{c}=1}{
                    \node [circle, fill=black, inner sep=0pt, minimum size=3pt, label=left:{$u_j^1$}] (B\thec) at (\value{c}*.3,1) {};

                    \draw[blue] (A\thec) -- (B\thec);
                }{
                    \ifthenelse{\value{c}=31}{
                        \node [circle, fill=black, inner sep=0pt, minimum size=3pt, label=right:{$u_j^{|T_j|-2}$}] (B\thec) at (\value{c}*.3,1) {};

                        \draw[blue] (A\thec) -- (B\thec);
                    }{
                        \ifthenelse{\value{c}=10}{
                            \node [circle, fill=black, inner sep=0pt, minimum size=3pt, label=above:{$v_t$}] (B\thec) at (\value{c}*.3,1) {};

                            \draw[blue] (A\thec) -- (B\thec);
                        }{
                            \node [circle, fill=black, inner sep=0pt, minimum size=3pt] (B\thec) at (\value{c}*.3,1) {};

                            \draw[blue] (A\thec) -- (B\thec);
                        }
                    }
                }
            }{}

            \ifthenelse{\value{c}>0}{
                \draw (A\thec) -- (A\thed);
            }{}

            \ifthenelse{\value{c}=0}{
                \draw[red,dashed] (-1,0) -- (A\thec);
            }{}
            \ifthenelse{\value{c}=32}{
                \draw[red,dashed] (A\thec) -- (\value{c}*.3+1,0);
            }{}

            \stepcounter{c}
            \stepcounter{d}
            \stepcounter{i}
            \ifnum \value{c}<33
            \repeat

            \draw[<-] (1.9,-.3) -- (2.6,-.3);
            \draw[->] (3.7,-.3) -- (5.6,-.3);
    \end{tikzpicture}
    \caption{Structure of a $T_j$ with 33 vertices, along with the extra vertices connected to it. The dashed line represents the fact that other subpaths may be connected to a $T_j$ on either or both ends.}
    \label{figure:TStructureIG}
\end{figure}

Let $u^1_j$ be the vertex connected to the $2nd$ vertex of each $T_j$ and $u^{|T_j|-2}_j$ be the vertex connected to its $2nd$-last vertex of $T_j$, where $|T_j|$ stands for the number of vertices in the sub path $T_j$. Let $A_{T_j}^\prime=\{u^1_j, u^2_j, ..., u^{|T_j|-2}_j\}$ be the set of all $|T_j|-2$ additional vertices corresponding to $T_j$. Let $A_{T_j}$ be the subgraph induced by the vertices in $A_{T_j}^\prime\cup T_j$. We have the following observation regarding $A_{T_j}$.

\begin{observation}\label{observation:overlap}
If $A_{T_j}$ is burnt by $o\geq 2$ fire sources put on $T_j$, then the burning clusters of at least two of these fire sources will overlap (contain common vertices).
\end{observation}

\begin{proof}
Let that some $T^c_j$ be completely burnt by two or more fire sources and yet there is no overlap between the burning clusters of any of those fire sources.
Since all the fire sources are on $T_j$, which is a sub path of $P_I$, we call two fire sources on $T_j$ are \textit{adjacent} if there is a path in $T_j$ between those two fire sources such that the path does not contain any other fire sources.
For any two adjacent fire sources let us assume that there is no vertex which lies in the burning clusters of both the fire sources. Let $v$ be a vertex on the path joining those two adjacent fire sources $x_1$ and $x_2$, such that the vertices in the left side of $v$ including it (vertices towards $x_1$ as shown in Figure \ref{figure:TStructureIG} using the left arrow) are burnt by $x_1$ and the vertices in the right (Figure \ref{figure:TStructureIG}) by $x_2$.

Let the vertex that is just right to $v$ is $v_r$. By pigeonhole principle, we have that at least one of $v$ or $v_r$ having a neighbor $v_t$ in $T^c_j$ which is not in $T_j$. Without the loss of generality, let that $v$ is having such a neighbor. Since the burning cluster of $x_1$ extends till $v$ and not to its one hop neighbor $v_r\ (\in T_j)$, so it does not burn the other one hop neighbor $v_t\ (\not\in T_j)$ too. It is easy to see that the second fire source can not burn  $v_t$. This is contradiction to our assumption that $T^c_j$ is burnt completely without overlapping clusters.
\end{proof}

\begin{corollary}\label{corollary:Tone}
    If a single fire source is able to burn $T_j$, then $T^c_j$'s can also be burnt with it.
\end{corollary}

\begin{lemma}\label{lemma:notopt}
If at least one $T_j$ is burnt using more than one fire sources, then $P_I$ can not be burnt optimally, i.e., in $b(P_I)= (2m+1)$ steps.
\end{lemma}
\begin{proof}
Since $P_I$ is a simple path of length $(2m+1)^2$, according to \Cref{observation:path}, each fire source in a optimal burning sequence must burn disjoint set of vertices of $P_I$. Let ${x_1, x_2,...,x_{2m+1}}$ be an optimal burning sequence of $P_I$ such that some $T_j$ is burnt using more than one fire sources, then according to \Cref{observation:overlap}, at least two fire sources burn a common vertex of $P_I$ and hence ${x_1, x_2,...,x_{2m+1}}$ can not be an optimal burning sequence.
\end{proof}

Before going to the NP-completeness proof, we construct a specific example of $IG(X)$ below.

\subsection{Example construction}\label{subsection:ExConsIG}

Let $X=\{10,11,12,14,15,16\}$. Then $n=2,\ m = 16,\ B = 39,$ and $k=10$. Also $F_m=\{1,2,...,16\}$ and $F_m^\prime=\{1,3,...,31\}$. Further, $X^\prime = \{19,21,23,27,29,31\}$, $B^\prime = 75=2B-3$ and $Y=\{1,3,5,7,9,11,13,15,17,25\}$.
Observe that $Q_1$ and $Q_2$ are paths of size $75$, and each $Q_1^\prime, Q_2^\prime,...,Q_{k}^\prime$ are paths of order of $25$, $17$, $15$, $13$, $11$, $9$, $7$, $5$, $3$, $1$ respectively. $T_1,T_2,T_3,...T_{m+1}$ are of order of $65,63, 61...,33$ respectively.
We add a vertex connected to each vertex from $2nd$ to $2nd$-last vertices of $T_j (1\leq j\leq m+1)$.
Observe that this is a valid interval graph. The constructed example $IG(X)$ is shown in \Cref{figure:IGNPCEIG}.

Next we show that this interval graph can be burned optimally only if 3-partition problem can be solved for $X=\{10,11,12,14,15,16\}$.

\newcounter{r}
\begin{figure}
    \centering
    \begin{tikzpicture}[scale=.7]
        \setcounter{c}{0}
        \setcounter{r}{0}
        \loop
            \node [circle, fill=black, inner sep=0pt, minimum size=3pt] (A) at (\value{r}*3.75,-\value{c}*2) {};
            \ifnum \value{r}>0
                \draw[red] (B) -- (A);
            \fi
            \node [circle, fill=black, inner sep=0pt, minimum size=3pt] (B) at (\value{r}*3.75+2.75,-\value{c}*2) {};

            \draw (A)--(B);

            \ifthenelse{\isodd{\value{r}}}{
                \node [circle, fill=black, inner sep=0pt, minimum size=3pt] (P1) at (\value{r}*3.75+.2,-\value{c}*2) {};
                \node [circle, fill=black, inner sep=0pt, minimum size=3pt] (P2) at (\value{r}*3.75+2.75-.2,-\value{c}*2) {};

                \node [circle, fill=black, inner sep=0pt, minimum size=3pt] (U1) at (\value{r}*3.75+.2,-\value{c}*2+.75) {};
                \node [circle, fill=black, inner sep=0pt, minimum size=3pt] (U2) at (\value{r}*3.75+2.75-.2,-\value{c}*2+.75) {};

                \draw[blue] (P1) -- (U1);
                \draw[blue] (P2) -- (U2);

                \node [circle, fill=black, inner sep=0pt, minimum size=3pt] (P3) at (\value{r}*3.75+.4,-\value{c}*2) {};
                \node [circle, fill=black, inner sep=0pt, minimum size=3pt] (P4) at (\value{r}*3.75+2.75-.4,-\value{c}*2) {};

                \node [circle, fill=black, inner sep=0pt, minimum size=3pt] (U3) at (\value{r}*3.75+.4,-\value{c}*2+.75) {};
                \node [circle, fill=black, inner sep=0pt, minimum size=3pt] (U4) at (\value{r}*3.75+2.75-.4,-\value{c}*2+.75) {};

                \draw[blue] (P1) -- (U1);
                \draw[blue] (P2) -- (U2);

                \node [circle, fill=black, inner sep=0pt, minimum size=3pt] (P5) at (\value{r}*3.75+.6,-\value{c}*2) {};
                \node [circle, fill=black, inner sep=0pt, minimum size=3pt] (P6) at (\value{r}*3.75+2.75-.6,-\value{c}*2) {};

                \node [circle, fill=black, inner sep=0pt, minimum size=3pt] (U5) at (\value{r}*3.75+.6,-\value{c}*2+.75) {};
                \node [circle, fill=black, inner sep=0pt, minimum size=3pt] (U6) at (\value{r}*3.75+2.75-.6,-\value{c}*2+.75) {};

                \draw[blue] (P3) -- (U3);
                \draw[blue] (P4) -- (U4);
                \draw[blue] (P5) -- (U5);
                \draw[blue] (P6) -- (U6);

                \node [circle, fill=black, inner sep=0pt, minimum size=2pt] at (\value{r}*3.75+2.75/2,-\value{c}*2+.375) {};
                \node [circle, fill=black, inner sep=0pt, minimum size=2pt] at (\value{r}*3.75+2.75/2-.2,-\value{c}*2+.375) {};
                \node [circle, fill=black, inner sep=0pt, minimum size=2pt] at (\value{r}*3.75+2.75/2+.2,-\value{c}*2+.375) {};
            }{}

            \stepcounter{r}
            \ifnum\value{r}<4
            \repeat

        \node at (1.375,-.5) {$Q_1$}; \node at (5.125,-.5) {$T_1$};
        \node at (8.625,-.5) {$Q_2$}; \node at (13,-.5) {$T_2$};

        \stepcounter{c}
        \setcounter{r}{5}
        \loop
            \node [circle, fill=black, inner sep=0pt, minimum size=3pt] (A\ther) at (\value{r}*2.5,-\value{c}*2) {};
            \node [circle, fill=black, inner sep=0pt, minimum size=3pt] (B\ther) at (\value{r}*2.5+1.5,-\value{c}*2) {};
            \ifnum\value{r}=5
                \draw[red] (B\ther) -- (B);
            \fi
            \ifnum\value{r}<5
                \setcounter{n}{\value{r}}
                \stepcounter{n}
                \draw[red] (B\ther) -- (A\then);
            \fi

            \draw (A\ther)--(B\ther);

            \ifthenelse{\isodd{\value{r}}}{}{
                \node [circle, fill=black, inner sep=0pt, minimum size=3pt] (P1) at (\value{r}*2.5+.2,-\value{c}*2) {};
                \node [circle, fill=black, inner sep=0pt, minimum size=3pt] (P2) at (\value{r}*2.5+1.5-.2,-\value{c}*2) {};

                \node [circle, fill=black, inner sep=0pt, minimum size=3pt] (U1) at (\value{r}*2.5+.2,-\value{c}*2+.75) {};
                \node [circle, fill=black, inner sep=0pt, minimum size=3pt] (U2) at (\value{r}*2.5+1.5-.2,-\value{c}*2+.75) {};

                \draw[blue] (P1) -- (U1);
                \draw[blue] (P2) -- (U2);

                \node [circle, fill=black, inner sep=0pt, minimum size=3pt] (P3) at (\value{r}*2.5+.4,-\value{c}*2) {};
                \node [circle, fill=black, inner sep=0pt, minimum size=3pt] (P4) at (\value{r}*2.5+1.5-.4,-\value{c}*2) {};

                \node [circle, fill=black, inner sep=0pt, minimum size=3pt] (U3) at (\value{r}*2.5+.4,-\value{c}*2+.75) {};
                \node [circle, fill=black, inner sep=0pt, minimum size=3pt] (U4) at (\value{r}*2.5+1.5-.4,-\value{c}*2+.75) {};

                \draw[blue] (P3) -- (U3);
                \draw[blue] (P4) -- (U4);

                \node [circle, fill=black, inner sep=0pt, minimum size=3pt] (P3) at (\value{r}*2.5+.4,-\value{c}*2) {};
                \node [circle, fill=black, inner sep=0pt, minimum size=3pt] (P4) at (\value{r}*2.5+1.5-.4,-\value{c}*2) {};

                \node [circle, fill=black, inner sep=0pt, minimum size=3pt] (U3) at (\value{r}*2.5+.4,-\value{c}*2+.75) {};
                \node [circle, fill=black, inner sep=0pt, minimum size=3pt] (U4) at (\value{r}*2.5+1.5-.4,-\value{c}*2+.75) {};

                \draw[blue] (P3) -- (U3);
                \draw[blue] (P4) -- (U4);

                \node [circle, fill=black, inner sep=0pt, minimum size=2pt] (U4) at (\value{r}*2.5+1.5/2,-\value{c}*2+.375) {};
                \node [circle, fill=black, inner sep=0pt, minimum size=2pt] (U4) at (\value{r}*2.5+1.5/2+.15,-\value{c}*2+.375) {};
                \node [circle, fill=black, inner sep=0pt, minimum size=2pt] (U4) at (\value{r}*2.5+1.5/2-.15,-\value{c}*2+.375) {};
            }

            \addtocounter{r}{-1}
            \ifnum\value{r}>-1
            \repeat

        \node at (.75,-2.5) {$T_5$}; \node at (3.25,-2.5) {$Q_3^\prime$};
        \node at (5.75,-2.5) {$T_4$}; \node at (8.25,-2.5) {$Q_2^\prime$};
        \node at (10.75,-2.5) {$T_3$}; \node at (13.25,-2.5) {$Q_1^\prime$};

        \stepcounter{c}
        \setcounter{r}{0}
        \loop
            \node [circle, fill=black, inner sep=0pt, minimum size=3pt] (A) at (\value{r}*2.5,-\value{c}*2) {};
            \ifnum\value{r}=0
                \draw[red] (A) -- (A0);
            \fi
            \ifnum \value{r}>0
                \draw[red] (B) -- (A);
            \fi
            \node [circle, fill=black, inner sep=0pt, minimum size=3pt] (B) at (\value{r}*2.5+1.5,-\value{c}*2) {};

            \draw (A)--(B);
            \ifthenelse{\isodd{\value{r}}}{
                \node [circle, fill=black, inner sep=0pt, minimum size=3pt] (P1) at (\value{r}*2.5+.2,-\value{c}*2) {};
                \node [circle, fill=black, inner sep=0pt, minimum size=3pt] (P2) at (\value{r}*2.5+1.5-.2,-\value{c}*2) {};

                \node [circle, fill=black, inner sep=0pt, minimum size=3pt] (U1) at (\value{r}*2.5+.2,-\value{c}*2+.75) {};
                \node [circle, fill=black, inner sep=0pt, minimum size=3pt] (U2) at (\value{r}*2.5+1.5-.2,-\value{c}*2+.75) {};

                \draw[blue] (P1) -- (U1);
                \draw[blue] (P2) -- (U2);

                \node [circle, fill=black, inner sep=0pt, minimum size=3pt] (P3) at (\value{r}*2.5+.4,-\value{c}*2) {};
                \node [circle, fill=black, inner sep=0pt, minimum size=3pt] (P4) at (\value{r}*2.5+1.5-.4,-\value{c}*2) {};

                \node [circle, fill=black, inner sep=0pt, minimum size=3pt] (U3) at (\value{r}*2.5+.4,-\value{c}*2+.75) {};
                \node [circle, fill=black, inner sep=0pt, minimum size=3pt] (U4) at (\value{r}*2.5+1.5-.4,-\value{c}*2+.75) {};

                \draw[blue] (P3) -- (U3);
                \draw[blue] (P4) -- (U4);

                \node [circle, fill=black, inner sep=0pt, minimum size=2pt] (U4) at (\value{r}*2.5+1.5/2,-\value{c}*2+.375) {};
                \node [circle, fill=black, inner sep=0pt, minimum size=2pt] (U4) at (\value{r}*2.5+1.5/2+.15,-\value{c}*2+.375) {};
                \node [circle, fill=black, inner sep=0pt, minimum size=2pt] (U4) at (\value{r}*2.5+1.5/2-.15,-\value{c}*2+.375) {};
            }{}

            \stepcounter{r}
            \ifnum\value{r}<6
            \repeat

        \node at (.75,-4.5) {$Q_4^\prime$}; \node at (3.25,-4.5) {$T_6$};
        \node at (5.75,-4.5) {$Q_5^\prime$}; \node at (8.25,-4.5) {$T_7$};
        \node at (10.75,-4.5) {$Q_6^\prime$}; \node at (13.25,-4.5) {$T_8$};

        \stepcounter{c}
        \setcounter{r}{5}
        \loop
            \node [circle, fill=black, inner sep=0pt, minimum size=3pt] (A\ther) at (\value{r}*2.5,-\value{c}*2) {};
            \node [circle, fill=black, inner sep=0pt, minimum size=3pt] (B\ther) at (\value{r}*2.5+1.5,-\value{c}*2) {};
            \ifnum\value{r}=5
                \draw[red] (B\ther) -- (B);
            \fi
            \ifnum\value{r}<5
                \setcounter{n}{\value{r}}
                \stepcounter{n}
                \draw[red] (B\ther) -- (A\then);
            \fi

            \draw (A\ther)--(B\ther);

            \ifthenelse{\isodd{\value{r}}}{}{
                \node [circle, fill=black, inner sep=0pt, minimum size=3pt] (P1) at (\value{r}*2.5+.2,-\value{c}*2) {};
                \node [circle, fill=black, inner sep=0pt, minimum size=3pt] (P2) at (\value{r}*2.5+1.5-.2,-\value{c}*2) {};

                \node [circle, fill=black, inner sep=0pt, minimum size=3pt] (U1) at (\value{r}*2.5+.2,-\value{c}*2+.75) {};
                \node [circle, fill=black, inner sep=0pt, minimum size=3pt] (U2) at (\value{r}*2.5+1.5-.2,-\value{c}*2+.75) {};
                \draw[blue] (P1) -- (U1);
                \draw[blue] (P2) -- (U2);\node [circle, fill=black, inner sep=0pt, minimum size=3pt] (P3) at (\value{r}*2.5+.4,-\value{c}*2) {};
                \node [circle, fill=black, inner sep=0pt, minimum size=3pt] (P4) at (\value{r}*2.5+1.5-.4,-\value{c}*2) {};

                \node [circle, fill=black, inner sep=0pt, minimum size=3pt] (U3) at (\value{r}*2.5+.4,-\value{c}*2+.75) {};
                \node [circle, fill=black, inner sep=0pt, minimum size=3pt] (U4) at (\value{r}*2.5+1.5-.4,-\value{c}*2+.75) {};

                \draw[blue] (P3) -- (U3);
                \draw[blue] (P4) -- (U4);

                \node [circle, fill=black, inner sep=0pt, minimum size=2pt] (U4) at (\value{r}*2.5+1.5/2,-\value{c}*2+.375) {};
                \node [circle, fill=black, inner sep=0pt, minimum size=2pt] (U4) at (\value{r}*2.5+1.5/2+.15,-\value{c}*2+.375) {};
                \node [circle, fill=black, inner sep=0pt, minimum size=2pt] (U4) at (\value{r}*2.5+1.5/2-.15,-\value{c}*2+.375) {};
            }

            \addtocounter{r}{-1}
            \ifnum\value{r}>-1
            \repeat

        \node at (.75,-6.5) {$T_{11}$}; \node at (3.25,-6.5) {$Q_9^\prime$};
        \node at (5.75,-6.5) {$T_{10}$}; \node at (8.25,-6.5) {$Q_8^\prime$};
        \node at (10.75,-6.5) {$T_9$}; \node at (13.25,-6.5) {$Q_7^\prime$};

        \stepcounter{c}
        \setcounter{r}{0}
        \loop

            \ifnum\value{r}=0
                \node [circle, fill=black, inner sep=0pt, minimum size=3pt] (a) at (\value{r}*2.5,-\value{c}*2) {};
                \draw[red] (a) -- (A0);
            \fi
            \ifnum\value{r}>0
                \node [circle, fill=black, inner sep=0pt, minimum size=3pt] (A) at (\value{r}*2.5,-\value{c}*2) {};
                \ifnum \value{r}=1
                    \draw[red] (A) -- (a);
                \fi
                \ifnum \value{r}>1
                    \draw[red] (B) -- (A);
                \fi
                \node [circle, fill=black, inner sep=0pt, minimum size=3pt] (B) at (\value{r}*2.5+1.5,-\value{c}*2) {};

                \draw (A)--(B);

                \node [circle, fill=black, inner sep=0pt, minimum size=3pt] (P1) at (\value{r}*2.5+.2,-\value{c}*2) {};
                \node [circle, fill=black, inner sep=0pt, minimum size=3pt] (P2) at (\value{r}*2.5+1.5-.2,-\value{c}*2) {};

                \node [circle, fill=black, inner sep=0pt, minimum size=3pt] (U1) at (\value{r}*2.5+.2,-\value{c}*2+.75) {};
                \node [circle, fill=black, inner sep=0pt, minimum size=3pt] (U2) at (\value{r}*2.5+1.5-.2,-\value{c}*2+.75) {};\node [circle, fill=black, inner sep=0pt, minimum size=3pt] (P3) at (\value{r}*2.5+.4,-\value{c}*2) {};
                \node [circle, fill=black, inner sep=0pt, minimum size=3pt] (P4) at (\value{r}*2.5+1.5-.4,-\value{c}*2) {};

                \node [circle, fill=black, inner sep=0pt, minimum size=3pt] (U3) at (\value{r}*2.5+.4,-\value{c}*2+.75) {};
                \node [circle, fill=black, inner sep=0pt, minimum size=3pt] (U4) at (\value{r}*2.5+1.5-.4,-\value{c}*2+.75) {};

                \draw[blue] (P3) -- (U3);
                \draw[blue] (P4) -- (U4);

                \node [circle, fill=black, inner sep=0pt, minimum size=2pt] (U4) at (\value{r}*2.5+1.5/2,-\value{c}*2+.375) {};
                \node [circle, fill=black, inner sep=0pt, minimum size=2pt] (U4) at (\value{r}*2.5+1.5/2+.15,-\value{c}*2+.375) {};
                \node [circle, fill=black, inner sep=0pt, minimum size=2pt] (U4) at (\value{r}*2.5+1.5/2-.15,-\value{c}*2+.375) {};
            \fi

            \draw[blue] (P1) -- (U1);
            \draw[blue] (P2) -- (U2);

            \stepcounter{r}
            \ifnum\value{r}<6
            \repeat

        \node at (.2,-8.5) {$Q_{10}^\prime$}; \node at (3.25,-8.5) {$T_{12}$};
        \node at (5.75,-8.5) {$T_{13}$}; \node at (8.25,-8.5) {$T_{14}$};
        \node at (10.75,-8.5) {$T_{15}$}; \node at (13.25,-8.5) {$T_{16}$};

        \stepcounter{c}
        \setcounter{r}{5}
        \loop
            \node [circle, fill=black, inner sep=0pt, minimum size=3pt] (A\ther) at (\value{r}*2.5,-\value{c}*2) {};
            \node [circle, fill=black, inner sep=0pt, minimum size=3pt] (B\ther) at (\value{r}*2.5+1.5,-\value{c}*2) {};
            \ifnum\value{r}=5
                \draw[red] (B\ther) -- (B);
            \fi
            \ifnum\value{r}<5
                \setcounter{n}{\value{r}}
                \stepcounter{n}
                \draw[red] (B\ther) -- (A\then);
            \fi

            \draw (A\ther)--(B\ther);

            \node [circle, fill=black, inner sep=0pt, minimum size=3pt] (P1) at (\value{r}*2.5+.2,-\value{c}*2) {};
            \node [circle, fill=black, inner sep=0pt, minimum size=3pt] (P2) at (\value{r}*2.5+1.5-.2,-\value{c}*2) {};

            \node [circle, fill=black, inner sep=0pt, minimum size=3pt] (U1) at (\value{r}*2.5+.2,-\value{c}*2+.75) {};
            \node [circle, fill=black, inner sep=0pt, minimum size=3pt] (U2) at (\value{r}*2.5+1.5-.2,-\value{c}*2+.75) {};

            \draw[blue] (P1) -- (U1);
            \draw[blue] (P2) -- (U2);

            \node [circle, fill=black, inner sep=0pt, minimum size=3pt] (P3) at (\value{r}*2.5+.4,-\value{c}*2) {};
            \node [circle, fill=black, inner sep=0pt, minimum size=3pt] (P4) at (\value{r}*2.5+1.5-.4,-\value{c}*2) {};

            \node [circle, fill=black, inner sep=0pt, minimum size=3pt] (U3) at (\value{r}*2.5+.4,-\value{c}*2+.75) {};
            \node [circle, fill=black, inner sep=0pt, minimum size=3pt] (U4) at (\value{r}*2.5+1.5-.4,-\value{c}*2+.75) {};

            \draw[blue] (P3) -- (U3);
            \draw[blue] (P4) -- (U4);

            \node [circle, fill=black, inner sep=0pt, minimum size=2pt] (U4) at (\value{r}*2.5+1.5/2,-\value{c}*2+.375) {};
            \node [circle, fill=black, inner sep=0pt, minimum size=2pt] (U4) at (\value{r}*2.5+1.5/2+.15,-\value{c}*2+.375) {};
            \node [circle, fill=black, inner sep=0pt, minimum size=2pt] (U4) at (\value{r}*2.5+1.5/2-.15,-\value{c}*2+.375) {};

            \addtocounter{r}{-1}
            \ifnum\value{r}>4
            \repeat

        \node at (13.25,-10.5) {$T_{17}$};
    \end{tikzpicture}
    \caption{Construction of example $IG(X)$.}
    \label{figure:IGNPCEIG}
\end{figure}
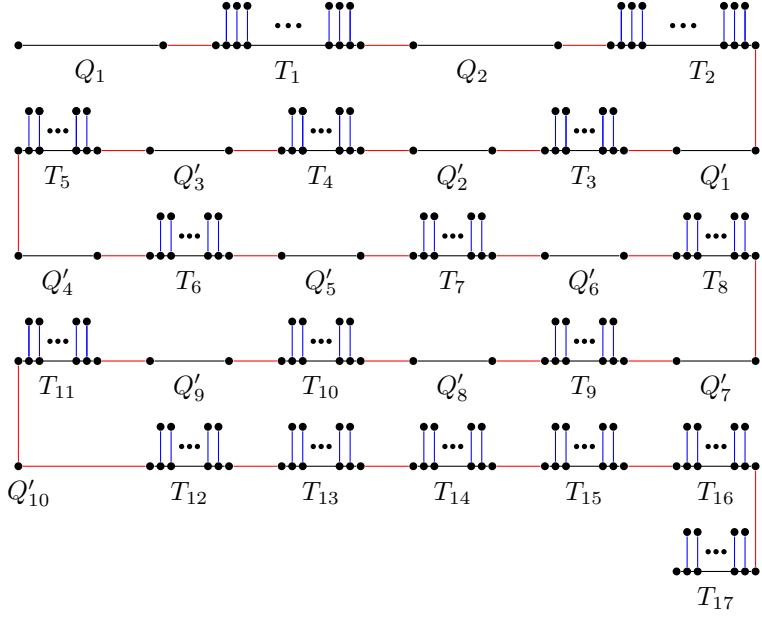

\subsection{NP-Completeness}

\begin{lemma}\label{lemma:BNb(G)IG}
Given that $Q_1,...,Q_n$ can be partitioned into $Q^{\prime\prime}_1,...,Q^{\prime\prime}_{3n}$ of orders in $X^\prime$, then burning number of $IG(X)$ is $2m+1$.
\end{lemma}

\begin{proof}
Let 
$P^\prime = \{Q^{\prime\prime}_1$, $...$, $Q^{\prime\prime}_{3n}$, $Q^\prime_1$, $...$, $Q^\prime_k$, $T_1$, $...$, $T_{m+1}\}$. Let $r_i$ be the $((2m+1)-i+1)^{th} = (2m-i+2)^{th}$ vertex on the $i^{th}$ largest sub path in $P^\prime$. Then, we can burn $P_I$ and subsequently $IG(X)$ if we put $2m+1$ fire sources such that the fire source $y_i$ is put on $r_i$ (from \Cref{corollary:Tone}, if $T_j$ is burnt by a single fire source, then $A_{T_j}$ is also burnt with the same). So, $S^\prime = (y_1,y_2,..,y_{2m+1})$ is a valid burning sequence in this case. This implies that $b(IG(X))\leq 2m+1$. Since the union of all sub paths in $P^\prime$ produces entire $P_I$ which is a path of length $(2m+1)^2$, we have $b(IG(X))\geq 2m+1$. Hence, $b(IG(X)) = 2m+1$.
\end{proof}

Let $P^\prime$ denote the set of all sub paths of $P_I$ such that $P^\prime = \{Q^{\prime\prime}_1$, $...$, $Q^{\prime\prime}_{3n}$, $Q^\prime_1$, $...$, $Q^\prime_k$, $T_1$, $...$, $T_{m+1}\}$.

\begin{lemma}\label{lemma:EFSOnPath}
Each fire source must be on $P_I$.
\end{lemma}

\begin{proof}
If the all the fire sources are on $P_I$, then we have that\\ $G.N_{2m+1-i}[y_i]$ $\cap$ $P_I$\\
is a path of order at most $2(2m+1-i)+1$.

Let for contradiction that we put a fire source $y_i$ on any vertex adjacent to some vertex on $P_I$ and not on $P_I$, and $IG(X)$ can still be burnt within $2m+1$ steps (\Cref{lemma:BNb(G)IG}). Then we have that the subgraph induced by $G.N_{2m+1-i}[y_i]\cap P_I$ is a path of order less that $2(2m+1-i)+1$. This implies that (from \Cref{equation:burn-verify}) $|\cup_{i=1}^{2m+1}G.N_{2m+1-i}[y_i]\cap P_I|<(2m+1)^2$ which is a contradiction. Therefore each $y_i$ must be a put on some vertex on $P_I$.
\end{proof}

Let $S^\prime=(y_1,y_2,...,y_{2m+1})$ be an optimal burning sequence. Let $r_i$ be the $(2m-i+2)^{th}$ vertex on the $i^{th}$ largest sub path in $P^\prime$.
Observe that $T_j$'s are the largest $m+1$ sub paths in $P^\prime$.

\begin{lemma}\label{lemma:AFSOnr_iIG}
We must have $y_i=r_i \forall\ 1\leq i\leq m+1$.
\end{lemma}

\begin{proof}
We are going to prove this lemma using the strong induction hypothesis. We have that each $u_j$ must receive fire from a $y_i$ in $P_I$ (\Cref{lemma:EFSOnPath}). 
For $i=1$, the only vertex connected to both $u^1_1$ and $u^{|T_1|-2}_1$ within a distance $2m+1-i=2m$ is $r_1$.
Now we must have $y_1=r_1$, else, if we put $y_1$ somewhere else, then no other fire source can burn $A_{T_1}$ alone. 
If we utilize more than one fire sources to burn $A_{T_j}$, then then at least one vertex of $T_1$ would be burnt by both of those two fire sources (\Cref{observation:overlap}), following that $P_I$ cannot be burnt completely (\Cref{lemma:notopt}) which is a contradiction.

So, we must have that $y_1=r_1$. Let we need to have $y_k$ on $r_k$ for all $1\leq k\le m$. We need to show for $y_{k+1}$. After burning $T_k$, we must burn $T_{k+1}$ first of all, otherwise we will again obtain overlaps in the burning clusters of the fire sources and we will not be able to burn $P_I$ completely with is a contradiction.

So, using strong induction hypothesis, we must have that $y_{k+1}=r_{k+1}$ to burn entire $A_{T_{k+1}}$ ($\forall\ 1\leq k\leq m$) since the only vertex connected to both $u^1_{k+1}$ and $u^{|T_{k+1}|-2}_{k+1}$  within distance $2m+1-(k+1)$ is $r_{k+1}$.
\end{proof}

We define $P^{\prime\prime}$ by $P^{\prime\prime} = IG(X)\setminus(A_{T_1}\cup A_{T_2}\cup ... \cup A_{T_{m+1}})$. Now we present the following lemma on burning this remaining graph $P^{\prime\prime}$.

\begin{lemma}\label{lemma:PartPFx1TO31}
There is a partition of $P^{\prime\prime}$, induced by the fire sources $y_i (m+1\leq i\leq 2m+1)$, into paths of orders in $F_m^\prime$.
\end{lemma}

\begin{proof}
From \Cref{lemma:AFSOnr_iIG}, we have that $\forall\ 1\leq i\leq m+1$, all the vertices in $T_i$, along with all the vertices connected to it, shall be burnt by $y_i$. Therefore, we have to burn the vertices in $Q_1,...,Q_n,Q_1^\prime,...,Q_k^\prime$ by the fire sources $y_{m+2},x_{m+3},...,y_{2m+1}$ (the last $m$ sources of fire). Since $P^{\prime\prime}$ is a disjoint union of paths, so we have that $\forall\ m+2\leq i\leq 2m+1$, the subgraph induced by the vertices in $G.N_{2m+1-i}[y_i]$ is a path of length at most $2(2m+1-i)+1$. Moreover, we have that the path forest $P^{\prime\prime}$ is of order $\sum_{i=1}^{m}(2i-1)=m^2$. This implies that $\forall\ m+2\leq i\leq 2m+1$, the subgraph induced by the vertices in $G.N_{2m+1-i}[y_i]$ is a path of order equal to $2(2m+1-i)+1$, otherwise we cannot burn all the vertices of $P^{\prime\prime}$ which is a contradiction. Therefore there must be a partition of $P^{\prime\prime}$, induced by the burning sequence $y_{m+2},y_{m+3},...,y_{2m+1}$, into sub paths of order as per each element in $F_m^\prime$.
\end{proof}

\begin{theorem}\label{theorem:BIGNPCIG}
Burning interval graphs optimally is NP-Complete.
\end{theorem}

\begin{proof}
    Considering the partition provided in lemma \ref{lemma:PartPFx1TO31}, we claim that there is a partition of $P^{\prime\prime}$ into subpaths of order as per each element in $F_m^\prime$.
    
    On the other hand, let say we have a optimal solution of burning interval graphs. If each of the $Q_1^{\prime\prime},\dots,Q_{3n}^{\prime\prime}$ is burned by a single fire source, then it gives a solution for the distinct 3-partition problem.
    
    We apply the following process subject to each subpath $s$ in $Q_1^\prime,Q_2^\prime,\dots,Q_k^\prime$. Let that some subpath $s$ is burned using multiple fire sources such that the sum of the cluster sizes of these fire sources is exactly same as $|s|$. Now some fire source with cluster size $|s|$ must be present on some other subpath. We can interchange that fire source (whose cluster size is $|s|$) by these fire sources (which are presently burning $s$). This way we can make each subpath $s$ to be burnt by a single fire source whose cluster size is equal to $|s|$. This process takes $O(m)$ time.
    
    Hence we are left with $Q_1,Q_2,\dots,Q_n$ to burn. Therefore, we must part $Q_1,Q_2,\dots,Q_3$ into $Q^{\prime\prime}_1,Q^{\prime\prime}_2,\dots Q^{\prime\prime}_{3n}$ as per the orders in $X^\prime$. Equivalently we have to part $X$ as per the distinct 3-partition problem. Therefore, we have reduced the burning problem of $IG(X)$ to the distinct 3-partition problem in pseudo-polynomial time. Since, the distinct 3-partition problem is NP-Complete in the strong sense, burning $IG(X)$ is also NP-Complete in the strong sense. Therefore, burning interval graphs optimally is also NP-Complete in the strong sense.
\end{proof}

\section{Burning permutation graphs}\label{section:burn-permutation-graphs}

\subsection{Permutation graph construction for NP-completeness}

Let $X$ be an input set to a distinct 3-partition problem; let $n=\frac{|X|}{3}$, $m = \max (X)$, $B = \frac{s(X)}{n}$, and $k=m-3n$. Let $F_m$ be the set of first $m$ numbers, $F_m = \{1,2,3,...,m\}$, and $F^{\prime}_m$ be the set of first $m$ odd numbers, $F^\prime_m = \{2\ f_i-1: f_i \in F_m\} = \{1, 3, 5, . . ., 2m-1\}$. Let $X^\prime = \{2\ a_i-1:a_i \in X\}$, $B^\prime = \frac{s(X^\prime)}{n}$. Observe that $s(X^\prime) = \sum_{i=1}^{3n} 2\ a_i -1 = 2nB-3n$, so $B^\prime = 2B-3$. Let $Y=F^\prime_m\setminus X^\prime$. Let $O$ be the original sequence of numbers $1$ to $s(F^\prime_m)$, $O=(1,2,3,...,m^2)$.

Now, we are going to construct $n+k$ permutations $P_1, P_2,..., P_{n+k}$ in a specific manner such that these will produce path forests of $(n+k)$ disjoint simple paths. Each $P_j$ is a permutation of the numbers $x_j$ to $y_j$ belonging $O$.
Let $t_j=y_j-x_j+1$; then $P_j = \{p^1_j,p^2_j,...,p^{t_j}_j\}$.
We construct each $P_j$ based on the subsequence $(x_j,... y_j) \in O$ and each of the $p^h_j$ where $h \in [1, t_j]$. Below we first provide a formula to calculate $x_j$, $y_j$. We divide the range of $j$ in two parts, $1\leq j\leq n$ and $n+1\leq j\leq n+k$.

We define $y_0=0$. Now $\forall\ 1\leq j\leq n$, $x_j = y_{j-1}+1,$ and $y_j = j \times B^\prime$.
For the remaining values of $j$, i.e., $\forall\ 1\leq j\leq k$, $x_{n+j} = y_{n+j-1}+1$ and  $y_{n+j} = y_{n+j-1} + L^Y_j$, where $L^Y_j$ is the $j^{th}$ largest element of $Y$. See that, $y_{n+k} = y_n + s(F^\prime_m\setminus X^\prime) = nB^\prime + s(F^\prime_m\setminus X^\prime) = s(X^\prime) + s(F^\prime_m\setminus X^\prime) = s(F^\prime_m) = m^2$. Hence, total number of elements in $\bigcup^{n+k}_{j=1} P_j$ is $m^2$. Now we provide formula to find $p^h_j$ for each $j$ and all $h \in (1, t_j)$.\\

\noindent If $t_j$ is even, then $\forall\ 1\leq j \leq n+k$, we define as follows:

$\forall$ odd $i,\ 1 \leq i \leq (t_j-3), p^i_j = 2+(x_j+i-1)$.
The only odd value of $i=t_j-1$ remains and we define it as, $p^{t_j-1}_j = y_j$.

Further, $\forall$ even $i, 4 \leq i \leq t_j, p^i_j = i-2\}$ and for the remaining value of even $i=2$, we define $p^2_j = x_j$.\\

\noindent Else, if $t$ is odd, then $\forall\ 1\leq j \leq n+k$, we define as follows:

$\forall$ odd $i,\ 1 \leq i \leq t_j-2, p^i_j = 2+(x_j+i-1)$. The only odd value of $i=t_j$ remains and we define it as,  $p^{t_j}_j = y_j-1$.

Further, $\forall$ even $i, 4 \leq i \leq t_j-1, p^i_j = (x_j+i-1)-2$ and for the remaining value of even $i=2$, we define $p^2_j = x_j$.\\

We follow the above construction where we have to compute the permutation of an a subsequence of $O$ of length $5$ or above, that is if $y_j-x_j+1=t_j\geq 5$. If otherwise $t_j\leq 4$, we construct the permutation $P_j$ as follows. If $t_j=1$, then $P_j=(x_j)$. If $t_j=2$, then $P_j=(y_j,x_j)$. If $t_j=3$, then $P_j=(y_j,x_j,x_j+1)$. If $t_j=4$, then $P_j=(x_j+1,y_j,x_j,x_j+2)$.

Now, $P$ $=$ $P_1$ $\cup_{s\setminus}$ $P_2$ $\cup_{s\setminus}$ $...$ $\cup_{s\setminus}$ $P_{n+k}$ $=$ $(p^1_1$, $p^2_1$, $...$, $p^{t_1}_1$, $p^1_2$, $p^2_2$, $...$, $p^{t_2}_2$, $...$, $p^1_{n+k}$, $p^2_{n+k}$, $...$, $p^{t_{n+k}}_{n+k})$ $=$ $(p_1$, $p_2$, $p_3$, $...$, $p_{m^2})$ is the subject permutation of $O$.

We call $P(X)$ to be the permutation graph corresponding to the original sequence $O$, and its subject permutation $P$. $\forall\ 1\leq j\leq n+k$ let $Q_j$ be the subgraph in $P(X)$ induced by the permutation $P_j=(p^1_j,p^2_j,...,p^{t_j}_j)$ of the original sequence $(x_j,...,y_j)$. Observe that $P(X) = Q_1\cup Q_2\cup...\cup Q_{n+k}$ is a path forest where the paths $Q_1, Q_2,..., Q_{n+k}$ are disjoint from each other.

The burning number of $P(X)$ is $m$. It follows trivially from the arguments that we have used in the proofs of \Cref{lemma:PartPFx1TO31} and \Cref{theorem:BIGNPCIG} to argue the burning procedure that should be followed to burn the path forest $P^{\prime\prime}$ because $P^{\prime\prime}$ is similar to $P(X)$.

\subsection{Example construction}

Let $X=\{10,11,12,14,15,16\} \implies n=2,\ m = 16,\ B = 39,$ and $k=10$. $F_m=\{1,2,...,16\}$, and $F_m^\prime=\{1,3,...,31\}$. $X^\prime = \{19,21,23,27,29,31\}$, $B^\prime = 75=2B-3$. $Y=$ $\{1$, $3$, $5$, $7$, $9$, $11$, $13$, $15$, $17$, $25\}$.

We finally form paths $Q_1$ and $Q_2$ each of order of $75$. Also, we form paths $Q_3,Q_4,...,Q_{12}$ of order of $25,17,15,13,11,9,7,5,3,1$ respectively. $P(X)$ is a path forest of the paths $Q_1,...,Q_{12}$, which are disjoint from each other. Burning number of $P(X)$ in this case is $m=16$.

The above example is followed from the general construct that we used to reduce burning permutation graph from a distinct 3-partition problem. In this example, we have constructed paths from subsequences (of $O$) of odd length only. For the sake of another example, let the original sequence be $(1,2,3,...,34).$ Let the $x_1=1, y_1=9, x_2=10, y_2=17, x_3=18, y_3=26, x_4=27, y_4=34$. Now the subject permutation of this sequence becomes $(3$, $1$, $5$, $2$, $7$, $4$, $9$, $6$, $8$, $12$, $10$, $14$, $11$, $16$, $13$, $17$, $15$, $20$, $18$, $22$, $19$, $24$, $21$, $26$, $23$, $25$, $29$, $27$, $31$, $28$ ,$33$, $30$, $34$, $32)$. The resultant permutation graph is a path forest of four paths of order of $9,8,9$ and $8$ respectively.
This shows that we can induce a path forest of any shape and size (containing paths of both even and odd lengths) from a permutation of an original sequence.

\subsection{NP-Completeness}

The path forest induced by $Q_1,Q_2,\dots,Q_{n+k}$ is exactly same as the path forest $P^{\prime\prime}$ of $IG(X)$ that we constructed in \Cref{section:burn-interval-graphs}. So here also $G$ can be burnt optimally only if $Q_1, Q_2, ..., Q_{n+k}$ can be broken into paths of length in $F^\prime_m$ which can happen iff $Q_1,\dots,Q_n$ can be broken into the subpaths $W_1,W_2,\dots,W_{3n}$ of lengths in $X^\prime$ as per the distinct 3-partition problem. Therefore, by the arguments similar to those in the proofs of \Cref{lemma:PartPFx1TO31} and \Cref{theorem:BIGNPCIG}, we have \Cref{theorem:BPGNPC} as follows.

\begin{theorem}\label{theorem:BPGNPC}
    Burning of general permutation graphs optmally is NP-Complete.
\end{theorem}\qed

\section{Burning disk graphs}\label{section:burn-disk-graphs}

\subsection{Disk graph construction for NP-completeness}

Let $X$ be an input set to a distinct 3-partition problem; let $n=\frac{|X|}{3}$, $m = \max (X)$, $B = \frac{s(X)}{n}$, and $k=m-3n$. Let $p=m-1$.

Let $C^\prime$ be a circle of radius $R^\prime$. Let there be a set $Cir$ of $q$ disks $\{c_1, c_2, ..., c_q\}$ of unit radius placed around $C^\prime$ such that their circumference touches circumference of $C^\prime$, but they do not overlap with each other, or with $C^\prime$. The maximum value that $q$ can take is limited. As an example, we can put a maximum of 6 unit radius disks around a disk of unit radius. As the radius $R^\prime$ of the central disk tends to $\infty$, the amount of unit radius disks that we can put tends to $R^\prime\times\pi$ \cite{Marathe1995}. In our construction, the value of $R^\prime$ is chosen such that we can put $q\geq 2(p+2)$ (and $q\leq 3p$ to ensure that the setting of disks that we construct can be constructed in time polynomial to the underlying distinct 3-partition instance) disks of unit radius around $C^\prime$.

We give the definition of \textit{disk-chain} (of a certain size) below.
\begin{definition}\label{definition:DCSizeKDG}{Disk-chain of size} {$k$}. A disk-chain of size $k$ is a sequence of disks $Ch_x = (c_x^1, c_x^2, c_x^3, ..., c_x^k)$ such that $c_x^1$ overlaps only with $c_x^2$, $c_x^k$ overlaps only with $c_x^{k-1}$, and $\forall\ 2\leq j\leq k-1,\ c_x^j$ overlaps only with $c_x^{j-1}$ and $c_x^{j+1}$.
\end{definition}

$1\leq i\leq q$, let that a disk chain of size $p$, $Ch_i=\{c_i^1,c_i^2,...,c_i^p\}$, is attached to each circle $c_i$ such that, apart from the overlaps that give it a chain structure, $c_i^1$ overlaps with $c_i$ and $c_i^2$ only. 
Let $Ch$ be the set of all these $q$ chains, $Ch = \{Ch_i\}_{i=1}^{q}$.

Let there be a disk $C$ of radius $R: R^\prime < R \leq R^\prime + 0.5$ is positioned with its its center exactly at the center of $C^\prime$ defined above. Observe that all the disks in $Cir$ now overlap with $C$. Now consider the corresponding disk graph. Let the vertex corresponding to $C$ be called head $h$, vertices corresponding to $c_i$ be called $v_i$, and the vertices corresponding to $c_i^j$ be called $v_i^j$, $\forall\ 1 \leq i \leq q$, and $\forall\ 1 \leq j \leq p$.

We shall call this setting of disks $DK(R,r,q,p,C,Cir,Ch)$. This setting of disks correspond to the graph $SP(q,p+1)$. Now we are going to extend $DK(R,r,q,p,C,Cir,Ch)$ by adding more disks to it; in fact, we are going to add chains at the terminus of the chains that are already present, which we can do very easily. Before that, let us define what do we mean by attaching disk chain behind another disk chain.

\begin{definition}\label{AttachDCDG}
\textbf{Attaching a disk-chain behind another}. If disk-chain $C_1 = (c_1^1, c_1^2, c_1^3, ..., c_1^{k_1})$ is attached behind another chain $C_2 = (c_2^1, c_2^2, c_2^3, ..., c_2^{k_2})$, a new chain is formed $Ch_{FIN}$ $=$ $c_2^1$, $c_2^2$, $c_2^3$, $\dots$, $c_2^{k_2}$, $c_1^1$, $c_1^2$, $c_1^3$, $\dots$, $c_1^{k_1}$. Clearly, this attachment is done in such a manner that $c_1^1$ overlaps only with $c_2^{k_2}$ and $c_1^2$.
\end{definition}

Let $F_m$ be the set of first $m$ numbers, $F = \{1,2,3,...,m\}$, and $F^{\prime}_m$ be the set of first $m$ odd numbers, $F^\prime_m = \{2\ f_i-1: f_i \in F_m\} = \{1, 3, 5, . . ., 2m-1\}$. Let $X^\prime = \{2\ a_i-1:a_i \in X\}$, $B^\prime = \frac{s(X^\prime)}{n}$. Observe that $s(X^\prime) = \sum_{i=1}^{3n} 2\ a_i -1 = 2nB-3n$, so $B^\prime = 2B-3$. Let $Y=F^\prime_m\setminus X^\prime$.

Let there be $n$ disk-chains $Q_1,Q_2...,Q_n$, each of size $B^\prime$. $\forall\ 1\leq j\leq k$, let $Q_{n+j}$ be a disk chain of size $L^Y_i$, where $L^Y_i$ is the $i^{th}$ largest element of $Y$.

We now attach disk-chains $Q_1,Q_2...,Q_{n+k}$ behind $Ch_1$, $Ch_2$, $\dots$, $Ch_{n+k}$ respectively. In the corresponding disk graph, let $P_i^\prime$ be the path induced by the disk chain $Q_i$, $\forall\ 1\leq i\leq n+k$.

Let the corresponding disk graph of this updated setting of disks be called $DK(X)$.

\subsection{Example construction}

Let $X=\{10,11,12,14,15,16\} \implies n=2,\ m = 16,\ B = 39,$ and $k=10$, $p=15$.

Let $R^\prime=10$, then we can attach $q=34$ chains each of length $p$. Take $R=10.5$.

$F_m=\{1,2,...,16\}$, and $F_m^\prime=\{1,3,...,31\}$. $X^\prime =$ $\{19$, $21$, $23$, $27$, $29$, $31\}$, $B^\prime = 75=2B-3$. $Y=\{1,3,5,7,9,11,13,15,17,25\}$.

We finally obtain paths $P^\prime_1$ and $P^\prime_2$ each of order of $75$. Also, we form paths $P^\prime_3,P^\prime_4,...,P^\prime_{12}$ of order of $25,17,15,13,11,9,7,5,3,1$ respectively..

The central spider graph formed is $SP(34,16)$, and $P^\prime_1,P^\prime_2,...,P^\prime_{12}$ are attached to $v_1^{15},v_2^{15},...,v_{12}^{15}$ respectively at vertices on one of their ends.

Construction of this example $DK(X)$ is demonstrated in \Cref{figure:EDKXDG}. Burning number of $DK(X)$ in this case is $m+1=17$.

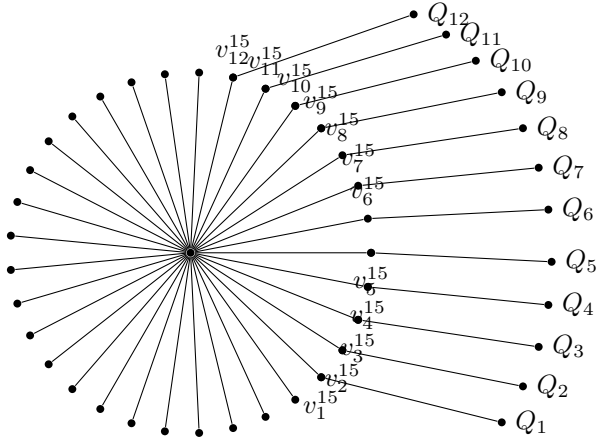
\begin{figure}
    \begin{minipage}{1\textwidth}
        \centering
        \begin{tikzpicture}[scale=1.2]

            \setcounter{n}{34}
            \setcounter{r}{2}
            \setcounter{d}{6}

            \node [circle, fill=black, inner sep=0pt, minimum size=3pt] (A) at (0,0) {};
            \setcounter{c}{1}
            \loop

                \ifnum \value{c} > 3
                    \node [circle, fill=black, inner sep=0pt, minimum size=3pt] (B) at ({\value{r}*cos(\value{c}*360/(\value{n}-1))},{\value{r}*sin(\value{c}*360/(\value{n}-1))}) {};
                \fi

                \ifnum \value{c} < 8

                    \node [circle, fill=black, inner sep=0pt, minimum size=3pt, label=above:{$v_{\thed}^{15}$}] (B) at ({\value{r}*cos(\value{c}*360/(\value{n}-1))},{\value{r}*sin(\value{c}*360/(\value{n}-1))}) {};

                    \node [circle, fill=black, inner sep=0pt, minimum size=3pt, label=right:{$Q_{\thed}$}] (C) at ({\value{r}*cos(\value{c}*360/(\value{n}-1))+2},{\value{r}*sin(\value{c}*360/(\value{n}-1))+.1*\value{c}}) {};

                    \draw (B) -- (C);
                \fi

                \ifnum \value{c} > 28

                    \node [circle, fill=black, inner sep=0pt, minimum size=3pt, label=below:{$v_{\thed}^{15}$}] (B) at ({\value{r}*cos(\value{c}*360/(\value{n}-1))},{\value{r}*sin(\value{c}*360/(\value{n}-1))}) {};

                    \node [circle, fill=black, inner sep=0pt, minimum size=3pt, label=right:{$Q_{\thed}$}] (C) at ({\value{r}*cos(\value{c}*360/(\value{n}-1))+2},{\value{r}*sin(\value{c}*360/(\value{n}-1))-.1*(34-\value{c})}) {};

                    \draw (B) -- (C);
                \fi

                \draw (A) -- (B);

                \stepcounter{c}
                \stepcounter{d}

                \ifnum \value{d}=34
                    \setcounter{d}{1}
                \fi

                \ifnum \value{c} < \value{n}
                    \repeat

        \end{tikzpicture}
    \end{minipage}
    \caption{Construction of example $DK(X)$.}
    \label{figure:EDKXDG}
\end{figure}

\subsection{NP-Completeness}

Observe that $G$ can be burnt optimally only if $P_1^\prime, P_2^\prime, ..., P_{n+k}^\prime$ can be broken into paths of length in $F^\prime_m$ iff $P_1^\prime,\dots,P_n^\prime$ can be broken into the subpaths $W_1,W_2,\dots,W_{3n}$ of lengths in $X^\prime$ as per the distinct 3-partition problem.
Let the final set of subpaths that we desire be $P_P=\{W_1,W_2,\dots,W_{3n},P_{n+1},\dots,P_{n+k}\}$. We have that $|P|=m$. Let $r_i$ ($\forall\ 2\leq i\leq m+1$) be the middle vertex on the $(i-1)^{th}$ largest subpath in $P_P$. Let $r_1$ be the head vertex $h$ of the spider graph $SP(q,p+1)$ induced by $DK(R,r,q,p,C,Cir,Ch)$. $r_1$ will be able to burn $SP(q,p+1)$ because $p=m-1$. So we have that the graph can be burned by the burning sequence $S^\prime=(r_1,r_2,\dots,r_{m+1})$ of length $m+1$. Hence we have \Cref{lemma:b(DK(X))-leq-m+1} as follows.

\begin{lemma}\label{lemma:b(DK(X))-leq-m+1}
    $b(DK(X))\leq m+1$.
\end{lemma}\qed

In \Cref{lemma:first-fire-source-BDG}, we show that the first fire source must be placed at $r_0$, which is stated as follows.

\begin{lemma}\label{lemma:first-fire-source-BDG}
    The first fire source must be placed at $r_0$ (the head vertex $h$).
\end{lemma}

\begin{proof}
    Let for contradiction that the first fire source is not placed at $h$, an we can still burn $DK(X)$ in $m+1$ time steps. From here, we have that $r_0$ will not be able to burn at least $q-1$ leaf nodes of $SP(q,p+1)$ induced by $DK(R,r,q,p,C,Cir,Ch)$. According to our construction, we only have $n+k<m=p+1$ subpaths ($Q_i$'s) attached to $SP(q,p+1)$. So we have that more than $p+1$ unburned leaf nodes are not attached to any subpaths.
    
    Any fire source $r_2,...,r_{m+1}$ will not be able to burn more than one leaf nodes of $SP(q,p+1)$ because distance between any two leaf nodes is $2(p+1)=2m$. So even if we ignore those unburned leaf nodes which are attached to some subpath ($Q_i$), we have that all the fire sources $r_2,...,r_{m+1}$ together will not be able to burn $SP(q,p+1)$. This is a contradiction to our assumption because $DK(X)$ will not be burned in $m+1$ time steps.
\end{proof}

Now we have that the path forest induced by $P_1^\prime,P_2^\prime,\dots,P_{n+k}^\prime$ which is exactly same as the path forest $P^{\prime\prime}$ of $IG(X)$ that we constructed in \Cref{section:burn-interval-graphs}. So it follows trivially by the arguments similar to those in the proofs of \Cref{lemma:PartPFx1TO31} and \Cref{theorem:BIGNPCIG} that (1) $b(DK(X))=m+1$, and (2) the optimal burning of $DK(X)$ is NP-Complete. Therefore, we have \Cref{lemma:b(DK(X))=m+1} and \Cref{theorem:BDGNPC} as follows.

\begin{lemma}\label{lemma:b(DK(X))=m+1}
    $b(DK(X))= m+1$.
\end{lemma}\qed

\begin{theorem}\label{theorem:BDGNPC}
    Burning disk graphs optimally is NP-Complete even if the underlying disk representation is given.
\end{theorem}\qed

\section{Corollary NP-Hard results}

Although the following NP-Completeness have been shown in \cite{Bessy2017}, our constructions in \Cref{section:burn-interval-graphs} imply the corollaries in \Cref{subsection:burn-trees}, \Cref{subsection:burn-chordal-graphs}, \Cref{subsection:burn-planar-graphs}, \Cref{subsection:burn-bipartite-graphs},the constructions in \Cref{section:burn-permutation-graphs} imply the corollary in \Cref{subsection:burn-path-forests}, the constructions in \Cref{section:burn-disk-graphs} imply the corollary in \Cref{subsection:burn-spider-graphs}, and the constructions in \Cref{section:burn-interval-graphs} and \Cref{section:burn-permutation-graphs} together imply the corollary in \Cref{subsection:burn-forests}. Finally, since we have already discussed in \Cref{section:burning-verify} that verification of the correctness of a burning sequence can be done in polynomial time, we have that the graph burning problem is NP-Complete which we state formally in \Cref{subsection:burn-general-graphs-NPC} as \Cref{corollary:burn-general-graphs-NPC}.

\subsection{Burning trees with maximum degree 3}\label{subsection:burn-trees}

\begin{corollary}
    The optimal burning of trees with maximum degree $3$ is NP-Complete.
\end{corollary}

\subsection{Burning chordal graphs}\label{subsection:burn-chordal-graphs}

\begin{corollary}
    The optimal burning of chordal graphs is NP-Complete.
\end{corollary}

\subsection{Burning planar graphs}\label{subsection:burn-planar-graphs}

\begin{corollary}
    The optimal burning of planar graphs is NP-Complete.
\end{corollary}

\subsection{Burning bipartite graphs}\label{subsection:burn-bipartite-graphs}

\begin{corollary}
    The optimal burning of bipartite graphs is NP-Complete.
\end{corollary}

\subsection{Burning path forests}\label{subsection:burn-path-forests}

\begin{corollary}
    The optimal burning of path forests is NP-Complete.
\end{corollary}

\subsection{Burning spider graphs}\label{subsection:burn-spider-graphs}

\begin{corollary}
    The optimal burning of spider graphs is NP-Complete.
\end{corollary}

\subsection{Burning forests}\label{subsection:burn-forests}

\begin{corollary}
    The optimal burning of forests is NP-Complete.
\end{corollary}

\subsection{NP-Completeness of general graph burning}\label{subsection:burn-general-graphs-NPC}

\begin{corollary}\label{corollary:burn-general-graphs-NPC}
    The optimal burning of general graphs is an NP-Complete problem.
\end{corollary}

\section{Summary of this chapter}

We utilize the distinct 3-partition problem to reduce to the burning of interval graphs, permutation graphs, and disk graphs. Several graph classes show their burning properties to be NP-Complete directly from our constructions, such as trees, forests, chordal graphs, planar graphs, bipartite graphs, path forests, and spider graphs.

\chapter{Easy burning subproblems}\label{chapter:where-easy}

Not much work has been done which gave algorithms for easy (polynomial time) burning of graph classes. This chapter includes the findings which show that optimal burning can be performed on certain graph classes in polynomial time.

\section{Burning path or cycle (already discussed)}

We have already discussed in \Cref{section:burn-path} that a simple path or cycle can be burned optimally in polynomial time.

\section{Burning split graphs}

\subsection{Burning connected split graphs}

If the clique $C$ and the independent set $I$ are given for an arbitrary split graph (see \Cref{subsection:split-graphs}) $G$, then we can burn a connected split graph in two or three steps \cite{Kare2019}. \Cref{algorithm:burn-connected-split-graphs} shows how to burn a connected split graph optimally in polynomial time.

\begin{algorithm}\label{algorithm:burn-connected-split-graphs}
Given the input $(G,\ C,\ I)$, where $G$ is the input connected split graph, $C$ (represents the clique) and $I$ (represents the the independent set), both are subsets of $G.V$, perform the following steps.
\end{algorithm}

\textbf{\textit{Step 1.}} Time step $t=0$.\\
$B$ stores the set of vertices which are burned, initially $B=\phi$. $S$ will store an optimal burning sequence of $G$, initially $S=\phi$.

\textbf{\textit{Stage 2.}} Time step $t=1$.\\
$S=\{c\}$ where $c$ is an arbitrary vertex in $C$, preferably adjacent to at least one vertex in $I$. $B=\{c\}$. If $B=G.V$, then return $S$.

\textbf{\textit{Stage 3.}} Time step $t=2$.\\
$S = S\cup \{i\}$, where $i$ is an arbitrary unburned vertex in $I$ (preferably, such that $i$ is not connected to the same vertex in $C$ as $S[|S|]$) or $C$. $B = B \cup \{i\} \cup G.Adj[B]$.\\
If $B=G$, then return $S$.

\textbf{\textit{Stage 4.}} Time step $t=3$.\\
$S = S\cup \{i\}$, where $i$ is an arbitrary unburned vertex in $I$. $B = B \cup \{i\} \cup G.Adj[B]$.\\
Return $S$.\\

\Cref{algorithm:burn-connected-split-graphs} runs in linear time, given that $C$ and $I$ are computed already. For a connected split graph $G$, the burning number $b(G)$ is $2$ or $3$. It is interesting to observe that even if the split graph is not a connected graph, then also we can burn the graph in polynomial time $O(n)$, given that the first fire source is placed in $C$.

\subsection{Burning general split graphs}

\begin{theorem}\label{theorem:burning-general-split-graphs}
    For an arbitrary split graph $G$, the first fire source in an optimal burning sequence should be placed on $c$, an arbitrary vertex in $C$ (representing the clique), a subgraph of $G.V$.
\end{theorem}

\begin{proof}
\noindent \textbf{General construction}: Let $G$ be a split graph of $n$ vertices such that $C = \{c_1, c_2, c_3, . . ., c_{n-1}\}$ is a complete subgraph of $G$ of size $n-1$, and $I = \{i\}$ is the independent set excerpt of $G$. $i$ is not connected to any vertex in $G$. Let the optimal burning sequence be $S^\prime = (y_1, y_2, y_3, . . ., y_{k^\prime})$ be of $k^\prime$ vertices. Consider for contradiction that the fire source $y_1$ is placed on $i$, and yet we are able to burn $G$ in the mininimum possible time steps.\\

\noindent \textbf{Example construction and solution}: Let $n=4$. When $y_1$ is placed on $i$, a burning sequence required to burn $G$ is $(i, c_1, c_2)$. This is shown in figure \thech.2.

\begin{figure}
    \centering
    \begin{tikzpicture}
        \node [circle, fill=black, inner sep=0pt, minimum size=3pt, label=right:{$c_1 = y_2$}] (A) at (0,0) {};
        \node [circle, fill=black, inner sep=0pt, minimum size=3pt, label=below:{$c_2 = y_3$}] (B) at (.707,-1) {};
        \node [circle, fill=black, inner sep=0pt, minimum size=3pt, label=left:{$c_3$}] (C) at (-.707,-1) {};
        
        \node [circle, fill=black, inner sep=0pt, minimum size=3pt, label=above:{$i = y_1$}] (D) at (0,1) {};
        
        \draw (A) -- (B); \draw (A) -- (C); \draw (B) -- (C);
    \end{tikzpicture}
    \caption{If the first fire source is placed on $i$, then the burning sequence is $i,c_1,c_2$. If, otherwise, the first fire source is placed on $c_1$ (for example), then the burning sequence is $c_1,i$}
    \label{figure:burning-general-split-graphs}
\end{figure}
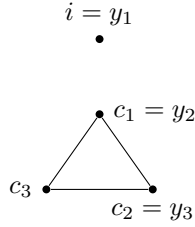

On the other hand, if $y_1$ were placed on $c_1$, an arbitrary vertex in $C$, a burning sequence required to burn $G$ is $(c_1, i)$.
\end{proof}

Given that the vertices of $I$ are connected to atmost one vertex in $C$, then in an optimal burning sequence, we can put the first fire source on a vertex in $I$ iff there are atmost $2$ vertices in $I$ which are disconnected from $C$. 
If otherwise the vertices of $I$ are connected to more than one vertices of $C$, then in an optimal burning sequence, we can put the first fire source on a vertex in $I$ iff there are atmost $3$ vertices in $I$ which are disconnected from $C$.

\Cref{theorem:burning-general-split-graphs} lays the foundation of algorithm \Cref{algorithm:burning-general-split-graphs}, which computes an optimal burning sequence for disconnected split graphs in $O(n)$.

\begin{algorithm}\label{algorithm:burning-general-split-graphs}
    Given the input $(G,\ C,\ I)$, where $G$ is a connected split graph, $C$ represents the clique, $C\subseteq G.V$, and $I$ represents the independent set in $G$, $I\subseteq G.V$, perform the following steps.
\end{algorithm}

\textbf{\textit{Step 1.}} Time step $t=0$.\\
$B$ stores the set of vertices which are burned, initially $B=\phi$. $S$ will store an optimal burning sequence of $G$, initially $S=\phi$.

\textbf{\textit{Stage 2.}} Time step $t=1$.\\
$S=\{c\}$ where $c$ is an arbitrary vertex in $C$, preferably adjacent to at least one vertex in $I$. $B=\{c\}$. If $B=G.V$, then return $S$.

\textbf{\textit{Stage 3.}} Time step $t=2$.\\
$S = S\cup \{i\}$, where $i$ is an arbitrary unburned vertex in $I$ (preferably, such that $i$ is not connected to the same vertex in $C$ as $S[|S|]$) or $C$. $B = B \cup \{i\} \cup G.Adj[B]$.\\
If $B=G$, then return $S$.

\textbf{\textit{Stage 4.}} Perform the following steps until $B=G$.

\textbf{\textit{Stage 4.1.}} Time step $t=t+1$.\\
$S = S\cup \{i\}$, where $i$ is an arbitrary unburned vertex in $I$. $B = B \cup \{i\} \cup G.Adj[B]$.\\

\textbf{\textit{Stage 5.}} Return $S$.

\section{Burning cographs}

Refer to the definition of \textit{cographs} in \Cref{subsection:cographs}. In a cograph $G$, each vertex in $G$ is atmost at a distance $2$ or $3$ from $v$, an arbitrary vertex in $G$ \cite{Kare2019}. Hence we can burn a cograph in atmost $3$ steps.

\section{Summary of this chapter}

In this chapter, we visit graph burning on several graph classes, especially those on which determining the burning number is easy. Although graph burning can be modelled on several graph classes in polynomial time, it remains NP-Hard on general graphs; it is NP-Hard even for many of those graph classes for which several problems, (which are NP-Hard to solve for general graphs), are straightforward to solve.

\chapter{Approximating the burning sequence}\label{chapter:approximation}

\section{Approximation for general graphs}

Bessy et. al. \cite{Bessy2017} have given a 3-approximation algorithm for burning of an arbitrary graph. \Cref{algorithm:generate-burning-sequence} (in combination with \Cref{procedure:choose-last-fire-source}) is able to compute a burning sequence for a graph $G$ in $O(n^3)$ time, where $n=|G.V|$.

\begin{procedure}\label{procedure:choose-last-fire-source}
Given the input $(G,k,S)$ where $G$ is the input graph, $k$ is a positive integer, and $S$ is a burning sequence of size $k-1$, perform the following steps.
\end{procedure}

\textbf{\textit{Stage 1.}} $x_k = \max\{\min\{\frac{d(u,x_j)}{k-j+1}: j\in [1:i-1]\} : u\in G.V \}$.

\textbf{\textit{Stage 2.}} $S=S \cup_{s/} (x_k)$. Return $S$.\\

\begin{algorithm}\label{algorithm:generate-burning-sequence}
Given the input graph $G$, perform the following steps.
\end{algorithm}

\textbf{\textit{Stage 1.}} $x_1 =$ some arbitrary vertex. $S=(x_1)$.

\textbf{\textit{Stage 2.}} $\forall\ k\geq 2$, perform the following steps.

\textbf{\textit{Stage 2.1.}} Invoke \Cref{procedure:choose-last-fire-source} by passing the input $(G,k,S)$ and store the return value in $S$. If $S$ satisfies \cref{equation:burn-verify}, then stop and return $S$.\\

We show in \Cref{lemma:burn-seq-geq3} and \Cref{theorem:proof-3-approx} \cite{Bessy2017} that the burning sequence generated by \Cref{algorithm:generate-burning-sequence} is a burning sequence within 3-approximation. It means that (since graph burning is a minimization problem) if \Cref{algorithm:generate-burning-sequence} generates a burning sequence of length $k$, then an optimal burning sequence (generated by \Cref{algorithm:burn-general-optimal}) of the subject graph $G$ shall contain at least $\frac{k}{3}$ fire sources.

\begin{lemma}\label{lemma:burn-seq-geq3}
Let $S = (x_1, x_2, x_3, . . ., x_k)$ be a burning sequence returned by \cref{procedure:choose-last-fire-source} for an input graph $G$. If $S$ is not able to burn $G$ completely (that is, if $S$ does not satisfy \Cref{equation:burn-verify}), then $b(G)\geq \big\lceil \frac{k}{3}\big\rceil+1$.
\end{lemma}

\begin{proof}
If $S$ does not satisfy \Cref{equation:burn-verify}, then $\exists$ a vertex $u$ such that

$$\min \Big\{ \frac{dist(u,x_j)}{k-j+1} : j\in [k]\Big\} \geq 1$$.

$\forall\ i,j \in [k], i<j$, we have that $dist(i,j) \geq k-j+1 > k-j$. Now $\because j>i$, we have that $k-j > \big\lceil\frac{k-j}{2}\big\rceil + \big\lceil\frac{k-i}{2}\big\rceil$, we obtain that the $k$ sets\\
$N_{\lceil \frac{k-1}{2}\rceil}[x_1], N_{\lceil \frac{k-2}{2}\rceil}[x_2], N_{\lceil \frac{k-3}{2}\rceil}[x_3], . . ., N_0[x_{k}]$\\
are pairwise disjoint.

Now let (for contradiction) that a burning sequence of length $S^\prime = (y_1, y_2, y_3, . . ., y_{k^\prime})$ of length $k^\prime$ is able to burn $G$, where $k^\prime\leq \big\lfloor\frac{k}{3}\big\rfloor$.

$\because$ the ``half'' neighbourhood of the first $\big\lfloor\frac{k}{3}\big\rfloor+1$ fire sources is more than the neighbourhood of the first fire source (the fire source with the maximum neighbourhood reachability) that is,\\
$$\bigg\lceil\frac{k-\big\lceil\frac{k}{3}\big\rceil-1}{2}\bigg\rceil \geq \big\lfloor\frac{k}{3}\big\rfloor-1 \geq k^\prime-1,$$\\
then each of the (first) $\big\lfloor\frac{k}{3}\big\rfloor+1$ sets\\
$N_{\lceil \frac{k-1}{2}\rceil}[x_1], N_{\lceil \frac{k-2}{2}\rceil}[x_2], N_{\lceil \frac{k-3}{2}\rceil}[x_3], . . ., N_{\bigg\lceil\frac{k-\big\lceil\frac{k}{3}\big\rceil-1}{2}\bigg\rceil}\bigg[x_{\big\lceil\frac{k}{3}\big\rceil+1}\bigg]$

\noindent contain at least one element from $S^\prime$. $\because$ all these sets are pairwise disjoint, we obtain the contradiction to our assumption because $S^\prime$ will not be able to burn $G$. So $S^\prime$ must contain at least one more fire source.
\end{proof}

\begin{theorem}\label{theorem:proof-3-approx}
Let $S = (x_1, x_2, x_3, . . ., x_k)$ be a burning sequence returned by \Cref{algorithm:generate-burning-sequence}, then $b(G) \geq \frac{k}{3}$.
\end{theorem}

\begin{proof}
If \Cref{algorithm:generate-burning-sequence} returns a burning sequence of length $k$, then procedure 24 must not have been able to burn $G$ with $k-1$ fire sources.\\
$\implies$ (from \Cref{lemma:burn-seq-geq3}) $b(G) \geq \big\lceil \frac{k-1}{3}\big\rceil+1 \geq \frac{k}{3}$.
\end{proof}

\section{How close can we approximate?}\label{section:no-better-than-3-approx}

Let that for a given arbitrary graph $G,\ S = (x_1, x_2, x_3, . . ., x_k)$ be a burning sequence returned by an arbitrary approximation algorithm $A$ in polynomial time. If for a pair of fire sources $x_i, x_j$ their neighbourhoods are disjoint then the following expression satisfies.
$$\big | G.N_{k-i} [x_i] \cap G.N_{k-j}[x_j] \big | = 0.$$\\
If for a pair of fire sources $x_i, x_j$ it is observed that the $f^{th}$ fraction ($0 \leq f \leq 1$) of their neighbourhoods are disjoint then the following expression satisfies.
$$\big | G.N_{\lceil f \times (k-i) \rceil} [x_i] \cap G.N_{\lceil f \times (k-j)\rceil}[x_j] \big | = 0.$$\\
So, ``half'' neighbourhoods (where $f=\frac{1}{2}$) of any two distinct fire sources $x_i$ and $x_j$ are disjoint iff
$$\Big |\ G.N_{\big\lceil\frac{k-i}{2}\big\rceil}[x_i] \cap G.N_{\big\lceil\frac{k-j}{2}\big\rceil}[x_j]\ \Big | = 0.$$

By \Cref{lemma:b(SP)} (in \Cref{section:burn-general-graphs-optimally}), we have that $b(SP(s, r)) = r+1$ if $s \geq r$; if $s \geq r+2$, then the first fire source $y_1$ of the optimal burning sequence $S^{\prime} = (y_1, y_2, y_3, . . ., y_{k^{\prime}})$ must be placed on its head node $c$.\\

If some algorithm $A$ claims that uptil it is not able to burn $G$, the $f$-fraction neighbourhood of all pairs of fire sources are disjoint, then, the following expression satisfies for all values of $k \in \mathbb{N}$ uptil which $A$ is not able to burn $G$.
$$\big | G.N_{\lceil f \times (k-i) \rceil} [x_i] \cap G.N_{\lceil f \times (k-j)\rceil}[x_j] \big | = 0,$$ $\forall\ 1 \leq i, j \leq k, i \neq j$.\\

\begin{lemma}\label{lemma:half-neighbourhood-limit-SP}
Let $S = (x_1, x_2, . . ., x_k)$ be a burning sequence returned by an approximation algorithm for $SP(d,k^{\prime}-1)$, $d \geq k^{\prime}+1$, the maximum burning neighbourhood fraction that it claim to be pairwise disjoint is the half neighbourhood for all pairs of fire sources uptil $SP(d,k^{\prime}-1)$ is not burnt, that is,  $\forall\ 1 \leq i, j \leq k-1, i \neq j$,\\
$$\Big |\ G.N_{\big\lceil\frac{k-i}{2}\big\rceil}[x_i] \cap G.N_{\big\lceil\frac{k-j}{2}\big\rceil}[x_j]\ \Big | = 0.$$
\end{lemma}

\begin{proof}
Let $S^{\prime} = (y_1, y_2, . . ., y_{k^{\prime}})$ be an optimal burning sequence for $SP(\infty, k^{\prime}-1)$. In a spider graph $SP(d, k^{\prime}-1)$, we know by lemma 12 in \cite{Bessy2017} that the first fire source $y_1$ will always be kept on the head vertex $c$, also observe that\\
$\because k^{\prime} > \big (\big\lceil\frac{k^{\prime}-1}{2}\big\rceil + 1 \big ) + \big (\big\lceil\frac{k^{\prime}-2}{2}\big\rceil + 1 \big )$,\\
$\implies \Big |\ N_{\big\lceil\frac{k^{\prime}-1}{2}\big\rceil}[y_1] \cap N_{\big\lceil\frac{k^{\prime}-2}{2}\big\rceil}[y_2]\ \Big | \geq 1$\\
if $y_2$ is placed on any vertex in $G.V \setminus \{c\}$.

Let that we are only supplied by a burning sequence of $S^{\prime\prime} = (y_1^{\prime}, y_2^{\prime}, . . ., y_{k^{\prime}-1}^{\prime})$ length $k^{\prime}-1$ to burn $SP(d, k^{\prime}-1)$. Let us assume that $c$ is still the fixed place for the first fire source $y_1^{\prime}$. Then,\\
$\because k^{\prime} \geq \big (\big\lceil\frac{k^{\prime}-2}{2}\big\rceil + 1 \big ) + \big (\big\lceil\frac{k^{\prime}-3}{2}\big\rceil + 1 \big )$,\\
$\implies \Big |\ N_{\big\lceil\frac{k^{\prime}-2}{2}\big\rceil}[y_1^{\prime}] \cap N_{\big\lceil\frac{k^{\prime}-3}{2}\big\rceil}[y_2^{\prime}]\ \Big | \geq 0$\\
if $y_2^{\prime}$ is placed on any vertex in $G.V \setminus \{c\}$.

Also observe that if the neighbourhood fraction is increased, then,\\
$\because k^{\prime} > \big (\big\lceil\frac{k^{\prime}-1}{2}\big\rceil + \epsilon + 1 \big ) + \big (\big\lceil\frac{k^{\prime}-2}{2}\big\rceil + \epsilon + 1\big )$,\\
$\implies \Big |\ N_{\big\lceil\frac{k^{\prime}-2}{2} \big\rceil + \epsilon}[y_1^{\prime}] \cap N_{\big\lceil\frac{k^{\prime}-3}{2} \big\rceil + \epsilon}[y_2^{\prime}]\ \Big | \geq 1$\\
if $y_2^{\prime}$ is placed on any vertex in $G.V \setminus \{c\}$ given that $\epsilon$ has a positive value.

This implies that if $S = (x_1, x_2, . . ., x_k)$ is a burning sequence of length $k \geq k^{\prime}$ returned by an arbitrary $z$-approximation algorithm on $SP(d,k^{\prime})$, $z \geq 1+\varepsilon, \varepsilon > 0\ (\implies k^{\prime} \leq k \leq (1 + \varepsilon) \times k^{\prime})$, then the algorithm can make sure that for any pair of distinct fire sources computed by it, $x_i, x_j, i \neq j$, a maximum of half neighbourhood is disjoint when it will be/was restricted to produce a sequence of length $k-1$, that is, for $1\leq i, j \leq k-1, i \neq j$\\
$$\Big |\ N_{\big\lceil\frac{k-1-i}{2} \big\rceil}[x_1] \cap N_{\big\lceil\frac{k-1-j}{2} \big\rceil}[x_2]\ \Big | = 0.$$

This implies the lemma.
\end{proof}

\begin{lemma}\label{lemma:3-approx-example}
    If an approximation algorithm is able to guarantee an $l$-fraction neighbourhood disjointness on all distinct pair of fire sources, $0 < l \leq 1/2$, until it burns a given graph $G$, it can only claim $z \geq 3$ approximation factor for burning general graphs.
\end{lemma}

\begin{proof}
\begin{figure}
	\begin{minipage}{1\textwidth}
		\centering
		\begin{tikzpicture}
		    \node [circle, fill=black, inner sep=0pt, minimum size=3pt, label=left:{$y_{k^{\prime}}$}] at (-.5,0) {};

		    \node [circle, fill=black, inner sep=0pt, minimum size=3pt, label=left:{$y_{k^{\prime}-1}$}] (AA) at (1,0) {};
		    \node [circle, fill=black, inner sep=0pt, minimum size=3pt] (AB) at (1,0.5) {};
		    \node [circle, fill=black, inner sep=0pt, minimum size=3pt] (AC) at (1.3535,0.3535) {};
		    \node [circle, fill=black, inner sep=0pt, minimum size=3pt] (AD) at (1.5,0) {};
		    \node [circle, fill=black, inner sep=0pt, minimum size=3pt] (AE) at (1.3535,-0.3535) {};
		    \node [circle, fill=black, inner sep=0pt, minimum size=3pt] (AF) at (1,-.5) {};
		    \node [circle, fill=black, inner sep=0pt, minimum size=3pt] (AG) at (1-.3535,-0.3535) {};
		    
		    \node [circle, fill=black, inner sep=0pt, minimum size=3pt] at (1-.6,0.2) {};
		    \node [circle, fill=black, inner sep=0pt, minimum size=3pt] at (1-.5,0.4) {};
		    \node [circle, fill=black, inner sep=0pt, minimum size=3pt] at (1-.4,0.6) {};
		    
			\draw (AA) -- (AB);
			\draw (AA) -- (AC);
			\draw (AA) -- (AD);
			\draw (AA) -- (AE);
			\draw (AA) -- (AF);
			\draw (AA) -- (AG);

		    \node [circle, fill=black, inner sep=0pt, minimum size=3pt, label=left:{$y_{k^{\prime}-2}$}] (BA) at (3,0) {};
		    \node [circle, fill=black, inner sep=0pt, minimum size=3pt] (BB) at (3,0.5) {};
		    \node [circle, fill=black, inner sep=0pt, minimum size=3pt] (BC) at (3,1) {};
		    \node [circle, fill=black, inner sep=0pt, minimum size=3pt] (BD) at (3.3535,0.3535) {};
		    \node [circle, fill=black, inner sep=0pt, minimum size=3pt] (BE) at (3.707,0.707) {};
		    \node [circle, fill=black, inner sep=0pt, minimum size=3pt] (BF) at (3.5,0) {};
		    \node [circle, fill=black, inner sep=0pt, minimum size=3pt] (BG) at (3+1,0) {};
		    \node [circle, fill=black, inner sep=0pt, minimum size=3pt] (BH) at (3.3535,-0.3535) {};
		    \node [circle, fill=black, inner sep=0pt, minimum size=3pt] (BI) at (3.707,-0.707) {};
		    \node [circle, fill=black, inner sep=0pt, minimum size=3pt] (BJ) at (3,-0.5) {};
		    \node [circle, fill=black, inner sep=0pt, minimum size=3pt] (BK) at (3,-1) {};
		    \node [circle, fill=black, inner sep=0pt, minimum size=3pt] (BL) at (3-0.3535,-0.3535) {};
		    \node [circle, fill=black, inner sep=0pt, minimum size=3pt] (BM) at (3-0.707,-0.707) {};
		    
		    \node [circle, fill=black, inner sep=0pt, minimum size=3pt] at (3.5-1.1,.5) {};
		    \node [circle, fill=black, inner sep=0pt, minimum size=3pt] at (3.5-1.1,.3) {};
		    \node [circle, fill=black, inner sep=0pt, minimum size=3pt] at (3.5-1.1,.7) {};
		    
			\draw (BA) -- (BC);
			\draw (BA) -- (BE);
			\draw (BA) -- (BG);
			\draw (BA) -- (BI);
			\draw (BA) -- (BK);
			\draw (BA) -- (BM);
			
			\node [circle, fill=black, inner sep=0pt, minimum size=2pt] at (5-.25,-.5) {};
		    \node [circle, fill=black, inner sep=0pt, minimum size=2pt] at (5.2-.25,-.5) {};
		    \node [circle, fill=black, inner sep=0pt, minimum size=2pt] at (5-0.2-.25,-.5) {};

		    \node [circle, fill=black, inner sep=0pt, minimum size=3pt, label=left:{$y_2$}] (CA) at (7-.75,0) {};
		    
		    \node [circle, fill=black, inner sep=0pt, minimum size=3pt] at (7-1.5,0) {};
		    \node [circle, fill=black, inner sep=0pt, minimum size=3pt] at (7-1.5,-0.2) {};
		    \node [circle, fill=black, inner sep=0pt, minimum size=3pt] at (7-1.5,0.2) {};
		    
			\draw (CA) -- (7-.75,1);
			\draw (CA) -- (7.707-.75,0.707);
			\draw (CA) -- (7+1-.75,0);
			\draw (CA) -- (7.707-.75,-0.707);
			\draw (CA) -- (7-.75,-1);
			\draw (CA) -- (7-0.707-.75,-0.707);
			
			\node [circle, fill=black, inner sep=0pt, minimum size=3pt, label=left:{$y_1$}] (DA) at (8.75,0) {};
		    
		    \node [circle, fill=black, inner sep=0pt, minimum size=3pt] at (9-1,0) {};
		    \node [circle, fill=black, inner sep=0pt, minimum size=3pt] at (9-1,-0.2) {};
		    \node [circle, fill=black, inner sep=0pt, minimum size=3pt] at (9-1,0.2) {};
		    
			\draw (DA) -- (8.75,1);
			\draw (DA) -- (8.75+.707,0.707);
			\draw (DA) -- (8.75+1,0);
			\draw (DA) -- (8.75+.707,-0.707);
			\draw (DA) -- (8.75,-1);
			\draw (DA) -- (8.75-0.707,-0.707);
			
			\node at (-.7,-.5) {$SP(d_1 ,0)$};
			\node at (1,-1.25) {$SP(d_2 ,1)$};
			\node at (3,-2) {$SP(d_3 ,2)$};
			\node at (6,-2) {$SP(d_{k^{\prime}-1} ,k^{\prime}-2)$};
			\node at (8.75,-1.75) {$SP(d_{k^{\prime}},k^{\prime}-1)$};
		\end{tikzpicture}
	\end{minipage}
	
	\caption{Burning of a graph with $k^{\prime}$ components $comp_i$. Each $comp_i$ having one central head vertex, has $d_i \geq i+1$ arms of equal length $i-i$. Burning number of this graph is $k^{\prime}$.}
\end{figure}
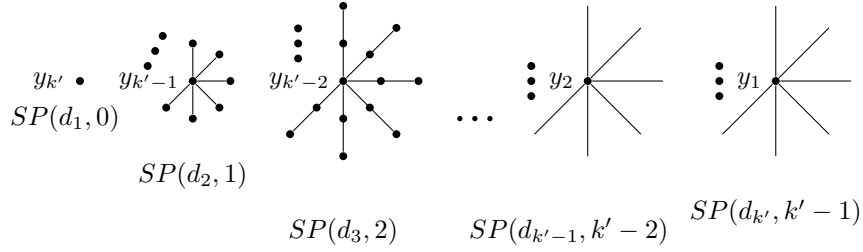

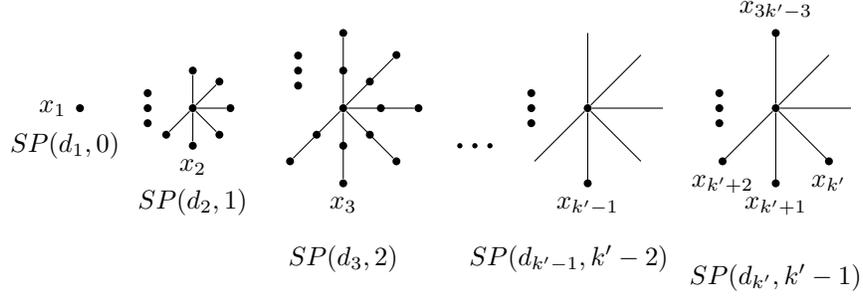
\begin{figure}
	\begin{minipage}{1\textwidth}
		\centering
		\begin{tikzpicture}
		    \node [circle, fill=black, inner sep=0pt, minimum size=3pt, label=left:{$x_1$}] at (-.5,0) {};

		    \node [circle, fill=black, inner sep=0pt, minimum size=3pt] (AA) at (1,0) {};
		    \node [circle, fill=black, inner sep=0pt, minimum size=3pt] (AB) at (1,0.5) {};
		    \node [circle, fill=black, inner sep=0pt, minimum size=3pt] (AC) at (1.3535,0.3535) {};
		    \node [circle, fill=black, inner sep=0pt, minimum size=3pt] (AD) at (1.5,0) {};
		    \node [circle, fill=black, inner sep=0pt, minimum size=3pt] (AE) at (1.3535,-0.3535) {};
		    \node [circle, fill=black, inner sep=0pt, minimum size=3pt, label=below:{$x_2$}] (AF) at (1,-.5) {};
		    \node [circle, fill=black, inner sep=0pt, minimum size=3pt] (AG) at (1-.3535,-0.3535) {};
		    
		    \node [circle, fill=black, inner sep=0pt, minimum size=3pt] at (1-.6,0) {};
		    \node [circle, fill=black, inner sep=0pt, minimum size=3pt] at (1-.6,-0.2) {};
		    \node [circle, fill=black, inner sep=0pt, minimum size=3pt] at (1-.6,0.2) {};
		    
			\draw (AA) -- (AB);
			\draw (AA) -- (AC);
			\draw (AA) -- (AD);
			\draw (AA) -- (AE);
			\draw (AA) -- (AF);
			\draw (AA) -- (AG);

		    \node [circle, fill=black, inner sep=0pt, minimum size=3pt] (BA) at (3,0) {};
		    \node [circle, fill=black, inner sep=0pt, minimum size=3pt] (BB) at (3,0.5) {};
		    \node [circle, fill=black, inner sep=0pt, minimum size=3pt] (BC) at (3,1) {};
		    \node [circle, fill=black, inner sep=0pt, minimum size=3pt] (BD) at (3.3535,0.3535) {};
		    \node [circle, fill=black, inner sep=0pt, minimum size=3pt] (BE) at (3.707,0.707) {};
		    \node [circle, fill=black, inner sep=0pt, minimum size=3pt] (BF) at (3.5,0) {};
		    \node [circle, fill=black, inner sep=0pt, minimum size=3pt] (BG) at (3+1,0) {};
		    \node [circle, fill=black, inner sep=0pt, minimum size=3pt] (BH) at (3.3535,-0.3535) {};
		    \node [circle, fill=black, inner sep=0pt, minimum size=3pt] (BI) at (3.707,-0.707) {};
		    \node [circle, fill=black, inner sep=0pt, minimum size=3pt] (BJ) at (3,-0.5) {};
		    \node [circle, fill=black, inner sep=0pt, minimum size=3pt, label=below:{$x_3$}] (BK) at (3,-1) {};
		    \node [circle, fill=black, inner sep=0pt, minimum size=3pt] (BL) at (3-0.3535,-0.3535) {};
		    \node [circle, fill=black, inner sep=0pt, minimum size=3pt] (BM) at (3-0.707,-0.707) {};
		    
		    \node [circle, fill=black, inner sep=0pt, minimum size=3pt] at (3.5-1.1,.5) {};
		    \node [circle, fill=black, inner sep=0pt, minimum size=3pt] at (3.5-1.1,.3) {};
		    \node [circle, fill=black, inner sep=0pt, minimum size=3pt] at (3.5-1.1,.7) {};
		    
			\draw (BA) -- (BC);
			\draw (BA) -- (BE);
			\draw (BA) -- (BG);
			\draw (BA) -- (BI);
			\draw (BA) -- (BK);
			\draw (BA) -- (BM);
			
			\node [circle, fill=black, inner sep=0pt, minimum size=2pt] at (5-.25,-.5) {};
		    \node [circle, fill=black, inner sep=0pt, minimum size=2pt] at (5.2-.25,-.5) {};
		    \node [circle, fill=black, inner sep=0pt, minimum size=2pt] at (5-0.2-.25,-.5) {};

		    \node [circle, fill=black, inner sep=0pt, minimum size=3pt] (CA) at (7-.75,0) {};
		    
		    \node [circle, fill=black, inner sep=0pt, minimum size=3pt] at (7-1.5,0) {};
		    \node [circle, fill=black, inner sep=0pt, minimum size=3pt] at (7-1.5,-0.2) {};
		    \node [circle, fill=black, inner sep=0pt, minimum size=3pt] at (7-1.5,0.2) {};
		    
		    \node [circle, fill=black, inner sep=0pt, minimum size=3pt, label=below:{$x_{k^{\prime}-1}$}] at (7-.75,-1) {};
		    
			\draw (CA) -- (7-.75,1);
			\draw (CA) -- (7.707-.75,0.707);
			\draw (CA) -- (7+1-.75,0);
			\draw (CA) -- (7.707-.75,-0.707);
			\draw (CA) -- (7-.75,-1);
			\draw (CA) -- (7-0.707-.75,-0.707);
			
			\node [circle, fill=black, inner sep=0pt, minimum size=3pt] (DA) at (8.75,0) {};
		    
		    \node [circle, fill=black, inner sep=0pt, minimum size=3pt] at (9-1,0) {};
		    \node [circle, fill=black, inner sep=0pt, minimum size=3pt] at (9-1,-0.2) {};
		    \node [circle, fill=black, inner sep=0pt, minimum size=3pt] at (9-1,0.2) {};
		    
		    \node [circle, fill=black, inner sep=0pt, minimum size=3pt, label=below:{$x_{k^{\prime}}$}] (BK) at (8.75+.707,-.707) {};
		    \node [circle, fill=black, inner sep=0pt, minimum size=3pt, label=below:{$x_{k^{\prime}+1}$}] (BK) at (8.75,-1) {};
		    \node [circle, fill=black, inner sep=0pt, minimum size=3pt, label=below:{$x_{k^{\prime}+2}$}] (BK) at (8.75-0.707,-0.707) {};
		    \node [circle, fill=black, inner sep=0pt, minimum size=3pt, label=above:{$x_{3 k^{\prime}-3}$}] (BK) at (8.75,1) {};
		    
			\draw (DA) -- (8.75,1);
			\draw (DA) -- (8.75+.707,0.707);
			\draw (DA) -- (8.75+1,0);
			\draw (DA) -- (8.75+.707,-0.707);
			\draw (DA) -- (8.75,-1);
			\draw (DA) -- (8.75-0.707,-0.707);
			
			\node at (-.7,-.5) {$SP(d_1 ,0)$};
			\node at (1,-1.25) {$SP(d_2 ,1)$};
			\node at (3,-2) {$SP(d_3 ,2)$};
			\node at (6,-2) {$SP(d_{k^{\prime}-1} ,k^{\prime}-2)$};
			\node at (8.75,-2.25) {$SP(d_{k^{\prime}} ,k^{\prime}-1)$};
		\end{tikzpicture}
	\end{minipage}
	
	\caption{Burning of a graph with $k^{\prime}$ components $comp_i$. Each $comp_i$ having one central head vertex, has $d_i \geq i+1$ arms of equal length $i-i$. A burning sequence of length $3 k^{\prime}-3$ is shown to not to be able to burn $SP(X)$ completely.}
\end{figure}

As shown in figure 2, $SP(X)$ is a disjoint union of spider graphs $SP(d_1, 0)$, $SP(d_2, 1)$, $SP(d_3, 2)$, $. . .$, $SP(d_{k^{\prime}-1}, k^{\prime}-2)$, and $SP(d_{k^{\prime}}, k^{\prime}-1)$; $d_i \geq i+1$.

The construction of a burning sequence $S$ (not optimal) for $SP(X)$ for a fixed $k^{\prime}$ is as follows. $\forall\ i, 1 \leq i \leq k^{\prime}-2$, a fire source $x_i$ is placed on one of the leaf nodes of $SP(d_i, i-1)$. $\forall\ k^{\prime} \leq i \leq 3k^{\prime}-3$ a fire source $x_i$ is placed on a distinct leaf node of $SP(d_{k^{\prime}}, k^{\prime}-1)$ where no other fire source is already placed. Figure 1 demonstrates the construction of this burning sequence $S$ for $SP(X)$, which may be returned by an arbitrary approximation algorithm $A$.

For $1 \leq k \leq 3 k^{\prime}-3$, if a burning sequence $S=(x_1, x_2, x_3, . . ., x_k)$ is used to burn the graph and the fire sources $\{x_i\}$ are placed according to the preceding procedure, then we can observe that $S$ is not able to burn $SP(X)$, and for all the fire sources, their $0 < l \leq 1/2$ factor neighbourhood is disjoint with each other. Also, it is trivially observable that if the value of $l$ decreases, then the number of required fire source will increase accordingly.
\end{proof}

Let $G$ be an arbitrary graph, and $S = (x_1, x_2, . . ., x_k)$ be a burning sequence of length $k$ returned by an arbitrary $z$-approximation algorithm $A$, $z \geq 1+\varepsilon, \varepsilon > 0$. Let $b(G) = k^{\prime}$.

From the proofs of \Cref{lemma:half-neighbourhood-limit-SP} and \Cref{lemma:3-approx-example}, we have that if $A$ will be/was restricted to produce a burning sequence of length $l \leq k-1$ for an arbitrary $G$ if $k \geq k^{\prime}$, it can be claimed that $\forall\ x_i, x_j, 1 \leq i, j \leq k-1, i \neq j,$
$$\Big |\ N_{\big\lceil\frac{l-i}{2} \big\rceil + \epsilon}[x_1] \cap N_{\big\lceil\frac{l-j}{2} \big\rceil + \epsilon}[x_2]\ \Big | = 0$$\\
only if $\epsilon \leq 0$. It means that for an approximation algorithm for general graph burning, the upper bound of $\epsilon$ is 0.

This is the maximum that $A$ can claim for a general graph. Let that $A$ is able to guarantee this. It means that if $G$ is not burnt completely and $\epsilon \leq 0$, then $A$ guarantees that\\
$N_{\lceil\frac{k-1}{2}\rceil}[x_1], N_{\lceil\frac{k-2}{2}\rceil}[x_2], N_{\lceil\frac{k-3}{2}\rceil}[x_3], . . ., N_0[x_k]$ are pairwise disjoint, and\\
$\because \Bigg\lceil \frac{k-\big\lceil\frac{k}{3}\big\rceil-1}{2} \Bigg\rceil \geq \big\lfloor \frac{k}{3} \big\rfloor-1$,\\
$\implies b(G) \leq \big\lfloor\frac{k}{3}\big\rfloor+1$ (by \Cref{lemma:burn-seq-geq3}).\\
$\implies z \geq 3 \implies \varepsilon = 2$ (by \Cref{theorem:proof-3-approx}).

These properties allow us to suggest that graph burning may be hard to approximate better than the $3$-approximation ratio, which we formally state as \Cref{conjecture:no-better-than-3-approx} as follows.

\begin{conjecture}\label{conjecture:no-better-than-3-approx}
    A maximum of 3-approximation is possible in polynomial time to compute burning number of general graphs.\qed
\end{conjecture}

If it is a necessary attribute for an approximation algorithm for burning general graphs to claim a neighbourhood factor disjointness between each pair of fire sources in the burning sequence $S = (x_1, x_2, . . ., x_k)$ returned by it, then a maximum of 3-approximation is possible in polynomial time to compute burning number of general graphs.

\section{Approximating the burning of connected interval graphs}

We discussed in \Cref{subsection:IG-path-similar} that if $P$ is the diameter of an interval graph $G$, then all the other vertices in $G$, that is, the vertices in $G\setminus P$, are connected to at least one vertex in $P$ by a single edge.
We discussed in \Cref{section:burn-path} that we can burn a path optimally in polynomial time using \Cref{algorithm:burn-path-finite}. Following from this we discussed in \Cref{subsection:similar-burn-path-IG} that the burning number of an interval graph $G$ $b(P)\leq b(G)\leq b(P)+1$. We showed in \Cref{section:burn-interval-graphs} that determining whether $b(G)=b(P)$ is NP-Complete. We describe \Cref{algorithm:approximate-burn-IG} as follows, which is able to burn an interval graph $G$ within $b(G)+1$. We have used the \textsc{Algebraic-Floyd-Warshall} \cite{Kepner2011} to compute $P$, the diameter of $G$, which is a shortest path of maximum length in $G$.

\begin{algorithm}\label{algorithm:approximate-burn-IG}
    Given an interval graph $G$, perform the following steps.
\end{algorithm}

\textbf{\textit{Stage 1.}} $F=$ a set of shortest paths between all pairs of vertices from the \textsc{Algebraic-Floyd-Warshall} algorithm. $P=$ a path of maximum length in $F$.

\textbf{\textit{Stage 2.}} Invoke \Cref{algorithm:burn-path-finite} by passing $(G, P)$ as input, store the return value in $S=(x_1,x_2,\dots,x_k)$.

\textbf{\textit{Stage 3.}} If $S$ is not able to burn $G$ completely, that is, if $S$ does not satisfy \Cref{equation:burn-verify}, then $v=$ an arbitrary vertex in $G.V\setminus (G.N_{k-1}[x_1]\cup G.N_{k-2}[x_2]\cup\dots\cup G.N_0[x_k])$. $S=S\cup_{s/}v$.

\textbf{\textit{Stage 4.}} Return $S$.\\

Time complexity of \Cref{algorithm:approximate-burn-IG} is $O(n)$, where $n$ is the number of vertices in $G$. $S$ is the burning sequence returned by \Cref{algorithm:approximate-burn-IG}.

In \Cref{section:inapproximability-NPH}, we discussed in \Cref{theorem:approximability-general} that no approximation algorithm $A$ can guarantee an approximation ratio $R_A$ of less than $k+(1/k)$ if $k$ is the cost of the optimal solution of a problem $P$ on some input $x$. If the burning number of a connected interval graph $G$ be $k=b(G)$, \Cref{algorithm:approximate-burn-IG} is able to approximate burning of interval graphs within $k+1$. Thus we get that the approximation ratio for burning interval graphs is $R_A=k+(1/k)$.

\section{Summary of this chapter}

We describe a 3-approximation algorithm for burning general graphs, as described in \cite{Bessy2017}. Based on some properties of graph burning, we provide a conjecture that graph burning cannot be approximated to a factor less than $3$. We also discuss the approximation on interval graphs, deriving from the properties presented in \cite{Kare2019}.

\chapter{Conclusion and future work}\label{chapter:conclusion}

\section{Observations}

\subsection{On graph burning}

Discovery of the graph burning and the burning number property was motivated by the works that were trying to model the spread of an object. This object multiplies and spreads to nodes through connections.

This object is, so far, related to spread of electronic information in a network, for example, spread of a meme, a gossip, a social contagion, influence or emotion, message, or alarms through a social network. The burning number describes the minimum number of (discrete) time-steps that are able to ``infect'' the whole network with an object.

\subsection{Burning versus other problems}

Optimal graph burning is proved to be computationally hard on various graph classes. It is clear with reference to interval graphs, permutation graphs, disk graphs, spider graphs, trees, path forests and several other graph classes that optimal burning of such graph classes is NP-Complete where other problems, which are NP-Complete for general graphs, can be solved in polynomial time. While proving NP-Completeness, we have extensively utilized the distinct 3-partition problem, by reducing it to the burning of several graph classes.

Although, burning is easy on some other graph classes such as split graphs and cographs.

\subsection{Approximating the burning sequence}

The 3-approximation algorithm for deriving a burning sequence for an arbitrary graph is based on the pairwise disjointness of the half neighbourhoods of the fire sources.

Approximating the burning sequence for general graphs better than the factor of $3$ may be computationally hard because pairwise disjointness on more than half neighbourhood of the fire sources is not possible on general graphs.

\section{What next?}

\subsection{Discrete mathematics and theoretical algorithm design}

So far a few graphs have been shown to be NP-Complete from the burning perspective, a few have been shown to be solvable in polynomial time.
Still several of the graph classes are left to research on, which are useful in various practical scenarios, and determine whether deriving an optimal burning sequence is solvable polynomially.

We showed that we can approximate burning of interval graphs, which we have shown to be NP-Complete to solve optimally, within an approximation ratio of $1+(1/k)$, but not less than this factor.
This observation is close enough to the limit imposed by the theorem that if P $\neq$ NP, then $R_A\geq 1+(1/k)$ in all cases. Although, the conjecture that P $\neq$ NP remains unaffected, and we have not proved or disproved it.

\subsection{Practical implementations}

Graph burning has been proposed to model several real-time systems which are complex otherwise.

The graph burning can also be used in several practical scenarios other than spreading a message, alarm, or contagion. For example, the spread of an infection, virus, etc can also be modelled using this newly discovered graph procedure. We can model the spread of a ``real-life'' contagion using graph burning, such as person-to-person spread of a communicable infection.

The multiplication of the virus inside a host and its behavior of infecting the network of target cells may also be modelled very precisely same as how the fire spreads throughout a graph. And then the firefighter problem can be used to simulate the defence of the host body. Similarly, firefighter can also be used in the simulation of the ``real-life'' social communication of a disease, where it can be used to model \textit{how do} or \textit{how good can} we plan for the defence mechanism, given the constraints such as availability of vaccines, etc by using the minimum resources to save the maximum population.

It might not be very simple and also, may not be exactly same as what we have modelled so far theoretically on static graphs; weights may be involved, time instances may not be discrete and at equal intervals as we model theoretically. Probabilities with respect to the spread or rescue may also be involved, along with deaths (with probabilities involved). In such cases, elimination of cells from the body and people from a real-life social network, which can be modelled by elimination of vertices in the graph being modelled. So, this model that we propose as one of the probable future works can also be extended to probabilistic or temporal graphs which may give us better, more precise results based on the ``real-life'' circumstances.


\addcontentsline{toc}{chapter}{Bibliography}
\bibliography{ref.bib}
\bibliographystyle{plain}
\end{document}